\documentclass[a4paper,onecolumn,superscriptaddress,11pt,accepted=2021-03-11]{quantumarticle}
\pdfoutput=1
\usepackage[utf8]{inputenc}
\usepackage[english]{babel}
\usepackage[T1]{fontenc}
\usepackage{amsmath}
\usepackage{hyperref}
\usepackage[numbers,sort&compress]{natbib}
\usepackage{stmaryrd} 
\usepackage{mathtools}

\usepackage{tikzit}
\input{zx.tikzdefs}

\tikzstyle{box}=[shape=rectangle, text height=1.5ex, text depth=0.25ex, yshift=0.5mm, fill=white, draw=black, minimum height=5mm, yshift=-0.5mm, minimum width=5mm, font={\small}]
\tikzstyle{Z dot}=[inner sep=0mm, minimum size=2mm, shape=circle, draw=black, fill={rgb,255: red,221; green,255; blue,221}]
\tikzstyle{Z phase dot}=[minimum size=1.2em, font={\footnotesize\boldmath}, shape=rectangle, rounded corners=0.5em, inner sep=0.2em, outer sep=-0.2em, scale=0.8, tikzit shape=circle, draw=black, fill={rgb,255: red,221; green,255; blue,221}, tikzit draw=blue]
\tikzstyle{X dot}=[Z dot, shape=circle, draw=black, fill={rgb,255: red,255; green,136; blue,136}]
\tikzstyle{X phase dot}=[Z phase dot, tikzit shape=circle, tikzit draw=blue, fill={rgb,255: red,255; green,136; blue,136}, font={\footnotesize\boldmath}]
\tikzstyle{hadamard}=[fill=yellow, draw=black, shape=rectangle, inner sep=0.6mm, minimum height=1.5mm, minimum width=1.5mm]
\tikzstyle{vertex}=[inner sep=0mm, minimum size=1mm, shape=circle, draw=black, fill=black]
\tikzstyle{vertex set}=[inner sep=0mm, minimum size=1mm, shape=circle, draw=black, fill=white, font={\footnotesize\boldmath}]
\tikzstyle{target}=[inner sep=0mm, minimum size=3mm, shape=circle, draw=black]

\tikzstyle{hadamard edge}=[-, dashed, dash pattern=on 2pt off 1.5pt, thick, draw={rgb,255: red,68; green,136; blue,255}]
\tikzstyle{brace edge}=[-, tikzit draw=blue, decorate, decoration={brace,amplitude=1mm,raise=-1mm}]
\tikzstyle{diredge}=[->]
\tikzstyle{dashed edge}=[-, dashed, dash pattern=on 2pt off 0.5pt, draw=black]

\usepackage[shortlabels]{enumitem}

\usepackage{xspace}
\usepackage{algorithm}
\usepackage{algorithmicx}
\usepackage[noend]{algpseudocode}
\usepackage{verbatim}

\newcommand{\symd}{\mathbin{\Delta}\xspace}

\newcommand{\set}[1]{\{#1\}}

\makeatletter
\newcommand\etc{etc\@ifnextchar.{}{.\@}\xspace}
\newcommand\ie{i.e.\@\xspace}  

\makeatother

\newcommand{\odd}[2]{\textsf{Odd}_{#1}\left(#2\right)}
\newcommand{\eve}[2]{\textsf{Even}_{#1}\left(#2\right)}

\newcommand{\pow}[1]{\ensuremath{\mathcal{P}\left( #1 \right)}}
\newcommand{\pat}{\ensuremath{\mathfrak{P}}} 
\newcommand{\gflow}{\ensuremath{\mathfrak{g}}} 
\newcommand{\intf}[1]{\left\llbracket #1 \right\rrbracket} 

\newcommand{\ld}{\lambda}
\newcommand{\sse}{\subseteq}

\newcommand{\CZ}{\ensuremath{\textrm{CZ}}\xspace}
\newcommand{\CX}{\ensuremath{\textrm{CNOT}}\xspace}

\newcommand{\CNOT}{\CX}

\newcommand{\XY}{\normalfont XY\xspace}
\newcommand{\XZ}{\normalfont XZ\xspace}
\newcommand{\YZ}{\normalfont YZ\xspace}
\newcommand{\XYm}{\ensuremath\normalfont\textrm{XY}\xspace}
\newcommand{\XZm}{\normalfont\normalfont\textrm{XZ}\xspace}
\newcommand{\YZm}{\normalfont\normalfont\textrm{YZ}\xspace}

\newcommand{\LOG}{labelled open graph}

\usepackage{bm}

\newcommand{\bra}[1]{\ensuremath{\left\langle #1 \right|}}
\newcommand{\ket}[1]{\ensuremath{\left|  #1 \right\rangle}}

\newcommand{\ketbra}[2]{\ensuremath{\ket{#1}\!\bra{#2}}}
\newcommand{\abs}[1]{\ensuremath{\left| #1 \right|}}
\usepackage{amsmath,amsthm,amssymb}
\theoremstyle{definition}
\newtheorem{theorem}{Theorem}[section]
\newtheorem{corollary}[theorem]{Corollary}
\newtheorem{lemma}[theorem]{Lemma}
\newtheorem*{lemma*}{Lemma}
\newtheorem*{proposition*}{Proposition}
\newtheorem{proposition}[theorem]{Proposition}

\newtheorem{definition}[theorem]{Definition}
\newtheorem{example}[theorem]{Example}
\newtheorem{remark}[theorem]{Remark}

\newcommand{\Gp}{\ensuremath G\star u}
\newcommand{\Gpp}{\ensuremath G\star u\star v}
\newcommand{\idrem}[1]{\ensuremath\underset{#1}{\twoheadrightarrow}} 

\newcommand{\precprimed}{{\prec'}}
\newcommand{\tildeprecprimed}{{\tilde \prec'}}

\def\al{\alpha}

\def\cz#1#2{[Z_{#1}]^{#2}}
\def\cx#1#2{X_{#1}^{#2}}

\def\cx#1#2{[X_{#1}]^{#2}}
\def\phZ#1#2{N_{#1}^{#2}}

\newtheorem*{rep@theorem}{\rep@title}
\newcommand{\newreptheorem}[2]{%
	\newenvironment{rep#1}[1]{%
    \def\rep@title{#2 \ref{##1} (restated)}%
		\begin{rep@theorem}}%
		{\end{rep@theorem}}}
\newreptheorem{thm}{Theorem}
\newreptheorem{lem}{Lemma}
\newreptheorem{prop}{Proposition}
\newreptheorem{cor}{Corollary}

\title{There and back again: A circuit extraction tale}

\author{Miriam Backens}
\affiliation{University of Birmingham}
\email{m.backens@cs.bham.ac.uk}
\orcid{0000-0002-5418-1084}

\author{Hector Miller-Bakewell}
\affiliation{University of Oxford}
\email{hector.miller-bakewell@cs.ox.ac.uk}
\orcid{0000-0002-2741-6714}

\author{Giovanni de Felice}
\affiliation{University of Oxford}
\email{giovanni.defelice@cs.ox.ac.uk}
\orcid{0000-0002-8550-718X}

\author{Leo Lobski}
\affiliation{ILLC, University of Amsterdam (until 2020)}
\email{contact@leolob.ski}

\author{John van de Wetering}
\affiliation{Radboud University Nijmegen}
\email{john@vdwetering.name}
\orcid{0000-0002-5405-8959}

\begin{document}

\maketitle

\begin{abstract}
Translations between the quantum circuit model and the measurement-based one-way model are useful for verification and optimisation of quantum computations.
They make crucial use of a property known as gflow.
While gflow is defined for one-way computations allowing measurements in three different planes of the Bloch sphere, most research so far has focused on computations containing only measurements in the \XY-plane.
Here, we give the first circuit-extraction algorithm to work for one-way computations containing measurements in all three planes and having gflow.
The algorithm is efficient and the resulting circuits do not contain ancillae.
One-way computations are represented using the \zxcalculus, hence the algorithm also represents the most general known procedure for extracting circuits from \zxdiagrams.
In developing this algorithm, we generalise several concepts and results previously known for computations containing only \XY-plane measurements.
We bring together several known rewrite rules for measurement patterns and formalise them in a unified notation using the \zxcalculus.
These rules are used to simplify measurement patterns by reducing the number of qubits while preserving both the semantics and the existence of gflow.
The results can be applied to circuit optimisation by translating circuits to patterns and back again.
\end{abstract}


\tableofcontents

\section{Introduction}

The gate-based model and the measurement-based model are two fundamentally different approaches to implementing quantum computations.
In the gate-based model \cite{NielsenChuang}, the bulk of the computation is performed via unitary (i.e.\ reversible) one- and two-qubit gates.
Measurements serve mainly to read out data and may be postponed to the end of the computation.
Conversely, in the measurement-based model, the bulk of the computation is performed via measurements on some general resource state, which is independent of the specific computation.
We focus here on the one-way model~\cite{MBQC1}, where the resource states are \emph{graph states} (see Section~\ref{sec:MBQC}).
In this paper, we study ways of converting between these two different approaches to quantum computation with a view towards optimising the implementation of both.

While computations in the gate-based model are represented as quantum circuits,
computations in the one-way model are usually represented by \emph{measurement patterns}, which describe both the graph state and the measurements performed on it~\cite{Patterns,danos_kashefi_panangaden_perdrix_2009}.
Measurement patterns in the one-way model generally do not allow arbitrary single-qubit measurements.
Instead, measurements are often restricted to the `planes' of the Bloch sphere that are orthogonal to the principal axes, labelled the \XY-, \XZ-, and \YZ-planes.
In fact, most research has focused on measurements in just the \XY-plane, which alone are sufficient for universal quantum computation~\cite{Patterns}.
Similarly, measurements in the \XZ-plane are also universal~\cite{mhalla2012graph}, although this model has been explored less in the literature.
In this paper, we will consider measurements in all three of the planes, since this usually leads to patterns involving fewer qubits, and allows for more non-trivial transformations of the graph state.

Due to the non-deterministic nature of quantum measurements, a one-way computation needs to be adaptive, with later measurement angles depending on the outcomes of earlier measurements~\cite{MBQC1}.
While the ability to correct undesired measurement outcomes is necessary for obtaining a deterministic computation, not all sequences of measurements support such corrections.
The ability to perform a deterministic computation depends on the underlying graph state and the choice of measurement planes, which together form an object called a \LOG.
If all measurements are in the \XY-plane, \emph{causal flow} (sometimes simply called `flow') is a sufficient condition for the \LOG\ to support deterministically implementable patterns~\cite{Danos2006Determinism-in-}.
Yet causal flow is not necessary for determinism.
Instead, the condition of \emph{generalized flow} (or \emph{gflow}) \cite{GFlow} is both sufficient and necessary for deterministic implementability.
Gflow can be defined for \LOG{}s containing measurements in all three planes, in which case it is sometimes called \emph{extended gflow}~\cite[Theorems~2 \&~3]{GFlow}.

A given representation of some computation can be transformed into a different representation of the same computation using local rewriting.
The new representation may be chosen to have more desirable properties.
For quantum circuits, such desirable properties include low depth~\cite{amy2013meet}, small total gate count~\cite{CliffOpt}, or small counts of some particular type of gate, such as the T-gate~\cite{amy2014polynomial}.
For measurement patterns, desirable properties include a small number of qubits~\cite{eslamy2018optimization,houshmand2018minimal} or a particularly simple underlying graph state~\cite{mhalla2012graph}.

Local processes can also be used to translate a pattern into a circuit.
This is used, for example, to verify that the pattern represents the desired operation~\cite{Danos2006Determinism-in-,beaudrap2010unitary,duncan2010rewriting,daSilva2013compact,miyazaki2015analysis}.
Conversely, a translation of a circuit into a pattern can be used to implement known algorithms in the one-way model, or it can be combined with a translation back to a circuit to trade depth against width, to parallelise Clifford operations, or to reduce the number of T gates~\cite{broadbent_2009_parallelizing,daSilva2013global,houshmand2017quantum}.
No complete set of rewrite rules is known for quantum circuits or for measurement patterns, although a completeness result does exist for 2-qubit circuits over the Clifford+T gate set \cite{Bian2Qubit}.

Rewriting of patterns or circuits, as well as translations between the two models, can be performed using the \zxcalculus, a graphical language for quantum computation \cite{CD2}.
This language is  more flexible than quantum circuit notation and also has multiple complete sets of graphical rewrite rules \cite{SimonCompleteness,HarnyAmarCompleteness,JPV-universal,euler-zx}.
While translating a measurement pattern to a quantum circuit can be difficult, the translation between patterns and \zxdiagrams is straightforward \cite{duncan2010rewriting,cliff-simp,kissinger2017MBQC}.

\subsection{Our contributions}

In this paper, we give an algorithm that extracts a quantum circuit from any measurement pattern whose underlying \LOG\ has extended gflow.
Our algorithm does not use ancillae.
While circuit extraction algorithms for patterns containing only XY-plane measurements were already known~\cite{broadbent_2009_parallelizing,duncan2010rewriting,miyazaki2015analysis,cliff-simp},
this is the first circuit extraction algorithm for extended gflow, i.e.\ where patterns may contain measurements in more than one plane.
The algorithm works by translating the pattern into the \zxcalculus and transforming the resulting \zxdiagram into a circuit-like form.
It generalises a similar algorithm, which works only for patterns where all measurements are in the \XY-plane \cite{cliff-simp}.
The circuit extraction algorithm employs the \zxcalculus, so it can be used not only on diagrams arising from measurement patterns but on any \zxdiagram satisfying certain properties.
Thus, this procedure is not only the most general circuit extraction algorithm for measurement patterns but also the most general known circuit extraction algorithm for ZX-diagrams.
In particular, our algorithm can handle ZX-diagrams that contain \emph{phase gadgets}~\cite{tcountpreprint}, which have been used recently in a variety of circuit optimisation schemes~\cite{phaseGadgetSynth,pi4parity,cowtan2020generic}.

In developing the circuit extraction algorithm, we derive a number of explicit rewrite rules for \zxdiagrams representing measurement patterns, in particular rewrites that involve graph transformations on the underlying resource state.
We show how the gflow changes for each of these rewrite rules, i.e.\ how the rewrites affect the instructions for correcting undesired measurement outcomes,
and hence that these transformations preserve the ability to deterministically implement the pattern.
Several of these rules were previously employed in the literature, e.g.\ the pivot-minor transformation in Ref.~\cite{mhalla2012graph} or the elimination of Clifford measurements first derived in a different context in Ref.~\cite{hein2004multiparty},
but we are the first to explicitly work out the general transformations of the gflow resulting from these transformations of the resource state and the pattern.

The rewrite rules serve not only to prove the correctness of the algorithm, but also to simplify the measurement patterns by reducing the number of qubits involved.
It was already known that qubits measured in a Pauli basis can be treated in special ways~\cite{GFlow}, for instance that these measurements can be performed simultaneously without losing determinism~\cite{di2016parallelizing},
and hence that the depth of the pattern only depends on the number of non-Clifford operations.
Additionally, based on~\cite{hein2004multiparty}, Houshmand et al.~\cite{houshmand2018minimal} stated that Pauli measurements can be removed from a pattern as long as the rest of the pattern is transformed appropriately. 
Note that their work is restricted to patterns with \XY-plane measurements only, and that they worked out the effects of the transformation on the gflow for only a few example cases.
We extend these results by combining our different rewrite rules to show exactly how the gflow is transformed when a qubit measured in a Pauli basis is removed from a pattern.
In particular, we explicitly prove that the existence of gflow is preserved, so that deterministic implementability is maintained.
So whereas~\cite{di2016parallelizing} showed that the depth of a computation depends directly on the number of non-Clifford operations,
our results show that the number of measured qubits needed to perform a measurement-based computation is directly related to the number of non-Clifford operations required for the computation, as every non-input qubit measured in a Pauli basis can be removed from a pattern without losing determinism.

These results require us to generalise several concepts originally developed for patterns containing only \XY-plane measurements to patterns with measurements in multiple planes.
In particular, we adapt the definitions of \emph{focused gflow}~\cite{mhalla2011graph} and \emph{maximally delayed gflow}~\cite{MP08-icalp} to the extended gflow case.
Our generalisation of focused gflow differs from the three generalisations suggested by Hamrit and Perdrix~\cite{hamrit2015reversibility}; indeed, the desired applications naturally lead us to one unique generalisation (see Section~\ref{sec:focusing-extended-gflow}).

Combined with the known procedure for transforming a quantum circuit into a measurement pattern using the \zxcalculus~\cite{cliff-simp}, our pattern simplification and circuit extraction procedure can be used to reduce the T-gate count of quantum circuits~\cite{tcountpreprint}.
This involves translating the circuit into a ZX-diagram, transforming this to a diagram which corresponds to a measurement pattern, simplifying the pattern, and then re-extracting a circuit.
A rough overview of the different translation and optimisation procedures is given in Figure~\ref{figRoughTranslationsAndOptimisations}.

\begin{figure}
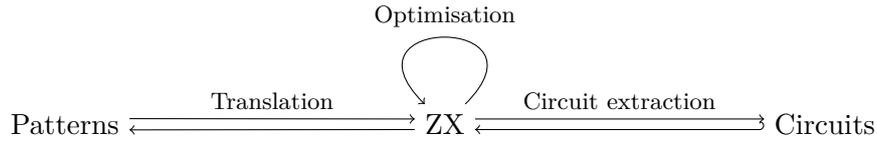

\ctikzfig{translationsOverview}
\caption{A rough overview over the translation procedures between the three paradigms, indicating where the optimisation steps are carried out.
\label{figRoughTranslationsAndOptimisations}}
\end{figure}

The remainder of this paper is structured as follows.
Known definitions and results relating to \zxcalculus, measurement patterns and gflow are presented in Section~\ref{sec:preliminaries}.
We derive the rewrite rules for extended measurement patterns and the corresponding gflow transformations in Section~\ref{sec:rewriting}. These results are used in Section~\ref{sec:simplifying} to simplify measurement patterns in various ways. Then in Section~\ref{sec:circuitextract}, we demonstrate the algorithm for extracting a circuit from a measurement pattern whose underlying \LOG\ has extended gflow.
The conclusions are given in Section~\ref{sec:conclusion}.

\section{Preliminaries} \label{sec:preliminaries}

We give a quick overview over the \zxcalculus in Section~\ref{sec:zx}
and introduce the one-way model of measurement-based quantum computing in Section~\ref{sec:MBQC}.
The graph-theoretic operations of local complementation and pivoting
(on which the rewrite rules for measurement patterns are based)
and their representation in the \zxcalculus are presented in Section~\ref{sec:lc}.
Section~\ref{sec:gflow} contains the definitions and properties of different notions of flow.

\subsection{The ZX-calculus}
\label{sec:zx}

The \zxcalculus is a diagrammatic language similar to the widely-used
quantum circuit notation. We will provide only a brief overview here,
for an in-depth reference see~\cite{CKbook}.

A \emph{\zxdiagram} consists of \emph{wires} and \emph{spiders}.
Wires entering the diagram from the left are called \emph{input wires}; wires exiting to the right are called \emph{output wires}.
Given two diagrams we can compose them horizontally (denoted $\circ$)
by joining the output wires of the first to the input wires of the second, or
form their tensor product (denoted $\otimes$) by simply stacking the two diagrams vertically.

Spiders are linear maps which can have any number of input or output
wires.  There are two varieties: $Z$ spiders, depicted as green dots, and $X$ spiders, depicted as red dots.\footnote{If you are reading this
  document in monochrome or otherwise have difficulty distinguishing green and red, $Z$ spiders will appear lightly-shaded and $X$ spiders will appear darkly-shaded.}
  Written explicitly in Dirac `bra-ket' notation, these linear maps are:
\[
\small
\hfill \input{Zsp-a.tikz} \ := \ \ketbra{0...0}{0...0} +
e^{i \alpha} \ketbra{1...1}{1...1} \hfill
\qquad
\hfill \input{Xsp-a.tikz} \ := \ \ketbra{+...+}{+...+} +
e^{i \alpha} \ketbra{-...-}{-...-} \hfill
\]
A \zxdiagram with $m$ input wires and $n$ output wires then represents a linear map $(\mathbb C^2)^{\otimes m} \to (\mathbb C^2)^{\otimes n}$ built from
spiders (and permutations of qubits) by composition and tensor product
of linear maps.  As a special case, diagrams with no inputs and $n$ outputs represent vectors in $(\mathbb C^2)^{\otimes n}$, i.e.
(unnormalised) $n$-qubit states.

\begin{example}\label{ex:basic-maps-and-states}
  We can immediately write down some simple state preparations and
  unitaries in the \zxcalculus:
  \[
  \begin{array}{rclcrcl}
  \begin{tikzpicture}
	\begin{pgfonlayer}{nodelayer}
		\node [style=Z dot] (0) at (-0.5, 0) {};
		\node [style=none] (1) at (0.5, 0) {};
	\end{pgfonlayer}
	\begin{pgfonlayer}{edgelayer}
		\draw (0) to (1.center);
	\end{pgfonlayer}
\end{tikzpicture} & = & \ket{0} + \ket{1} \ \propto \ket{+} &
  \qquad &
  \begin{tikzpicture}
	\begin{pgfonlayer}{nodelayer}
		\node [style=X dot] (0) at (-0.5, 0) {};
		\node [style=none] (1) at (0.5, 0) {};
	\end{pgfonlayer}
	\begin{pgfonlayer}{edgelayer}
		\draw (0) to (1.center);
	\end{pgfonlayer}
\end{tikzpicture} & = & \ket{+} + \ket{-} \ \propto \ket{0} \\[1em]
  \begin{tikzpicture}
	\begin{pgfonlayer}{nodelayer}
		\node [style=Z phase dot] (0) at (0, 0) {$\alpha$};
		\node [style=none] (1) at (-1.25, 0) {};
		\node [style=none] (2) at (1.25, 0) {};
	\end{pgfonlayer}
	\begin{pgfonlayer}{edgelayer}
		\draw (1.center) to (0);
		\draw (0) to (2.center);
	\end{pgfonlayer}
\end{tikzpicture} & = & \ketbra{0}{0} + e^{i \alpha} \ketbra{1}{1} =
  Z_\alpha &
  &
  \begin{tikzpicture}
	\begin{pgfonlayer}{nodelayer}
		\node [style=X phase dot] (0) at (0, 0) {$\alpha$};
		\node [style=none] (1) at (-1.25, 0) {};
		\node [style=none] (2) at (1.25, 0) {};
	\end{pgfonlayer}
	\begin{pgfonlayer}{edgelayer}
		\draw (1.center) to (0);
		\draw (0) to (2.center);
	\end{pgfonlayer}
\end{tikzpicture} & = & \ketbra{+}{+} + e^{i \alpha} \ketbra{-}{-} = X_\alpha
  \end{array}
  \]
  In particular we have the Pauli matrices:
  \[
  \hfill
  \begin{tikzpicture}
	\begin{pgfonlayer}{nodelayer}
		\node [style=Z phase dot] (0) at (0, 0) {$\pi$};
		\node [style=none] (1) at (-1.25, 0) {};
		\node [style=none] (2) at (1.25, 0) {};
	\end{pgfonlayer}
	\begin{pgfonlayer}{edgelayer}
		\draw (1.center) to (0);
		\draw (0) to (2.center);
	\end{pgfonlayer}
\end{tikzpicture} = Z \qquad\qquad   \begin{tikzpicture}
	\begin{pgfonlayer}{nodelayer}
		\node [style=X phase dot] (0) at (0, 0) {$\pi$};
		\node [style=none] (1) at (-1.25, 0) {};
		\node [style=none] (2) at (1.25, 0) {};
	\end{pgfonlayer}
	\begin{pgfonlayer}{edgelayer}
		\draw (1.center) to (0);
		\draw (0) to (2.center);
	\end{pgfonlayer}
\end{tikzpicture} = X \qquad\qquad
  \hfill
  \]
\end{example}
It will be convenient to introduce a symbol -- a yellow square -- for
the Hadamard gate. This is defined by the equation:
\begin{equation}\label{eq:Hdef}
\hfill
\input{had-alt.tikz}
\hfill
\end{equation}
We will often use an alternative notation to simplify the diagrams,
and replace a Hadamard between two spiders by a blue dashed edge, as
illustrated below.

\ctikzfig{blue-edge-def}

Both the blue edge notation and the Hadamard box can always be
translated back into spiders when necessary. We will refer to the blue
edge as a \emph{Hadamard edge}.

\begin{definition}\label{def:interpretation}
The \emph{interpretation} of a \zxdiagram $D$ is the linear map that such a diagram represents
and is written $\intf{D}$.
For a full treatment of the interpretation of a ZX diagram see, e.g.\ Ref.~\cite{SimonCompleteness}.
We say two \zxdiagrams $D_1$ and $D_2$ are \emph{equivalent} when $\intf{D_1}=z\intf{D_2}$ for some non-zero complex number $z$.
\end{definition}

We define equivalence up to a global scalar, as those scalars will not be important for the class of diagrams we study in this paper.

\begin{remark}\label{rem:global-scalars}
    There are many different sets of rules for the \zxcalculus. The version we present only preserves equality up to a global scalar. Versions of the \zxcalculus where equality is `on the nose' can be found in Ref.~\cite{Backens:2015aa} for the stabiliser fragment, in Ref.~\cite{SimonCompleteness} for the Clifford+T fragment, and in Ref.~\cite{JPV-universal,ng2017universal} for the full language.
\end{remark}

The interpretation of a \zxdiagram is invariant under
moving the vertices around in the plane, bending,
unbending, crossing, and uncrossing wires, so long as the connectivity
and the order of the inputs and outputs is maintained.
Furthermore, there is an additional set of equations that we call the \emph{rules} of the \zxcalculus; these are shown in
Figure~\ref{fig:zx-rules}. Two diagrams are equivalent if one can be transformed into the other using the rules of the \zxcalculus.

\begin{figure}[h]
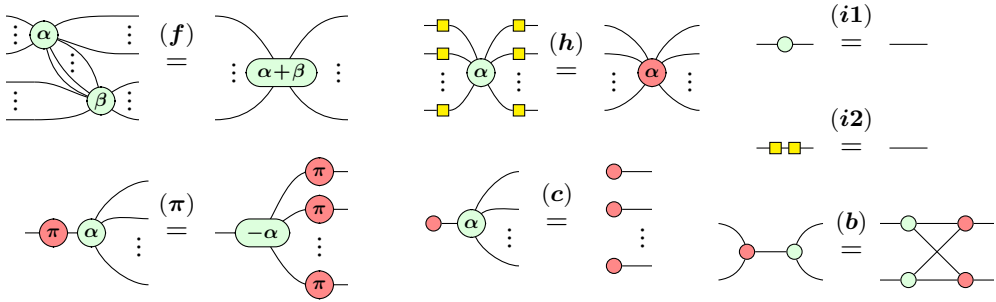

\ctikzfig{ZX-rules}
\vspace{-3mm}
\caption{\label{fig:zx-rules} A convenient presentation for the ZX-calculus. These rules hold
  for all $\alpha, \beta \in [0, 2 \pi)$, and due to $(\bm{h})$ and
  $(\bm{i2})$ all rules also hold with the colours   interchanged. Note the ellipsis should be read as `0 or more', hence the spiders on the LHS of \SpiderRule are connected by one or more wires.}
\end{figure}

\begin{remark}\label{rem:completeness}
  The \zxcalculus is \emph{universal} in the sense that any linear map can be expressed as a \zxdiagram. Furthermore, when restricted to \textit{Clifford \zxdiagrams}, i.e. diagrams whose phases are all multiples of $\pi/2$, the version we present in Figure~\ref{fig:zx-rules} can be made \emph{complete} as follows: for any two non-zero Clifford \zxdiagrams that are equivalent (cf.\ Definition~\ref{def:interpretation}), there exists a series of rewrites transforming one into the other if a meta rule is added that allows non-zero scalar subdiagrams to be dropped. For completeness when considering exact equality, see Remark~\ref{rem:global-scalars}. Recent extensions to the calculus have been introduced which are complete for the larger \textit{Clifford+T} family of \zxdiagrams \cite{SimonCompleteness}, where phases are multiples of $\pi/4$, and for \textit{all} \zxdiagrams~\cite{HarnyAmarCompleteness,JPV-universal,euler-zx}.
\end{remark}

Quantum circuits can be translated into \zxdiagrams in a straightforward manner.
We will take as our starting point circuits constructed
from the following universal set of gates:
\[
\CNOT
\qquad\qquad
Z_\alpha
\qquad\qquad
H
\]
We choose this gate
set because it admits a convenient representation in terms of
spiders:
\begin{align}\label{eq:zx-gates}
\CNOT & = \input{cnot.tikz} &
Z_\alpha & = \begin{tikzpicture}
	\begin{pgfonlayer}{nodelayer}
		\node [style=Z phase dot] (0) at (0, 0) {$\alpha$};
		\node [style=none] (1) at (-1.25, 0) {};
		\node [style=none] (2) at (1.25, 0) {};
	\end{pgfonlayer}
	\begin{pgfonlayer}{edgelayer}
		\draw (1.center) to (0);
		\draw (0) to (2.center);
	\end{pgfonlayer}
\end{tikzpicture} &
H & = \begin{tikzpicture}
	\begin{pgfonlayer}{nodelayer}
		\node [style=none] (0) at (-0.25, 0) {};
		\node [style=none] (1) at (1.75, 0) {};
		\node [style=hadamard] (2) at (0.75, 0) {};
	\end{pgfonlayer}
	\begin{pgfonlayer}{edgelayer}
		\draw (0.center) to (2);
		\draw (2) to (1.center);
	\end{pgfonlayer}
\end{tikzpicture}
\end{align}
These gates have the following interpretations:

\begin{align*}
\intf{\input{cnot.tikz}} &=
\left(\begin{matrix}
  1 & 0 & 0 & 0 \\
  0 & 1 & 0 & 0 \\
  0 & 0 & 0 & 1 \\
  0 & 0 & 1 & 0 \\
\end{matrix}\right)
\qquad
\intf{\begin{tikzpicture}
	\begin{pgfonlayer}{nodelayer}
		\node [style=Z phase dot] (0) at (0, 0) {$\alpha$};
		\node [style=none] (1) at (-1.25, 0) {};
		\node [style=none] (2) at (1.25, 0) {};
	\end{pgfonlayer}
	\begin{pgfonlayer}{edgelayer}
		\draw (1.center) to (0);
		\draw (0) to (2.center);
	\end{pgfonlayer}
\end{tikzpicture}} &=
\left(\begin{matrix}
  1 & 0 \\
  0 & e^{i \alpha}
\end{matrix}\right)
\qquad
\intf{\begin{tikzpicture}
	\begin{pgfonlayer}{nodelayer}
		\node [style=none] (0) at (-0.25, 0) {};
		\node [style=none] (1) at (1.75, 0) {};
		\node [style=hadamard] (2) at (0.75, 0) {};
	\end{pgfonlayer}
	\begin{pgfonlayer}{edgelayer}
		\draw (0.center) to (2);
		\draw (2) to (1.center);
	\end{pgfonlayer}
\end{tikzpicture}} &= \frac{1}{\sqrt{2}}
\left(\begin{matrix}
  1 & 1 \\
  1 & -1
\end{matrix}\right)
\end{align*}

For the \CNOT gate, the green spider is the first (i.e. control) qubit and the red spider is the second (i.e. target) qubit. Other common gates can easily be expressed in terms of these gates. In particular, $S := Z_{\frac\pi2}$, $T := Z_{\frac\pi4}$ and:
\begin{align}\label{eq:zx-derived-gates}
X_\alpha & = \input{X-a-expanded.tikz} &
\CZ & = \input{cz-small.tikz}
\end{align}

\begin{remark}
  Note that the directions of the wires in the depictions of the \CNOT and \CZ gates are irrelevant, hence we can draw them vertically without ambiguity. Undirected wires are a general feature of \zxdiagrams, and from hence forth we will ignore wire directions without further comment. We will also freely draw inputs/outputs entering or exiting the diagram from arbitrary directions if the meaning is clear from context (as for example in Example~\ref{ex:gflow-in-action}).
\end{remark}

\noindent For our purposes, we will define quantum circuits as a special case of \zxdiagrams.

\begin{definition}\label{def:circuit}
  A \emph{circuit} is a \zxdiagram generated by composition and tensor products of the \zxdiagrams in equations~\eqref{eq:zx-gates} and~\eqref{eq:zx-derived-gates}.
  The interpretation of such a circuit is given by the composition and tensor products of the interpretations of the generating diagrams given
  in  equation~\eqref{eq:zx-gates}, in accordance with:
  \begin{align*}
  \intf{D \otimes D'} = \intf{D} \otimes \intf{D'} \qquad \intf{D \circ D'} = \intf{D} \circ \intf{D'}
  \end{align*}
\end{definition}

Important subclasses of circuits are \textit{Clifford circuits}, sometimes called stabilizer circuits, which are obtained from compositions of only \CNOT, $H$, and $S$ gates. They are efficiently classically simulable, thanks to the \textit{Gottesman-Knill theorem}~\cite{aaronsongottesman2004}. A unitary is \textit{local Clifford} if it arises from a single-qubit Clifford circuit, i.e. a composition of $H$ and $S$ gates.
The addition of $T$ gates to Clifford circuits yields \textit{Clifford+T circuits}, which are capable of approximating any $n$-qubit unitary to arbitrary precision, whereas the inclusion of $Z_\alpha$ gates for all $\alpha$ enables one to construct any unitary exactly~\cite{NielsenChuang}.

\subsection{Measurement-based quantum computation}
\label{sec:MBQC}
Measurement-based quantum computing (MBQC) is a universal model for quantum computation, underlying the \emph{one-way quantum computing} scheme \cite{MBQC2}.
The basic resource for MBQC is a \emph{graph state}~\cite{hein2004multiparty,hein2006entanglement}, built from a register of qubits by applying $CZ$-gates pairwise to obtain an entangled quantum state.\footnote{There are models of MBQC where the basic resource is not a graph state, but we do not consider those models in this paper.}
The graph state is then gradually consumed by measuring single qubits.
By choosing an appropriate resource state and appropriate measurements, any quantum computation can be performed.
The difficulty is the non-determinism inherent in quantum measurements, which means computations need to be adaptive to implement deterministic operations.

Measurement-based quantum computations are usually formalised in terms of \emph{measurement patterns}, a syntax describing both the construction of graph states and their processing via successive measurements.
In the following, we first present measurement patterns, and then introduce other
-- and, in our opinion, simpler --
formalisms for reasoning about these computations.
Instead of allowing arbitrary single-qubit measurements, measurements are usually restricted to three planes of the Bloch sphere, labelled XY, XZ, and YZ.
For each measurement, the state denoted `$+$' is taken to be the desired result of the measurement and the state denoted `$-$' is the undesired result, which will need to be adaptively corrected.
The allowed measurements are (\cite[p.~292]{danos_kashefi_panangaden_perdrix_2009}):
\begin{align*}
  \ket{+_{\XYm,\alpha}} &= \frac{1}{\sqrt{2}}\left(\ket{0} + e^{i\alpha} \ket{1} \right) &
  \ket{-_{\XYm,\alpha}} &= \frac{1}{\sqrt{2}}\left(\ket{0} - e^{i\alpha} \ket{1} \right) \\
  \ket{+_{\XZm,\alpha}} &= \cos\left(\frac{\alpha}{2}\right)\ket{0} + \sin\left(\frac{\alpha}{2}\right) \ket{1} &
  \ket{-_{\XZm,\alpha}} &= \sin\left(\frac{\alpha}{2}\right)\ket{0} - \cos\left(\frac{\alpha}{2}\right) \ket{1} \\
  \ket{+_{\YZm,\alpha}} &= \cos\left(\frac{\alpha}{2}\right)\ket{0} + i \sin\left(\frac{\alpha}{2}\right) \ket{1} &
  \ket{-_{\YZm,\alpha}} &= \sin\left(\frac{\alpha}{2}\right)\ket{0} - i \cos\left(\frac{\alpha}{2}\right) \ket{1}
\end{align*}
\noindent where $0 \leq \alpha \leq 2\pi$.
Usually, the desired measurement outcome is labelled 0 and the undesired measurement outcome is labelled 1.

\begin{definition}[\cite{Patterns}]\label{def:meas_pattern}
  A \emph{measurement pattern} consists of an $n$-qubit register $V$ with distinguished sets $I, O \sse V$ of input and output qubits and a sequence of commands consisting of the following operations:
  \begin{itemize}
    \item Preparations $N_i$, which initialise a qubit $i \notin I$ in the state $\ket{+}$.
    \item Entangling operators $E_{ij}$, which apply a $CZ$-gate to two distinct qubits $i$ and $j$.
    \item Destructive measurements $M_i^{\ld,\gamma(\alpha,s,t)}$, which project a qubit $i\notin O$ onto the orthonormal basis $\{\ket{+_{\ld,\gamma}},\ket{-_{\ld,\gamma}}\}$, where $\lambda \in \{ \XYm, \XZm, \YZm \}$ is the measurement plane, $\alpha$ is the non-corrected measurement angle and $\gamma(\alpha,s,t)\coloneqq (-1)^s\alpha+t\pi$ is the corrected one with variables $s,t\in\{0,1\}$ indicating measurement outcomes. The projector $\ket{+_{\ld,\gamma}}\bra{+_{\ld,\gamma}}$ corresponds to outcome $0$ and $\ket{-_{\ld,\gamma}}\bra{-_{\ld,\gamma}}$ corresponds to outcome $1$. If there are no corrections involved, we will simply write $M_i^{\ld,\alpha}$.
    \item Corrections $[X_i]^s$, which depend on a measurement outcome (or a linear combination of measurement outcomes) $s\in\{0,1\}$ and act as the Pauli-$X$ operator on qubit $i$ if $s$ is $1$ and as the identity otherwise,
    \item Corrections $[Z_j]^t$, which depend on a measurement outcome (or a linear combination of measurement outcomes) $t\in\{0,1\}$ and act as the Pauli-$Z$ operator on qubit $j$ if $t$ is $1$ and as the identity otherwise.
  \end{itemize}
\end{definition}

We will only consider \emph{runnable patterns}, which satisfy certain conditions ensuring they are implementable in practice.

\begin{definition}[\cite{Patterns}]\label{def:runnable_pattern}
A measurement pattern is \emph{runnable} if the following conditions hold.
\begin{itemize}
\item No correction depends on an outcome not yet measured.
\item All non-input qubits are prepared.
\item All non-output qubits are measured.
\item A command $C$ acts on a qubit $i$ only if $i$ has not already been measured, and one of (1)-(3) holds:
\begin{enumerate}[label=({\arabic*})]
\item $i$ is an input and $C$ is not a preparation,
\item $i$ has been prepared and $C$ is not a preparation,
\item $i$ has not been prepared, $i$ is not an input, and $C$ is a preparation.
\end{enumerate}
\end{itemize}
\end{definition}

The entangling operators $E_{ij}$ in a pattern determine the resource graph state.
In fact, they correspond to the edges of the underlying \emph{\LOG}. This is formalised in the following definitions and remark, which will be used throughout the paper.

\begin{definition}
    An \emph{open graph} is a tuple $(G,I,O)$ where $G=(V,E)$ is an undirected graph, and $I,O \sse V$ are distinguished (possibly overlapping) subsets representing \emph{inputs} and \emph{outputs}. We will write $\comp{O} := V\setminus O$ for the \emph{non-outputs} and $\comp{I}:= V\setminus I$ for the \emph{non-inputs}. If $v \in \comp{I}$ and $v \in \comp{O}$, we call
    $v$ an \emph{internal} vertex, and if $v\in I\cup O$, we call $v$ a \emph{boundary vertex}.
    For vertices $u,v \in V$ we write $u\sim v$ when $u$ and $v$ are adjacent in $G$, and denote by $N_G(u)\coloneqq\{w\in V \mid w\sim u\}$ the set of neighbours of $u$.
\end{definition}

\begin{definition}[{cf.~\cite[p.5]{GFlow}}]\label{def:LOG}
  A \emph{\LOG} is a tuple $\Gamma = (G,I,O, \lambda)$ where $(G,I,O)$ is an open graph, and $\lambda : \comp{O} \rightarrow \{ \XYm, \YZm, \XZm\}$ assigns a measurement plane to each non-output vertex.%
  \footnote{We adopt the notation of this from Ref.~\cite{GFlow} where their corresponding notion is called an `open graph state'. While they explicitly consider this to be a particular type of quantum state, our notion of a labelled open graph is just a type of graph. We define the associated `open graph state' separately in Definition~\ref{def:ogs-to-linear-map}.}
\end{definition}

\begin{remark}
 A \LOG\ and an assignment of measurement angles $\alpha:\comp{O}\to [0,2\pi)$ carry the same information as a runnable measurement pattern with no corrections.

 Given the measurement pattern, the corresponding \LOG\ $(G,I,O,\ld)$ may be determined as follows: the vertices of the graph $G$ are given by the set of qubits $V$.
The edges of $G$ are given by the set of pairs $(i, j) \in V \times V$ such
that the total number of occurences of $E_{ij}$ and $E_{ji}$ in the pattern is odd.
The sets $I$ and $O$ are the ones given in Definition~\ref{def:meas_pattern}.
The functions $\ld$ and $\al$ are determined by the measurement operators $M_i^{\ld,\al}$; Definition~\ref{def:runnable_pattern} guarantees that both are well-defined.

Given a \LOG\ we can apply this process in reverse to construct a runnable measurement pattern without corrections.

A \LOG\ and an assignment of measurement angles can also be determined from a measurement pattern including corrections by simply ignoring the latter (i.e.\ by assuming that all measurements yield the desired result).
In Section~\ref{sec:gflow}, we discuss necessary and sufficient conditions under which appropriate corrections commands can be determined when constructing a measurement pattern from a \LOG; these corrections then put constraints on the order of the measurements.
\end{remark}

In general, a single measurement pattern can result in a variety of measurement instructions, with each instruction depending on earlier measurement outcomes and the resulting corrections that must be applied.
We are, however, interested in the subset of measurement patterns that always implement the same linear map, regardless of the measurement outcomes.
For these patterns, we can identify the unique linear map implemented by the pattern with the set of measurement instructions obtained when all the measurement outcomes are $0$ (and thus no corrections are necessary).
This leads us to the following definition:
\begin{definition}\label{def:ogs-to-linear-map}
 Suppose $\Gamma=(G,I,O,\ld)$ is a \LOG, and let $\alpha:\comp{O}\to [0,2\pi)$ be an assignment of measurement angles.
 The \emph{linear map associated with $\Gamma$ and $\alpha$} is given by
 \[
  M_{\Gamma,\alpha} := \left( \prod_{i\in\comp{O}} \bra{+_{\ld(i),\alpha(i)}}_i \right) E_G N_{\comp{I}},
 \]
  where $E_G := \prod_{i\sim j} E_{ij}$ and $N_{\comp{I}} := \prod_{i\in\comp{I}} N_i$.
\end{definition}

\begin{remark}
 Note that the projections $\bra{+_{\ld(i),\alpha(i)}}_i$ on different qubits $i$ commute with each other.
 Similarly, the controlled-Z operations $E_{ij}$ commute even if they involve some of the same qubits.
 Finally, all the state preparations $N_i$ on different qubits commute.
 Thus, $M_{\Gamma,\alpha}$ is fully determined by $\Gamma$ and $\alpha$.
\end{remark}

\begin{definition}\label{def:ogs-to-ZX}
 Suppose $\Gamma=(G,I,O,\ld)$ is a \LOG\ and $\alpha:\comp{O}\to [0,2\pi)$ is an assignment of measurement angles.
 Then its \emph{associated \zxdiagram} $D_{\Gamma,\alpha}$ is defined by translating the expression for $M_{\Gamma,\alpha}$ from Definition~\ref{def:ogs-to-linear-map} according to Table~\ref{tab:MBQC-to-ZX}. The elements are composed in the obvious way and any sets of adjacent phase-free green spiders other than measurement effects are merged. In other words, merging affects green spiders which come from the translation of a preparation or entangling command.
 \begin{table}
  \centering
  \renewcommand{\arraystretch}{2}
  \begin{tabular}{c||c|c|c|c|c}
   operator & $N_i$ & $E_{ij}$ & $\bra{+_{\XYm,\alpha(i)}}_i$ & $\bra{+_{\XZm,\alpha(i)}}_i$ & $\bra{+_{\YZm,\alpha(i)}}_i$ \\ \hline
   diagram & \input{plus-state.tikz} & \input{cz.tikz} & \input{XY-effect.tikz} & \input{XZ-effect.tikz} & \input{YZ-effect.tikz}
  \end{tabular}
  \renewcommand{\arraystretch}{1}
  \caption{Translation from an associated linear map to a \zxdiagram.}
  \label{tab:MBQC-to-ZX}
 \end{table}
\end{definition}

\begin{example}
  The measurement pattern with the qubit register $V=\{ 1,2,3,4\}$, input and output sets $I=\{ 1,2 \}$ and $O = \{ 1,4 \}$ and the sequence of commands
  $$ M_2^{\XYm,\frac{\pi} 2}M_3^{\YZm,\pi}E_{14}E_{23}E_{24} E_{34} N_3 N_4$$
  is represented by the following \zxdiagram:
  \ctikzfig{example-MBQC-translation}
\end{example}

\begin{lemma}\label{lem:zx-equals-linear-map}
 Suppose $\Gamma=(G,I,O,\ld)$ is a \LOG\ and $\alpha:\comp{O}\to [0,2\pi)$ is an assignment of measurement angles.
 Let $M_{\Gamma,\alpha}$ be the linear map specified in Definition~\ref{def:ogs-to-linear-map} and let $D_{\Gamma,\alpha}$ be the \zxdiagram constructed according to Definition~\ref{def:ogs-to-ZX}.
 Then $\intf{D_{\Gamma,\alpha}}=M_{\Gamma,\alpha}$.
\end{lemma}
\begin{proof}
 For each operator $M$ in Table~\ref{tab:MBQC-to-ZX} and its corresponding diagram $D_M$, it is straightforward to check that $\intf{D_M}=M$.
 The result thus follows by the compositional properties of the interpretation $\intf{\cdot}$ and the fact that rewriting ZX-diagrams preserves semantics.
\end{proof}

In order to specify a converse direction to this result, we will define a special class of ZX-diagrams. Before we do that, we recall which ZX-diagrams correspond to graph states.

\begin{definition}[\cite{DP1}]\label{def:graph-state}
  A \emph{graph state diagram} is a \zxdiagram where all vertices are green, all the connections between vertices are Hadamard edges and a single output wire is incident on each vertex in the diagram.
  The \emph{graph corresponding to a graph state diagram} is the graph whose vertices are the green spiders of the \zxdiagram and whose edges are given by the Hadmard edges of the \zxdiagram.
\end{definition}

\begin{definition}\label{def:MBQC-form}
 A \zxdiagram is in \emph{MBQC form} if it consists of a graph state diagram in which each vertex of the graph may also be connected to:
 \begin{itemize}
  \item an input (in addition to its output), and
  \item a measurement effect (in one of the three measurement planes) instead of the output.
 \end{itemize}
\end{definition}

\begin{definition}\label{def:graph-of-diagram}
 Given a \zxdiagram $D$ in MBQC form, its \emph{underlying graph} $G(D)$ is the graph corresponding to the graph state part of $D$.
\end{definition}

See Figure~\ref{fig:graph-state} for an example of a graph state diagram and a diagram in MBQC form.

\begin{figure}
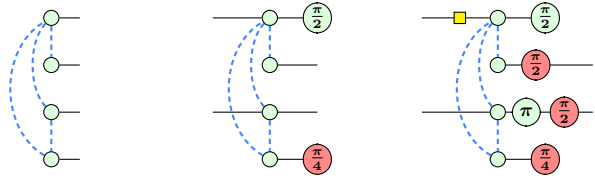

    \ctikzfig{graph-state-ex}
    \caption{On the left, a graph state diagram. In the middle, a diagram in MBQC form with the same underlying graph. On the right, an MBQC+LC form diagram with the same underlying labelled open graph.\label{fig:graph-state}}
\end{figure}

\noindent Given these definitions we can now show the following:

\begin{lemma}\label{lem:ogs-to-ZX-is-MBQC-Form}
 Let $\Gamma=(G,I,O,\ld)$ be a \LOG\ and let $\alpha:\comp{O}\to [0,2\pi)$ be an assignment of measurement angles.
 Then the \zxdiagram $D_{\Gamma,\alpha}$ constructed according to Definition~\ref{def:ogs-to-ZX} is in MBQC form.
\end{lemma}
\begin{proof}
 Consider performing the translation described in Definition~\ref{def:ogs-to-ZX} in two steps.
 The first step involves translating the preparation and entangling commands of the linear map $M_{\Gamma,\alpha}$ according to Table~\ref{tab:MBQC-to-ZX} and then merging any sets of adjacent green spiders.
 This yields a graph state diagram with some additional inputs.
 (The underlying graph is $G$.)
 The second step is the translation of the measurement projections of $M_{\Gamma,\alpha}$.
 This yields measurement effects on some of the outputs of the graph state diagram.
 Thus, the resulting \zxdiagram is in MBQC form by Definition~\ref{def:MBQC-form}.
\end{proof}

The converse of Lemma~\ref{lem:ogs-to-ZX-is-MBQC-Form} also holds.

\begin{lemma}\label{lem:zx-to-pattern}
 Suppose $D$ is a \zxdiagram in MBQC form.
 Then there exists a \LOG\ $\Gamma=(G,I,O,\ld)$ and an assignment of measurement angles $\alpha:\comp{O}\to [0,2\pi)$ such that $\intf{D} = M_{\Gamma,\alpha}$.
\end{lemma}
\begin{proof}
 Let $G:=G(D)$ be the underlying graph of the \zxdiagram $D$, cf.\ Definition~\ref{def:graph-of-diagram}.
 Define $I\sse V$ to be the set of vertices of $D$ on which an input wire is incident.
 Analogously, define $O\sse V$ to be the set of vertices of $D$ on which an output wire is incident.
 Fix $\ld:\comp{O}\to\{\XYm,\XZm,\YZm\}$ by using Table~\ref{tab:MBQC-to-ZX} in reverse to determine the measurement planes from the measurement effects in the \zxdiagram.
 Let $\Gamma := (G,I,O,\ld)$.
 Finally, define $\alpha:\comp{O}\to [0,2\pi)$ to be the phase of the measurement effect connected to each non-output vertex in the \zxdiagram.
 Then $D = D_{\Gamma,\alpha}$ and thus the desired result follows from Lemma~\ref{lem:zx-equals-linear-map}.
\end{proof}

\begin{remark}
 Lemmas~\ref{lem:ogs-to-ZX-is-MBQC-Form} and \ref{lem:zx-to-pattern}
 show that the correspondence between MBQC form \zxdiagrams and
 the pairs $(\Gamma,\alpha)$, where $\Gamma$ is a \LOG\ and $\alpha$ is an assignment of measurement angles, is one-to-one up to isomorphisms of labelled open graphs that also preserve $\alpha$.
\end{remark}

It will turn out to be useful to consider a `relaxed' version of the MBQC form for \zxdiagrams.

\begin{definition}
  We say a \zxdiagram is in \emph{MBQC+LC} form when it is in MBQC form (see Definition~\ref{def:MBQC-form}) up to arbitrary single-qubit Clifford unitaries on the input and output wires (LC stands for `local Clifford').
  When considering the underlying graph of a \zxdiagram in MBQC+LC form, we ignore these single qubit Clifford unitaries.
\end{definition}
Note that an MBQC form diagram is an MBQC+LC form diagram with trivial single-qubit unitaries on its inputs and outputs. An example diagram in MBQC+LC form is given in Figure~\ref{fig:graph-state}.

\subsection{Graph-theoretic rewriting}
\label{sec:lc}

The rewrites we will use are based on the graph-theoretic notions of \emph{local complementation} and \emph{pivoting}.
We present these operations (and their effects) in Definitions~\ref{def:loc-comp} and~\ref{def:pivot} as they appear in Ref.~\cite{DP3}.
Our interest is in the effect these operations have on a measurement pattern.
In particular, we consider whether a \zxdiagram in MBQC form will remain in MBQC form after applying a local complementation or pivot
(or remain close enough to MBQC form to still be useful).

\begin{definition}[{\cite{kotzig}}]\label{def:loc-comp}
Let $G=(V,E)$ be a graph and $u\in V$ a vertex. The {\em local complementation of $G$ about the vertex $u$} is the operation resulting in the graph
$$G\star u\coloneqq \left( V, E\symd\{(b,c) : (b,u), (c,u)\in E\ \textrm{and}\ b\neq c\}\right) ,$$
where $\symd$ is the symmetric set difference, i.e.~$A\symd B\coloneqq (A\cup B)\setminus (A\cap B)$.
\end{definition}
In other words, $G\star u$ is a graph that has the same vertices as $G$. Two neighbours $b$ and $c$ of $u$ are connected in $G\star u$ if and only if they are not connected in $G$. All other edges are the same as in $G$.

\begin{definition}[{\cite{kotzig}}]\label{def:pivot}
Let $G=(V,E)$ be a graph and $u,v\in V$ two vertices connected by an edge. The \emph{pivot of $G$ about the edge $u\sim v$} is the operation resulting in the graph $G\land uv\coloneqq G\star u\star v\star u$.
\end{definition}
If we denote the set of vertices connected to both $u$ and $v$ by $A$,
the set of vertices connected to $u$ but not to $v$ by $B$
 and the set of vertices connected to $v$ but not to $u$ by $C$,
 then pivoting consists of interchanging $u$ and $v$ and complementing the edges between each pair of sets $A$, $B$ and $C$.
 That is, a vertex in $A$ is connected to a vertex in $B$ after pivoting if and only if the two vertices are not connected before pivoting; and similarly for the two other pairs.
 All the remaining edges are unchanged, including the edges internal to $A$, $B$ and $C$.
 We illustrate this by the following picture, where crossing lines between two sets indicate complementing the edges.
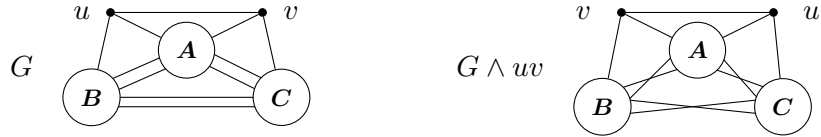
\[G \quad\input{pivot-L.tikz}\qquad\qquad \quad G\wedge uv \quad\input{pivot-R.tikz}
\]
\begin{remark}\label{rem:pivot_sym}
From the above characterisation it follows that pivoting is symmetric in the (neighbouring) vertices, that is, $G\land uv = G\land vu$.
\end{remark}

In the \zxcalculus, a spider with a zero phase and exactly two incident wires is equivalent to a plain wire (representing the identity operation) by rule $(\textit{\textbf {i1}})$ in Figure~\ref{fig:zx-rules}. The following definition represents the corresponding graph operation, which will be used to remove such vertices.
\begin{definition}\label{def:identity-removal}
Let $G=(V,E)$ be a graph and let $u,v,w\in V$ be vertices such that $N_G(v)=\{u,w\}$ and $u\notin N_G(w)$, that is, the neighbours of $v$ are precisely $u$ and $w$, and $u$ is not connected to $w$.
We then define \emph{identity removal} as
$$G\idrem{v} w\coloneqq ((G\land uv)\setminus\{u\})\setminus\{v\}.$$
Since $v$ has exactly two neighbours, one of which is $w$,
the choice of $u$ is implicit in the notation.
We think of this as `dragging $u$ along $v$ to merge with $w$'.
\end{definition}
The effect of the identity removal is to remove the middle vertex $v$ and to fuse the vertices $u$ and $w$ into one, as illustrated in the picture below. Thus the operation is symmetric in $u$ and $w$, in the sense that $G\idrem{v} u$ and $G\idrem{v} w$ are equal up to a relabelling of one vertex. Note that $u$ and $w$ are allowed to have common neighbours (which will disconnect from the fused vertex $w$ as a result of identity removal).
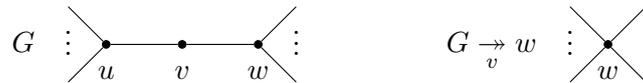
\[G \quad\input{identity-removal-L.tikz}\qquad\qquad \quad G\idrem{v} w \quad\input{identity-removal-R.tikz}
\]
\begin{example}\label{ex:identity-removal}
Consider the following graph:
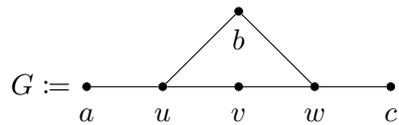
\[G \coloneqq \input{identity-removal-example1.tikz}.
\]
Note that the vertices $u,v$ and $w$ satisfy the condition for identity removal: $u$ and $w$ are not connected and are precisely the neighbours of $v$. Hence identity removal results in the graph
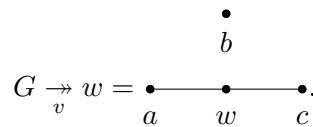
\[G\idrem{v} w = \input{identity-removal-example3.tikz}.
\]
\end{example}

\begin{remark}\label{rem:identity-removal-connected-vertices}
If we have vertices $u,v$ and $w$ with $N_G(v)=\{u,w\}$ but, unlike in the definition above, $u\in N_G(w)$, we can first perform a local complementation on $v$, so that $u$ and $w$ become disconnected, and then remove the identity vertex $v$. In symbols:
$$(G\star v)\idrem{v} w .$$
\end{remark}

The abstract application of a local complementation to a graph corresponds to the application of a specific set of local Clifford gates on the corresponding graph state:
\begin{theorem}[\cite{NestMBQC}, in the manner of {\cite[Theorem~2]{DP1}}]\label{thm:lc-in-zx}
 Let $G=(V,E)$ be a graph with adjacency matrix $\theta$ and let $u\in V$, then
 \[
  \ket{G\star u} = X_{\pi/2,u}\otimes\bigotimes_{v\in V} Z_{-\pi/2,v}^{\theta_{uv}}\ket{G}.
 \]
\end{theorem}

This result can be represented graphically in the \zxcalculus:

\begin{lemma}[{\cite[Theorem~3]{DP1}}]\label{lem:ZX-lcomp}
  The following equality involving graph state diagrams and Clifford phase shifts follows from the graphical rewrite rules:
  {
  \ctikzfig{local-comp-ex}
  }
  Here, the underlying graph on the LHS is $G\star u$ and the underlying graph on the RHS is $G$.
  Any vertices not adjacent to $u$ are unaffected and are not shown in the above diagram.
\end{lemma}

Combining this result with the definition of pivoting in terms of local complementations (cf.\ Definition~\ref{def:pivot}) we also get:

\begin{lemma}[{\cite[Theorem~3.3]{DP3}}]\label{lem:ZX-pivot}
  The following diagram equality follows from the graphical rewrite rules:
    \ctikzfig{pivot-desc}
    Here $u$ and $v$ are a connected pair of vertices, and the underlying graph on the LHS is $G\wedge uv$ while the RHS is $G$.
  Any vertices not adjacent to $u$ or $v$ are unaffected and are not shown in the above diagram.
\end{lemma}

\subsection{Generalised flow}
\label{sec:gflow}

The notion of \emph{flow} or \emph{causal flow} on open graphs was introduced by Danos and Kashefi~\cite{Danos2006Determinism-in-} as a sufficient condition to distinguish those open graphs capable of supporting a deterministic MBQC pattern with measurements in the \XY-plane.
Causal flow, however, is not a necessary condition for determinism.
That is there are graphs that implement a deterministic pattern even though they do not have causal flow~\cite{GFlow,duncan2010rewriting}.
Browne et al.~\cite{GFlow} adapted the notion of flow to what they called
\emph{generalised flow} (gflow), which is both a necessary and sufficient condition for
`uniformly and strongly stepwise deterministic' measurement patterns (defined below).
Unlike causal flow, gflow can also be applied to arbitrary \LOG{}s, i.e.\ it supports measurement patterns with measurements in all three planes.
This even more general case is sometimes called \emph{extended gflow}.

\begin{definition}[{\cite[p.5]{GFlow}}]\label{def:determinism}
    The linear map implemented by a measurement pattern for a specific set of measurement outcomes is called a \emph{branch} of the pattern.
    A pattern is \emph{deterministic} if all branches are equal up to a scalar.
	A pattern is \emph{strongly deterministic} if all branches are equal up to a global phase.
	It is \emph{uniformly deterministic} if it is deterministic for any choice of measurement angles.
	Finally, the pattern is \emph{stepwise deterministic} if any intermediate pattern -- resulting from performing some subset of the measurements and their corresponding corrections -- is again deterministic.
\end{definition}

The existence of gflow implies the uniform, strong and stepwise determinism of any pattern on the open graph (cf.~Theorem~\ref{t-flow} below).
Hence, by applying transformations to an open graph that preserve the existence of a gflow, we ensure that the modified open graph still supports a uniformly, strongly and stepwise deterministic pattern.
Note that the condition of interest is preservation of the \textit{existence} of gflow, not preservation of the specific gflow itself.

We will now give the formal definitions of (causal) flow and (extended) gflow.

\begin{definition}[{\cite[Definition~2]{Danos2006Determinism-in-}}]\label{def:causal-flow}
 Let $(G,I,O)$ be an open graph. We say $G$ has \textit{(causal) flow} if there exists a map $f:\comp{O}\longrightarrow \comp{I}$ (from measured qubits to prepared qubits) and a strict partial order $\prec$ over $V$ such that for all $u\in \comp{O}$:
\begin{itemize}
    \item $u \sim f(u)$
    \item $u \prec f(u)$
    \item $u \prec v$ for all neighbours $v\neq u$ of $f(u)$.
\end{itemize}
\end{definition}

When vertices are smaller than a vertex $v$ in the order~$\prec$, they are referred to as being `behind' or `in the past of' of $v$.

The notion of gflow differs from the above definition of causal flow in two ways.
The value of $f(u)$ is allowed to be a set of vertices instead of a single vertex, so that corrections can be applied to more than one vertex at a time.
As a result of this change, the third condition of causal flow is now too strong: requiring that no element of $f(u)$ is adjacent to any vertex `in the past' of $u$ would be too restrictive.
Since corrections are applied to sets of vertices at a time, it is possible to make use of the following idea: if corrections are simultaneously applied to an even number of neighbours of $v$, then there is no net effect on $v$.
Thus, the second change takes the form of a parity condition: all vertices in the neighbourhood of $f(u)$ that lie `in the past' of $u$ are required to be in the even neighbourhood of $f(u)$.
As a result, net effects of corrections do not propagate into `the past'.

Allowing measurements in more than one measurement plane requires further careful adjustment of the parity conditions depending on the measurement plane of the vertex being measured.

\begin{definition}[{\cite[p.7]{GFlow}}]
Given a graph $G=(V,E)$, for any $K\sse V$, let $\odd{G}{K}= \{u\in V: \abs{N(u)\cap K}\equiv 1 \mod 2\}$ be the \emph{odd neighbourhood} of $K$ in $G$, i.e.\ the set of vertices having an odd number of neighbours in $K$.
If the graph $G$ is clear from context, we simply write $\odd{}{K}$.
The \emph{even neighbourhood} of $K$ in $G$, $\eve{G}{K}$,  is defined in a similar way; $\eve{G}{K}= \{u\in V: \abs{N(u)\cap K}\equiv 0 \mod 2\}$.
\end{definition}

\begin{definition}[{\cite[Definition~3]{GFlow}}]
\label{defGFlow}
 A \LOG{} $(G,I,O,\ld)$ has generalised flow (or \emph{gflow}) if there exists a map $g:\comp{O}\to\pow{\comp{I}}$ and a partial order $\prec$ over $V$ such that for all $v\in \comp{O}$:
 \begin{enumerate}[label=({g}\theenumi), ref=(g\theenumi)]
  \item\label{it:g} If $w\in g(v)$ and $v\neq w$, then $v\prec w$.
  \item\label{it:odd} If $w\in\odd{}{g(v)}$ and $v\neq w$, then $v\prec w$.
  \item\label{it:XY} If $\ld(v)=\XYm$, then $v\notin g(v)$ and $v\in\odd{}{g(v)}$.
  \item\label{it:XZ} If $\ld(v)=\XZm$, then $v\in g(v)$ and $v\in\odd{}{g(v)}$.
  \item\label{it:YZ} If $\ld(v)=\YZm$, then $v\in g(v)$ and $v\notin\odd{}{g(v)}$.
 \end{enumerate}
 The set $g(v)$ is called the \emph{correction set} of $v$.
\end{definition}

\begin{remark}
 Every causal flow is indeed a gflow, where $g(v):=\{f(v)\}$ and the partial order remains the same.
 To see this, first note causal flow can only be defined on \LOG{}s where $\ld(v)=\XYm$ for all $v\in\comp{O}$, so conditions \ref{it:XZ} and \ref{it:YZ} are vacuously satisfied.
 Now, \ref{it:g} follows from the second bullet point of Definition~\ref{def:causal-flow}, \ref{it:odd} follows from the third bullet point, and \ref{it:XY} follows from the first bullet point.
\end{remark}

\begin{remark}
 In the original definition of gflow in Ref.~\cite{GFlow}, condition \ref{it:odd} is given as:
\begin{align}\label{eq:wrong-g2}
\text{if }j \preccurlyeq i \text{ and } j \neq i \text{ then } j \notin \odd{}{g(i)}
\end{align}
In other publications, such as Ref.~\cite{danos_kashefi_panangaden_perdrix_2009}, the definition is changed to the version we give as \ref{it:odd}, yet this is usually done without comment.
For completeness, we provide an example which demonstrates that the condition \eqref{eq:wrong-g2} is insufficient for determinism.
Consider the following open graph:
 \ctikzfig{bialgebra}
Here, the set of inputs is $\{i_1,i_2\}$, the set of outputs is $\{o_1,o_2\}$, and the non-outputs are measured in the planes $\lambda(i_1) = \lambda(i_2) = \XYm$.
If we choose both measurement angles to be~$0$, it is straightforwardly checked that this diagram implements the linear map:
\[
	\begin{pmatrix} 1&1&1&1\\ 1&-1&-1&1\\1&1&1&1\\ 1&-1&-1&1
	\end{pmatrix}
\]
This has rank 2 and thus is not invertible, and in particular not unitary.
It therefore cannot be deterministically implementable and hence it should not have a gflow.
However, it is considered to have a gflow under condition \eqref{eq:wrong-g2} instead of \ref{it:odd}:
pick the partial order $i_1 \prec o_1, o_2$ and $i_2 \prec o_1, o_2$ with all other vertices incomparable.
Set $g(i_1) = \{o_1\}$ and $g(i_2) = \{o_2\}$. It is then easily checked that $(g,\prec)$ satisfies conditions \ref{it:g}, \eqref{eq:wrong-g2}, and \ref{it:XY}--\ref{it:YZ}.
Yet $(g,\prec)$ does not satisfy condition \ref{it:odd} and hence is not a gflow under the revised definition.
\end{remark}

To demonstrate that, with condition~\ref{it:odd}, the presence of a gflow indeed guarantees determinism of a pattern, we give a detailed proof of the following sufficiency theorem, which was first stated in Ref.~\cite{GFlow} as Theorem 2 with a sketch proof. The precise statement of the theorem requires some additional notation and the proof is quite lengthy and technical, so we state a coarse version of the theorem here and refer the reader to Appendix~\ref{sec:gflow-determinism} (and more specifically to Theorem~\ref{t-flow-app}) for the details.

\begin{theorem}\label{t-flow}
Let $\Gamma = (G,I,O,\ld)$ be a \LOG~with a gflow and let $\alpha:\comp{O}\rightarrow [0,2\pi)$ be an assignment of measurement angles. Then there exists a runnable measurement pattern which is uniformly, strongly and stepwise deterministic, and which realizes the associated linear map $M_{\Gamma,\alpha}$ (cf.~Definition~\ref{def:ogs-to-linear-map}).
\end{theorem}

The converse also holds:
\begin{theorem}
 If a pattern is stepwise, uniformly and strongly deterministic, then its underlying \LOG{} $(G, I, O, \ld)$ has a gflow.
\end{theorem}
\begin{proof}
We present the proof sketch here, a complete treatment of the proof can be found in Ref.~\cite[Theorem~7.9.7]{danos_kashefi_panangaden_perdrix_2009}. The proof is by induction on the number of measurements. If the number of measurements is $n=0$, then the pattern trivially has a gflow. Suppose the pattern has $n+1$ qubits to be measured. Since the pattern is assumed to be stepwise deterministic, after performing the first measurement, the remaining pattern is still stepwise, uniformly and strongly deterministic. Hence it has a gflow by the induction hypothesis, where the partial order is given by the order in which the measurements are performed.

It remains to extend this gflow to include the first qubit to be measured. The essential part of this is to find a subset $S \subseteq \comp{I}$ that can act as the correction set of the first measurement (cf.\ Definition~\ref{defGFlow}). Given such a subset $S$, we define the full gflow as:
\begin{align*}
g^\prime(i) :=
\begin{cases}
 g(i) & i\neq n \\ S& i=n,
\end{cases}
\end{align*}
where $g$ is the gflow of the smaller pattern.
\end{proof}

It is not straightforward to actually find a concrete gflow from the procedure in this proof. A constructive algorithm has since been given in Ref.~\cite{mhalla2011graph}, which finds gflow on \LOG{}s where all measurements are in the \XY plane. In Section~\ref{sec:MaximallyDelayedGflow}, we extend this algorithm to find extended gflow on \LOG{}s with measurements in all three planes.

Gflow is a property that applies to \LOG{}s. For convenience we also define it for ZX-diagrams.
\begin{definition}\label{dfn:zx-gflow}
We say a \zxdiagram in MQBC(+LC) form has \emph{gflow} $(g,\prec)$ if the corresponding \LOG\ $\Gamma$ has gflow $(g,\prec)$.
\end{definition}

The following result from Ref.~\cite{cliff-simp} shows that any unitary circuit can be converted into a deterministic measurement pattern.

\begin{proposition}[{\cite[Lemma 3.7]{cliff-simp}}]\label{prop:circuit-to-pattern}
    Given a circuit there is a procedure for converting it into an equivalent measurement pattern.
    Furthermore, this measurement pattern only contains \XY-plane measurements and has causal flow, so it also has gflow.
\end{proposition}

Note that Ref.~\cite{cliff-simp} does not explicitly talk about measurement patterns.
What they call graph-like diagrams however correspond in a straightforward manner to diagrams in MBQC form with every measured vertex being measured in the \XY-plane.
(We also note that the procedure of Proposition~\ref{prop:circuit-to-pattern} takes $O(n^2)$ operations,
where $n$ is the number of vertices.)

Below, we present a concrete example illustrating how the presence of gflow allows measurement errors to be corrected as the computation progresses.
\begin{example}\label{ex:gflow-in-action}
Consider the following \LOG{} $\Gamma$:
\[
 \input{gflow-example-geometry.tikz}
\]
where $a$ is an input, $e$ and $f$ are outputs, and the measurement planes are given by $\ld(a)=\ld(b)=\XYm$, $\ld(c)=\XZm$ and $\ld(d)=\YZm$. As usual, we denote the measurement angles by $\alpha : V \rightarrow [0,2\pi)$, where $V$ is the vertex set of $\Gamma$. Using Definition~\ref{def:ogs-to-ZX} (that is, we translate according to Table~\ref{tab:MBQC-to-ZX}), we obtain the corresponding \zxdiagram:
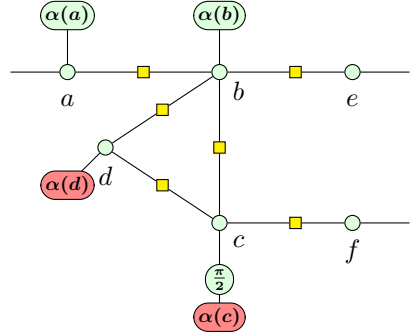
\[\input{gflow-example-zx.tikz}
\]
Note that the \LOG\ we started with has a gflow $(g,\prec)$ given by the following partial order
$$a\prec b\prec c\prec d\prec e,f,$$
with the function $g$ taking the values
\begin{align*}
g(a) &= \{b\} \\
g(b) &= \{c\} \\
g(c) &= \{c,d\} \\
g(d) &= \{d,e,f\}.
\end{align*}
It follows that we have
\begin{align*}
\odd{}{g(a)} &= \{a,c,d,e\} \\
\odd{}{g(b)} &= \{b,d,f\} \\
\odd{}{g(c)} &= \{c,d,f\} \\
\odd{}{g(d)} &= \varnothing,
\end{align*}
from which it is easy to verify that the conditions \ref{it:g}-\ref{it:YZ} hold.

By Theorem~\ref{t-flow}, the presence of gflow guarantees that we can correct the measurement errors provided that we measure the qubits according to the partial order. We demonstrate this for $\Gamma$ in Figure~\ref{fig:error-propagation}.

Thus suppose a measurement error of $\pi$ occurs when performing the measurement corresponding to the vertex $c$, as indicated in the top left part of Figure~\ref{fig:error-propagation}. The labels in this figure refer to the rules in Figure~\ref{fig:zx-rules}. In order to get the left diagram on the second row, we move each red $\pi$-phase past the corresponding Hadamard gate, which changes the colour of the phase to green. For the right diagram, the left green $\pi$ travels past the green node on the left and flips the sign of $\alpha(d)$. Next, to obtain the left diagram on the third row, the middle green $\pi$ travels along the middle triangle and past another Hadamard gate to become a red $\pi$. Finally, in the bottom left diagram, the red $\pi$ on the left has been fused with $-\alpha(d)$; and the red $\pi$ on the right has passed through the Hadamard gate switching its colour to green, and the adjacent green nodes have fused into a node with phase $\pi-\frac{\pi}{2}=\frac{\pi}{2}$. We then rearrange the diagram so that it looks like a measurement pattern again.

Note that all the vertices that are affected by the error are above $c$ in the partial order and hence `not yet measured' at this stage of the computation. Thus the necessary corrections may be applied to these vertices when they are measured.
\end{example}
\begin{figure}
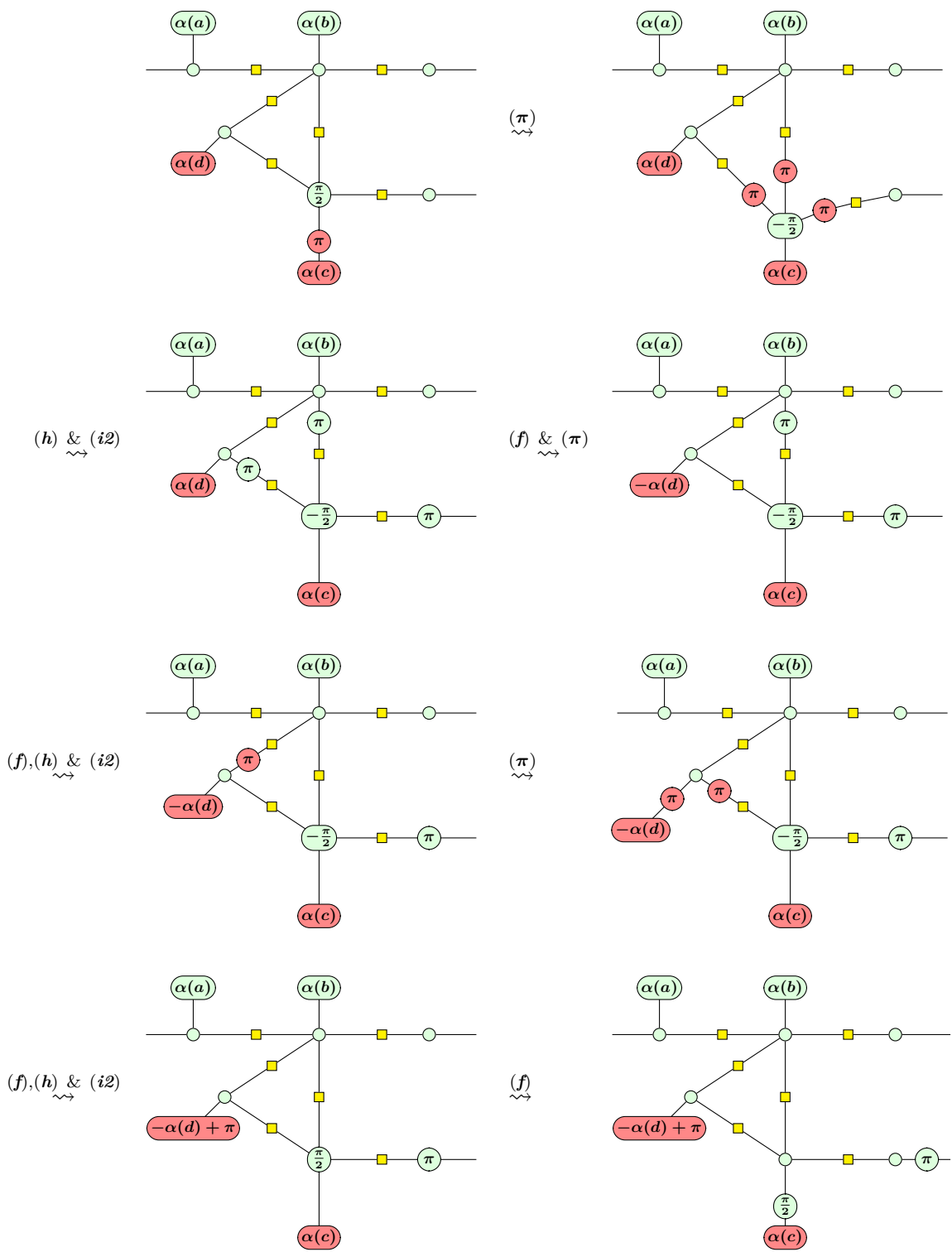

\begin{align*}
&\input{gflow-example-corr2-11.tikz}\quad\stackrel{(\boldsymbol{\pi})}{\rightsquigarrow} & &\input{gflow-example-corr2-12.tikz} \\ \\
\stackrel{(\textit{\textbf h})\ \&\ (\textit{\textbf {i2}})}{\rightsquigarrow}\quad &\input{gflow-example-corr2-21.tikz}\quad\stackrel{(\textit{\textbf f})\ \&\ (\boldsymbol{\pi})}{\rightsquigarrow} & &\input{gflow-example-corr2-22.tikz} \\ \\
\stackrel{(\textit{\textbf f}), (\textit{\textbf h})\ \&\ (\textit{\textbf {i2}})}{\rightsquigarrow}\quad &\input{gflow-example-corr2-31.tikz}\quad\stackrel{(\boldsymbol{\pi})}{\rightsquigarrow} & &\input{gflow-example-corr2-32.tikz} \\ \\
\stackrel{(\textit{\textbf f}), (\textit{\textbf h})\ \&\ (\textit{\textbf {i2}})}{\rightsquigarrow}\quad &\input{gflow-example-corr2-41.tikz}\quad\stackrel{(\textit{\textbf f})}{\rightsquigarrow} & &\input{gflow-example-corr2-42.tikz}
\end{align*}
\caption{\label{fig:error-propagation} Propagation of a measurement error of the pattern in Example~\ref{ex:gflow-in-action}.}
\end{figure}

\subsection{Focusing gflow for {\XY} plane measurements}\label{sec:focusing-gflow}

For a \LOG{} $(G,I,O,\ld)$ in which all measurements are in the \XY-plane, there is a special type of gflow which is specified by the correction function alone.
This gflow is called \emph{focused} because of the property that, among non-output vertices, corrections only affect the vertex they are meant to correct.
A \LOG{} where all measurements are in the \XY plane has gflow if and only if it has focused gflow.

\begin{definition}[{adapted from \cite[Definition~5]{mhalla2011graph}}]\label{def:focused-gflow}
 Suppose $(G,I,O,\ld)$ is a \LOG{} with the property that $\ld(v)=\XYm$ for all $v\in\comp{O}$.
 Then $(g,\prec)$ is a \emph{focused gflow} on $(G,I,O,\ld)$ if for all $v\in\comp{O}$, we have $\odd{G}{g(v)}\cap \comp{O}=\{v\}$, and furthermore $\prec$ is the transitive closure of the relation $\{(v,w) \mid v\in\comp{O} \wedge w\in g(v)\}$.
\end{definition}

\begin{theorem}[{reformulation of \cite[Theorem~2]{mhalla2011graph}}]\label{thm:mhalla2}
 Suppose $(G,I,O,\ld)$ is a \LOG{} with the property that $\ld(v)=\XYm$ for all $v\in\comp{O}$, then $(G,I,O,\ld)$ has gflow if and only if it has a focused gflow.
\end{theorem}

A \LOG{} in which all measurements are in the \XY-plane can be \emph{reversed} by swapping the roles of inputs and outputs.
More formally:

\begin{definition}\label{def:reversed-LOG}
 Suppose $(G,I,O,\ld)$ is a \LOG{} with the property that $\ld(v)=\XYm$ for all $v\in\comp{O}$.
 The corresponding \emph{reversed \LOG{}} is the \LOG{} where the roles of inputs and outputs are swapped, i.e.\ it is $(G,O,I,\ld')$, where $\ld'(v):=\XYm$ for all $v\in\comp{I}$.
\end{definition}

Now if the number of inputs and outputs in the \LOG{} is the same, its focused gflow can also be reversed in the following sense.

\begin{corollary}\label{cor:reverse_unitary_gflow}
 Suppose $(G,I,O,\ld)$ is a \LOG{} with the properties that $\abs{I}=\abs{O}$ and $\ld(v)=\XYm$ for all $v\in\comp{O}$, and suppose it has a focused gflow $(g,\prec)$.
 For all $v\in\comp{I}$, let $g'(v):=\{w\in\comp{O} \mid v\in g(w)\}$, and for all $u,w\in V$, let $u\prec' w$ if and only if $w\prec u$.
 Then $(g',\prec')$ is a focused gflow for the reversed \LOG{} $(G,O,I,\ld')$.
\end{corollary}

This follows immediately from the proofs of Ref.~\cite[Theorems~3--4]{mhalla2011graph} but it is not explicitly stated in that paper.

\section{Rewriting while preserving the existence of gflow}
\label{sec:rewriting}

In this section, we study a variety of topics dealing with \LOG{}s and gflow.
In the first subsection, we show how certain graph operations, such as local complementation and pivoting, affect the gflow.
In Section~\ref{sec:MaximallyDelayedGflow} we give a polynomial time algorithm for finding extended gflow using the concept of \emph{maximally delayed} gflow.
We combine this notion with that of a \emph{focused} extended gflow in Section~\ref{sec:focusing-extended-gflow} to transform a given gflow to give it certain useful properties.

\subsection{Graph operations that preserve the existence of gflow}

In this section, we prove some of our main technical lemmas, establishing that local complementation and related graph rewrites interact well with the gflow of the graph.

First, we show that a \LOG{} resulting from the local complementation of a \LOG{} with gflow will also have a gflow.

\newcommand{\statelcgflow}{
Let $(g,\prec)$ be a gflow for $(G,I,O,\ld)$ and let $u\in\comp{O}\cap\comp I$. Then $(g',\prec)$ is a gflow for $(G\star u, I, O,\ld')$, where
 \[
  \ld'(u) := \begin{cases} \XZm &\text{if } \ld(u)=\XYm \\ \XYm &\text{if } \ld(u)=\XZm \\ \YZm &\text{if } \ld(u)=\YZm \end{cases}
 \]
 and for all $v\in \comp{O}\setminus\{u\}$
 \[
  \ld'(v) := \begin{cases} \YZm &\text{if } v\in N_G(u) \text{ and } \ld(v)=\XZm \\ \XZm &\text{if } v\in N_G(u) \text{ and } \ld(v)=\YZm \\ \ld(v) &\text{otherwise.} \end{cases}
 \]
 Furthermore,
 \[
  g'(u) := \begin{cases} g(u)\symd \{u\} &\text{if } \ld(u)\in\{\XYm,\XZm\} \\ g(u) &\text{if } \ld(u)=\YZm \end{cases}
 \]
 and for all $v\in \comp{O}\setminus\{u\}$,
 \[
  g'(v) := \begin{cases} g(v) &\text{if } u\notin\odd{G}{g(v)} \\ g(v)\symd g'(u) \symd \{u\} &\text{if } u\in\odd{G}{g(v)}. \end{cases}
 \]
}

\begin{lemma}\label{lem:lc_gflow}
 \statelcgflow
\end{lemma}

\noindent The proof of this lemma can be found in Appendix~\ref{sec:proofs}. Note that the condition that the complemented vertex is not an output can in fact be dropped:

\begin{lemma}
 Let $(g,\prec)$ be a gflow for $(G,I,O,\ld)$ and let $u\in O$. Then $(g',\prec)$ is a gflow for $(G\star u, I, O,\ld')$, where for all $v\in \comp{O}$
 \[
  \ld'(v) := \begin{cases} \YZm &\text{if } v\in N_G(u) \text{ and } \ld(v)=\XZm \\ \XZm &\text{if } v\in N_G(u) \text{ and } \ld(v)=\YZm \\ \ld(v) &\text{otherwise.} \end{cases}
 \]
 Furthermore, for all $v\in \comp{O}$,
 \[
  g'(v) := \begin{cases} g(v) &\text{if } u\notin\odd{G}{g(v)} \\ g(v) \symd \{u\} &\text{if } u\in\odd{G}{g(v)}. \end{cases}
 \]
\end{lemma}
\begin{proof}
 The proof is basically the same as that of Lemma~\ref{lem:lc_gflow} if we take $g(u)$ and $g'(u)$ to be empty.
 The output vertex has no label, so its label does not need to be updated.
\end{proof}

Now by applying this lemma three times we see that a pivot also preserves the existence of a gflow.

\newcommand{\statecorpivotgflow}{
Let $(G,I,O,\ld)$ be a \LOG\ which has a gflow, and let $u,v\in\comp{O}\cap\comp I$ be connected by an edge. Then $(G\land uv, I, O,\hat\ld)$, where
 \[
  \hat\ld(a) = \begin{cases} \YZm &\text{if } \ld(a)=\XYm \\
                                     \XZm &\text{if } \ld(a)=\XZm \\
                                     \XYm &\text{if } \ld(a)=\YZm \end{cases}
 \]
 for $a\in\{u,v\}$, and $\hat\ld(w)=\ld(w)$ for all $w\in \comp{O}\setminus\{u,v\}$ also has a gflow.
}

\begin{corollary}\label{cor:pivot_gflow}
 \statecorpivotgflow
\end{corollary}

For more details regarding the correctness of this corollary, we refer to Appendix~\ref{sec:proofs}.

Somewhat surprisingly, the deletion of some types of vertices preserves the existence of gflow:

\begin{lemma}\label{lem:deletepreservegflow}
    Let $(g,\prec)$ be a gflow for $(G,I,O,\ld)$ and let $u\in \comp{O}$ with $\ld(u) \neq \XYm$. Then $(g',\prec)$ is a gflow for $(G\setminus\{u\},I,O,\ld)$ where $\forall v\in V, v\neq u$:
    \[g'(v) := \begin{cases} g(v) &\text{if } u\not \in g(v)\\ g(v)\symd g(u) &\text{if } u \in g(v) \end{cases}\]
\end{lemma}
\begin{proof}
    First, observe that $u\in g(u)$ since $\ld(u)\neq \XYm$. Thus $u\not \in g'(v)$ for either case of the definition. Hence, $g'$ is indeed a function on the graph $G\setminus u$.
    To check that $g'$ is indeed a gflow we check the necessary conditions for all $v\in G\setminus\{u\}$. If $u\not \in g(v)$, then $g'(v) = g(v)$ and hence we are done. If $u \in g(v)$, then $v\prec u$ and hence also $v\prec w$ for all $w \in g(u)$ or $w\in \odd{G}{g(u)}$. Since $g'(v) = g(v)\symd g(u)$,  conditions \ref{it:g} and \ref{it:odd} are satisfied. For conditions \ref{it:XY}-\ref{it:YZ}, note that we cannot have $v \in g(u)$ or $v\in \odd{G}{g(u)}$ because $v\prec u$. As a result, $v\in g'(v) \iff v\in g(v)$ and $v\in \odd{G\setminus\{u\}}{g'(v)} \iff v \in \odd{G}{g(v)}$. Since the labels of all the vertices stay the same, \ref{it:XY}-\ref{it:YZ} remain satisfied.
\end{proof}

\begin{remark}
    The condition that $\ld(u) \neq \XYm$ is necessary in the previous lemma. Removing a vertex with label \XY will, in general, lead to a \LOG{} which no longer has a gflow. For instance consider the following \LOG:
    \ctikzfig{line-graph}
    where the first two vertices both have label \XY. This graph has a gflow specified by $I\prec u \prec O$ and $g(I) = \{u\}$, $g(u) = \{O\}$, but removing $u$ will disconnect the graph and hence the resulting graph does not have a gflow.

    Note that if we were to consider the same \LOG, but with $u$ measured in a different plane, it would \emph{not} have a gflow to start with.
    This is because, if it did, we would need $u\in g(I)$, so that $I\prec u$ but also $u\in g(u)$ so that $I\in \odd{}{g(u)}$ giving $u\prec I$.
    Hence this does not contradict the lemma.
\end{remark}

The next corollary shows how the previous results can be combined to remove a vertex with arity 2 from a \LOG{} while preserving gflow. In the ZX-calculus, the idea behind this is that we use \IdRule to remove a vertex and then \SpiderRule to fuse the adjacent vertices (cf.~Definition~\ref{def:identity-removal}).
Recall from that definition that $G\idrem{v} w\coloneqq ((G\land uv)\setminus\{u\})\setminus\{v\}$.

\begin{corollary}\label{cor:id_removal}
Let $(g,\prec)$ be a gflow for the \LOG{} $(G,I,O,\ld)$, and let $u,v\in\comp O$ and $w\in V$ be vertices such that $N_G(v)=\{u,w\}$ and $\ld(u),\ld(v)\neq \YZ$. Then $(\tilde g,\prec)$ as defined below is a gflow for $(G\idrem{v} w,I,O,\ld)$.
For all $z\in \comp O\setminus\{u,v\}$ we have
 \[
  \tilde g(z) = \begin{cases} \hat g(z) &\text{if } u\notin\hat g(z), v\notin\hat g(z) \\
                              \hat g(z)\symd\hat g(v) &\text{if } u\notin\hat g(z)\symd\hat g(v), v\in\hat g(z) \\
                              \hat g(z)\symd\hat g(u)\symd\hat g(v) &\text{if } u\in\hat g(z)\symd\hat g(v), v\in\hat g(z)\symd\hat g(u) \\
                              \hat g(z)\symd\hat g(u) &\text{if } u\in\hat g(z), v\notin\hat g(z)\symd\hat g(u), \end{cases}
 \]
where $\hat g$ is as defined in Corollary \ref{cor:pivot_gflow}.
\end{corollary}
\begin{proof}
A computation using Corollary \ref{cor:pivot_gflow} and Lemma \ref{lem:deletepreservegflow}.
\end{proof}

\begin{lemma}\label{lem:gflow-add-output}
Let $\Gamma=(G,I,O,\ld)$ be a \LOG\ with $G=(V,E)$. Let $\Gamma'$ be the \LOG\ that results from converting an output $u\in O$ into a vertex measured in the \XY-plane and adding a new output vertex $u'$ in its stead:
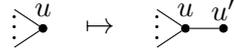
\ctikzfig{gflow-add-output}
Formally, let $\Gamma'=(G',I,O',\ld')$, where $G'=(V',E')$ with $V'=V\cup\{u'\}$ and $E'=E\cup\{u\sim u'\}$, $O'=(O\setminus\{u\})\cup\{u'\}$, and $\ld'(v)$ is the extension of $\ld$ to domain $V'\setminus O'$ with $\ld'(u)=\XYm$. Then if $\Gamma$ has gflow, $\Gamma'$ also has gflow.
\end{lemma}
\begin{proof}
 Suppose $\Gamma$ has gflow $(g,\prec)$.
 Let $g'$ be the extension of $g$ to domain $V'\setminus O'$ which satisfies $g'(u)=\{u'\}$, and let $\prec'$ be the transitive closure of $\prec\cup\{(u,u')\}$.

 The tuple $(g',\prec')$ inherits \ref{it:g} and \ref{it:XY}--\ref{it:YZ} for all $v\in V\setminus O$ because the correction sets have not changed for any of the original vertices.
 Furthermore, $u'\in\odd{G'}{g'(v)}$ for any $v$ implies $u\in g'(v)$, as $u$ is the only neighbour of $u'$.
 Hence $u'\in\odd{G'}{g'(v)}$ implies $v\prec' u \prec' u'$.
 Therefore \ref{it:odd} continues to be satisfied for all $v\in V\setminus O$.

 Now, for $u$, \ref{it:g} holds because $u\prec' u'$ by definition, \ref{it:odd} holds because $\odd{G'}{g'(u)}=\{u\}$, and \ref{it:XY} can easily be seen to hold.
 Thus, $(g',\prec')$ is a gflow for $\Gamma'$.
\end{proof}

\begin{lemma}\label{lem:gflow-add-input}
Let $\Gamma=(G,I,O,\ld)$ be a \LOG\ with $G=(V,E)$.
 Let $\Gamma'$ be the \LOG\ that results from adding an additional vertex measured in the \XY-plane `before' the input $u\in I$:
 \ctikzfig{gflow-add-input}
 Formally, let $\Gamma'=(G',I',O,\ld')$, where $G'=(V',E')$ with $V'=V\cup\{u'\}$ and $E'=E\cup\{u\sim u'\}$, $I'=(I\setminus\{u\})\cup\{u'\}$, and $\ld'(v)$ is the extension of $\ld$ to domain $V'\setminus O$ which satisfies $\ld'(u')=\XYm$. Then if $\Gamma$ has gflow, $\Gamma'$ also has gflow.
\end{lemma}
\begin{proof}
 Suppose $\Gamma$ has gflow $(g,\prec)$.
 Let $g'$ be the extension of $g$ to domain $V'\setminus O$ which satisfies $g'(u')=\{u\}$, and let $\prec'$ be the transitive closure of $\prec\cup\{(u',w):w\in N_G(u)\cup\{u\}\}$.

 The tuple $(g',\prec')$ inherits the gflow properties for all $v\in V\setminus O$ because the correction sets have not changed for any of the original vertices and because the additional inequalities in $\prec'$ do not affect the gflow properties for any $v\in V\setminus O$.
 The latter is because
 \begin{itemize}
  \item $u'\notin g'(v)=g(v)$ for any $v\in V\setminus O$, and
  \item $u'\notin\odd{G'}{g'(v)}=\odd{G'}{g(v)}$ for any $v\in V\setminus O$ since its only neighbour $u$ was an input in $\Gamma$ and thus satisfies $u\notin g(v)$ for any $v\in V\setminus O$.
 \end{itemize}
 Now, for $u'$, \ref{it:g} holds by the definition of $\prec'$.
 Note that $\odd{G'}{g(u')}=N_{G'}(u)$, so \ref{it:odd} also holds by the definition of $\prec'$.
 Finally, \ref{it:XY} holds because $u'\notin g(u')$ and $u'\in\odd{G'}{g(u')}=N_{G'}(u)$.
 Thus, $(g',\prec')$ is a gflow for $\Gamma'$.
\end{proof}

\subsection{Finding extended gflow} \label{sec:MaximallyDelayedGflow}

Ref.~\cite{MP08-icalp} gives a polynomial time algorithm for finding gflow for \LOG{}s with all measurements in the \XY-plane. In this section, we present an extension of this algorithm that works for measurements in all three measurement planes.

Before doing so, we note a few details from the algorithm. The intuition behind the procedure is to `maximally delay' any measurements,
thereby keeping potential correction options available for as long as possible. As a result, the algorithm finds a gflow of minimal `depth' (a notion we will make precise later).

The algorithm works backwards:
Starting from the output vertices, it iteratively constructs disjoint subsets of vertices
that can be corrected by vertices chosen in previous steps.
Viewed instead as information travelling forwards, from the inputs to the outputs,
this corresponds to only correcting a vertex at the last possible moment,
hence the name maximal delayed.

\begin{definition}[Generalisation of {\cite[Definition~4]{MP08-icalp}} to multiple measurement planes]
\label{defVk}
 For a given labelled open graph $(G,I,O,\ld)$ and a given gflow $(g,\prec)$ of $(G,I,O,\ld)$, let
 \[
  V_k^\prec = \begin{cases} \max_\prec (V) &\text{if } k= 0 \\ \max_\prec (V\setminus(\bigcup_{i<k} V_i^\prec)) &\text{if } k > 0 \end{cases}
 \]
 where $\max_\prec(X) := \{u\in X \text{ s.t. } \forall v\in X, \neg(u\prec v)\}$ is the set of the maximal elements of $X$.
\end{definition}

\begin{definition}[Generalisation of {\cite[Definition~5]{MP08-icalp}} to multiple measurement planes]
\label{defMoreDelayed}
 For a given labelled open graph $(G,I,O,\ld)$ and two given gflows $(g,\prec)$ and $(g',\prec')$ of $(G,I,O,\ld)$,
 $(g,\prec)$ is \emph{more delayed} than $(g',\prec')$ if for all $k$,
 \[
  \abs{\bigcup_{i=0}^k V_i^\prec} \geq \abs{\bigcup_{i=0}^k V_i^{\prec'}}
 \]
 and there exists a $k$ such that the inequality is strict.
 A gflow $(g,\prec)$ is \emph{maximally delayed} if there exists no gflow of the same open graph that is more delayed.
\end{definition}

\begin{theorem}[Generalisation of {\cite[Theorem~2]{MP08-icalp}} to multiple measurement planes]
\label{thmGFlowAlgo}
There exists a polynomial time algorithm that decides whether a given
\LOG\ has an extended gflow, and that outputs such a gflow if it exists.
Moreover, the output gflow is maximally delayed.
\end{theorem}
The proof of this theorem can be found in Appendix~\ref{sec:FindingGflow}.
The algorithm we find is a relatively straightforward generalisation of that of Ref.~\cite{MP08-icalp}, but the proofs required to show its correctness are complicated somewhat by the case distinctions required for the three measurement planes. We note also that our detailed proof shows that the original proofs of Ref.~\cite{MP08-icalp} missed some case distinctions and conditions.

\subsection{Focusing extended gflow}\label{sec:focusing-extended-gflow}

In Section~\ref{sec:focusing-gflow}, we introduced the notion of focused gflow for \LOG{}s in which all measurements are in the \XY plane.
There is no canonical generalisation of this notion to \LOG{}s with measurements in multiple planes.
Hamrit and Perdrix suggest three different extensions of focused gflow to the case of multiple measurement planes, which restrict correction operations on non-output qubits to only a single type of Pauli operator overall \cite[Definition~2]{hamrit2015reversibility}.
Here, we go a different route by requiring that non-output qubits only appear in correction sets, or odd neighbourhoods of correction sets, if they are measured in specific planes.
This means the correction operators which may be applied to some non-output qubit depend on the plane in which that qubit is measured.
The new notion of focused gflow will combine particularly nicely with the phase-gadget form of MBQC+LC diagrams of Section~\ref{sec:phasegadgetform}.

We begin by proving some lemmas that allow any gflow to be focused in our sense.

\begin{lemma}\label{lem:successor-gflow}
 Let $(G,I,O,\ld)$ be a labelled open graph which has gflow $(g,\prec)$.
 Suppose there exist $v,w\in\comp{O}$ such that $v\prec w$.
 Define $g'(v):=g(v)\symd g(w)$ and $g'(u):=g(u)$ for all $u\in\comp{O}\setminus\{v\}$, then $(g',\prec)$ is a gflow.
\end{lemma}
\begin{proof}
 As the correction set only changes for $v$, the gflow properties remain satisfied for all other vertices.
 Now, suppose $w'\in g'(v)$, then $w'\in g(v) \vee w'\in g(w)$.
 In the former case, $v\prec w'$, and in the latter case, $v\prec w\prec w'$, since $(g,\prec)$ is a gflow, so \ref{it:g} holds.
 Similarly, suppose $w'\in\odd{}{g'(v)}$, then by linearity of $\odd{}{\cdot}$ we have $w'\in\odd{}{g(v)} \vee w'\in\odd{}{g(w)}$.
 Again, this implies $v\prec w'$ or $v\prec w\prec w'$ since $(g,\prec)$ is a gflow, so \ref{it:odd} holds.
 Finally, $v\prec w$ implies $v\notin g(w)$ and $v\notin\odd{}{g(w)}$ since $(g,\prec)$ is a gflow.
 Therefore $v\in g'(v) \Longleftrightarrow v\in g(v)$ and $v\in\odd{}{g'(v)}\Longleftrightarrow v\in\odd{}{g(v)}$.
 Thus \ref{it:XY}--\ref{it:YZ} hold and $(g',\prec)$ is a gflow.
\end{proof}

\begin{lemma}\label{lem:focus-single-vertex}
 Let $(G,I,O,\ld)$ be a labelled open graph, let $(g,\prec)$ be a gflow for this open graph, and let $v\in\comp{O}$.
 Then there exists $g':\comp{O}\to\pow{\comp{I}}$ such that
 \begin{enumerate}
  \item for all $w\in\comp{O}$, either $v=w$ or $g'(w)=g(w)$,
  \item for all $w\in g'(v)\cap\comp{O}$, either $v=w$ or $\ld(w) = \XYm$,
  \item for all $w\in \odd{}{g'(v)}\cap\comp{O}$, either $v=w$ or $\ld(w)\neq \XYm$, and
  \item $(g',\prec)$ is a gflow for $(G,I,O,\ld)$.
 \end{enumerate}
 We can construct this $g'$ in a number of steps that is polynomial in the number of vertices of $G$.
\end{lemma}
\begin{proof}
 Let $g_0:=g$, we will modify the function in successive steps to $g_1,g_2$, and so on.
 For each non-negative integer $k$ in turn, define
 \begin{align*}
  S_{k,\XYm} &:= \{u\in (\odd{}{g_k(v)}\cap\comp{O}) \setminus\{v\} : \ld(u)=\XYm\}, \\
  S_{k,\XZm} &:= \{u\in (g_k(v)\cap\comp{O}) \setminus\{v\} : \ld(u)=\XZm\}, \\
  S_{k,\YZm} &:= \{u\in (g_k(v)\cap\comp{O}) \setminus\{v\} : \ld(u)=\YZm\},
 \end{align*}
 and set $S_k := S_{k,\XYm} \cup S_{k,\XZm} \cup S_{k,\YZm}$.
 Finding $S_k$ takes $O(\abs{V}^2)$ operations.
 If $S_k=\emptyset$, let $g':=g_k$ and stop.
 Otherwise, choose $w_k\in S_k$ among the elements minimal in $\prec$, and define
 \[
  g_{k+1}(u) := \begin{cases} g_k(v)\symd g_k(w_k) &\text{if } u=v \\ g_k(u) &\text{otherwise.} \end{cases}
 \]
 Note $w_k\in S_k$ implies $w_k\neq v$, as well as either $w_k\in g_k(v)$ or $w_k\in\odd{}{g_k(v)}$.
 Thus if $(g_k,\prec)$ is a gflow, then $v\prec w_k$, and hence by Lemma~\ref{lem:successor-gflow}, $(g_{k+1},\prec)$ is also a gflow.
 Since $(g_0,\prec)$ is a gflow, this means $(g_k,\prec)$ is a gflow for all $k$.

 Now, if $w_k\in S_{k,\XYm}$, then $w_k\in\odd{}{g_k(w_k)}$ by \ref{it:XY}.
 This implies $w_k\notin\odd{}{g_{k+1}(v)}$, and thus $w_k\notin S_{k+1}$.
 Similarly, if $w_k\in S_{k,\XZm} \cup S_{k,\YZm}$, then $w_k\in g_k(w_k)$ by \ref{it:XZ} or \ref{it:YZ}.
 This implies $w_k\notin g_{k+1}(v)$, and thus $w_k\notin S_{k+1}$.
 Hence, in each step we remove a minimal element from the set.

 Suppose there exists $w'\in S_{k+1}\setminus S_k$, then either $w'\in g_k(w_k)$ or $w'\in\odd{}{g_k(w_k)}$; in either case $w_k\prec w'$.
 In other words, we always remove a minimal element from the set and add only elements that come strictly later in the partial order.
 Therefore, the process terminates after $n\leq\abs{V}$ steps, at which point $S_n=\emptyset$,
 and the process requires $O(\abs{V}^2)$ operations at each step.
 The total complexity is therefore $O(\abs{V}^3)$.
The function $g'=g_n$ has the desired properties: (1) holds because we never modify the value of the function on inputs other than $v$, (2) and (3) hold because $S_n=\emptyset$, and (4) was shown to follow from Lemma~\ref{lem:successor-gflow}.
\end{proof}

Based on these lemmas, we can now show the focusing property.
These results state that correction sets can be simplified to only contain qubits measured in the \XY plane.
Moreover, side-effects of corrections (i.e.\ effects on qubits other than the one being corrected) never affect qubits measured in the \XY plane.

\begin{proposition}\label{prop:focused-gflow}
 Let $(G,I,O,\ld)$ be a labelled open graph which has gflow.
 Then $(G,I,O,\ld)$ has a maximally delayed gflow $(g,\prec)$ with the following properties for all $v\in V$:
 \begin{itemize}
  \item for all $w\in g(v)\cap\comp{O}$, either $v=w$ or $\ld(w)= \XYm$, and
  \item for all $w\in \odd{}{g(v)}\cap\comp{O}$, either $v=w$ or $\ld(w)\neq \XYm$.
 \end{itemize}
 This maximally delayed gflow can be constructed in a number of steps that is polynomial in the number of vertices in $G$.
\end{proposition}
\begin{proof}
 Let $(g_0,\prec)$ be a maximally delayed gflow of $(G,I,O,\ld)$.
 Set $n:=\abs{V}$ and consider the vertices in some order $v_1,\ldots,v_n$.
 For each $k=1,\ldots,n$, let $g_k$ be the function that results from applying Lemma~\ref{lem:focus-single-vertex} to the gflow $(g_{k-1},\prec)$ and the vertex $v_k$.
 Then $g_k$ satisfies the two properties for the vertex $v_k$.
 The function $g_k$ also equals $g_{k-1}$ on all inputs other than $v_k$, so in fact $g_k$ satisfies the two properties for all vertices $v_1,\ldots,v_k$.
 Thus, $g_n$ satisfies the two properties for all vertices.
 Moreover, the partial order does not change, so $(g_n,\prec)$ is as delayed as $(g_0,\prec)$; i.e.\ it is maximally delayed.
 Hence if $g:=g_n$, then $(g,\prec)$ has all the desired properties.
 The construction of each successive $g_{k+1}$ via  Lemma~\ref{lem:focus-single-vertex} takes $O(n^3)$ operations,
 which we perform at most $n$ times, giving a complexity of $O(n^4)$.
\end{proof}

The extended notion of focused gflow also allows us to prove another result which will be useful for the optimisation algorithm later.

First, note that if a \LOG{} has gflow, then the \LOG{} that results from deleting all vertices measured in the \XZ or \YZ planes still has gflow.

\begin{lemma}\label{lem:gflow_drop_gadgets}
 Suppose $(G,I,O,\ld)$ is a \LOG{} which has gflow.
 Let $(G',I,O,\ld')$ be the induced \LOG{} on the vertex set $V'=\{v\in V\mid v\in O \text{ or } \ld(v)=\XYm\}$.
 Then $(G',I,O,\ld')$ has gflow.
\end{lemma}
\begin{proof}
 Apply Lemma~\ref{lem:deletepreservegflow} to each vertex measured in the \XZ or \YZ plane one by one.
 Recall from Definition~\ref{defGFlow} that input vertices are measured in the \XY plane and so
 are not removed by this process.
\end{proof}

We can now show that in a \LOG{} which has gflow and which satisfies $\abs{I}=\abs{O}$, any internal \XY vertex must have more than one neighbour.\footnote{The condition $\abs{I}=\abs{O}$ is necessary: consider the \LOG{} $(G, \emptyset, \{o\}, \ld)$, where $G$ is the connected graph on two vertices $\{v,o\}$, and $\ld(v)=\XYm$.
Then $v$ is internal and has only a single neighbour, yet the \LOG{} has gflow with $g(v)=\{o\}$ and $v\prec o$.}

\begin{proposition}\label{prop:XY-neighbours}
 Let $(G,I,O,\ld)$ be a \LOG{} which has gflow and for which $\abs{I}=\abs{O}$.
 Suppose $v\in\comp{O}$ satisfies $\ld(v)=\XYm$.
 Then either $v\in I$ and $\abs{N_{G}(v)}\geq 1$, or $v\notin I$ and $\abs{N_{G}(v)}\geq 2$.
\end{proposition}
\begin{proof}
 Consider some $v\in\comp{O}$ such that $\ld(v)=\XYm$.
 Note that a vertex with no neighbours is in the even neighbourhood of any set of vertices.
 Therefore we must have $\abs{N_{G}(v)}\geq 1$, since $(G,I,O,\ld)$ has gflow and $v$ must be in the odd neighbourhood of its correction set by \ref{it:XY}.

 Now suppose for a contradiction that $v\notin I$ and $\abs{N_{G}(v)}=1$.
 Denote by $u$ the single neighbour of $v$.

 If the \LOG{} contains any vertices measured in the \XZ or \YZ planes, by Lemma~\ref{lem:gflow_drop_gadgets}, we can remove those vertices while preserving the property of having gflow.
 Since $\ld(v)=\XYm$, the removal process preserves~$v$.
 The new \LOG{} has gflow and $v$ cannot have gained any new neighbours, so $u$ must also be preserved by the argument of the first paragraph above.
 Thus, without loss of generality, we will assume that all non-outputs of $(G,I,O,\ld)$ are measured in the \XY plane.

 The \LOG{} $(G,I,O,\ld)$ has gflow and satisfies $\ld(w)=\XYm$ for all $w\in\comp{O}$, so by Theorem~\ref{thm:mhalla2} it has a focused gflow $(g,\prec)$.
 To satisfy the gflow condition \ref{it:XY} for $v$, that is, to satisfy $v\in\odd{G}{g(v)}$, we must have $u\in g(v)$.
 This then implies $v\prec u$ by \ref{it:g}.

 Since $\abs{I}=\abs{O}$, the focused gflow $(g,\prec)$ can be reversed in the sense of Corollary~\ref{cor:reverse_unitary_gflow}.
 Denote by $(G,O,I,\ld')$ the reversed \LOG{} (cf.\ Definition~\ref{def:reversed-LOG}) and by $(g',\prec')$ the corresponding reversed focused gflow.
 Since $v\notin I$, $v$ remains a non-output in the reversed graph, so it has a correction set.
 But $g'(v)=\{w\in\comp{O}\mid v\in g(w)\}$, so it cannot contain $u$ because $v\notin g(u)$ by $v\prec u$.
 Thus, $v\notin\odd{G}{g'(v)}$, contradicting \ref{it:XY}.

 Therefore, the initial assumption must be wrong, i.e.\ if $v\notin I$ then $\abs{N_{G}(v)}\geq 2$.
\end{proof}

\begin{remark}
 This implies that in any unitary MBQC-form \zxdiagram with gflow, any vertex measured in the \XY-plane has at least two incident wires (plus the wire leading to the measurement effect), since being an input vertex of the \LOG{} implies being connected to an input wire in the \zxdiagram.
\end{remark}

\section{Simplifying measurement patterns}\label{sec:simplifying}

In the previous section, we saw several ways in which \LOG{}s can be modified while preserving the existence of gflow. In this section, we will see how these modifications can be done on measurement patterns in a way that preserves the computation being performed.
The goal of the simplifications in this section is to reduce the number of qubits needed to implement the computation.
Since we are specifically interested in patterns with gflow, we will represent a pattern by a ZX-diagram in MBQC+LC form, which carries essentially the same information.

We show how local complementations affect a measurement pattern. The relation between local complementations and graph states was first considered in Ref.~\cite{NestMBQC}, and the combined action of a local complementation or pivot followed by a vertex deletion on a measurement pattern consisting of just XY-plane measurements was considered in Ref.~\cite{cliff-simp}. We here show the effect of a local complementation on a generic measurement pattern containing measurements in three possible planes.
We combine these results with ways to delete qubits measured in a Clifford angle. That qubits measured in a Clifford angle can be treated in special ways is well-known. For instance, Hein et al.\ showed how graph states are changed under Pauli measurements \cite{hein2004multiparty}.
More specific to MBQC, Browne at al.\ found gflow conditions specific to Pauli observables \cite{GFlow} and Di Matteo and Mosca showed that all Clifford measurements (and their corrections) in a pattern containing only XY-plane measurements can be done simultaneously, reducing the depth of a computation \cite{di2016parallelizing}.
Our results go beyond this in three ways. First, by explicitly showing how to find an equivalent measurement pattern where Clifford qubits are actually removed, instead of just `compressed' into a single layer. Second, by doing this for measurement patterns containing measurements in three planes instead of just one, and third, by showing how these operations preserve the existence of a gflow, and hence preserve deterministic realisability.
This last point was also considered in Ref.~\cite{cliff-simp}, restricted to XY-plane measurements, where they showed this property for the combination of a local complementation and a qubit deletion. This result actually becomes simpler for generic patterns, as the local complementation and deletion can be done separately.

Before we find qubit-removing rewrite rules however, we establish how local Cliffords in an MBQC+LC form diagram can be changed into measurements in Section~\ref{sec:local-Cliffords} and how local complementations affect a pattern in Section~\ref{sec:pattern-local-complementation}. We use local complementations to remove Clifford vertices from a pattern in Section~\ref{sec:removecliffordqubits}, and to change a pattern so that only two measurement planes are necessary in Section~\ref{sec:phasegadgetform}. Finally, in Section~\ref{sec:further-opt} we find some further simplifications that allow the removal of additional qubits.

\subsection{Transforming local Cliffords into measurements}\label{sec:local-Cliffords}

We used MBQC+LC diagrams as an extension of MBQC form diagrams. In this section we will see that we can always convert the local Clifford gates into measurements to turn the diagram into MBQC form.

\begin{lemma}\label{lem:SQU-to-MBQC-form}
 Any \zxdiagram $D$ which is in MBQC+LC form can be brought into MBQC form.
 Moreover, if the MBQC-form part of $D$ involves $n$ qubits, of which $p$ are inputs and $q$ are outputs, then the resulting MBQC-form diagram contains at most $n+2p+4q$ qubits.
\end{lemma}
\begin{proof}
 Any single-qubit Clifford unitary can be expressed as a composite of three phase shifts \cite[Lemma~3]{backens1}.
 Note that this result holds with either choice of colours, i.e.\ any single-qubit Clifford unitary can be expressed as \input{SQC-red.tikz} or \input{SQC-green.tikz}.

 Now, with the green-red-green version, for any Clifford operator on an input, we can `push' the final green phase shift through the graph state part onto the outgoing wire.
 There, it will either merge with the measurement effect or with the output Clifford unitary:
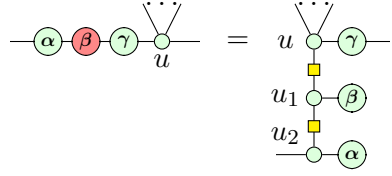
\ctikzfig{SQC-in-replacement}
 If $\gamma\in\{0,\pi\}$, merging the phase shift with a measurement effect may change the angle but not the phase label, e.g.\ if $\gamma=\pi$:
 \begin{center}
  \input{pivot-pi-phases-XY.tikz} \qquad \input{pivot-pi-phases-XZ.tikz} \qquad \input{pivot-pi-phases-YZ.tikz}
 \end{center}
 If $\gamma\in\{\frac\pi2,-\frac\pi2\}$, merging the phase shift with a measurement effect will flip the phase labels \XZ and \YZ, e.g.\ if $\gamma=-\frac\pi2$:
 \begin{center}
  \input{lc-N-XY.tikz} \qquad \input{lc-N-XZ.tikz} \qquad \input{lc-N-YZ.tikz}
 \end{center}
 Thus we need to add at most two new qubits to the MBQC-form part when removing a Clifford unitary on the input.

 For a Clifford unitary on the output, we have
  \ctikzfig{SQU-out-replacement}
 Thus we add at most four new qubits.

 Combining these properties, we find that rewriting to MBQC form adds at most $2p+4q$ new qubits to the pattern.
\end{proof}

\begin{proposition}
 Suppose $D$ is a \zxdiagram in MBQC+LC form and that its MBQC part has gflow.
 Let $D'$ be the \zxdiagram that results from bringing $D$ into MBQC form as in Lemma~\ref{lem:SQU-to-MBQC-form}.
 Then $D'$ has gflow.
\end{proposition}
\begin{proof}
 By applying Lemma~\ref{lem:SQU-to-MBQC-form} repeatedly, we can incorporate any local Clifford operators into the MBQC form part of the diagram.
 Lemmas~\ref{lem:gflow-add-output} and~\ref{lem:gflow-add-input} ensure that each step preserves the property of having gflow.
\end{proof}

\subsection{Local complementation and pivoting on patterns}\label{sec:pattern-local-complementation}

Lemmas~\ref{lem:ZX-lcomp} and~\ref{lem:ZX-pivot} showed how to apply a local complementation and pivot on a ZX-diagram by introducing some local Clifford spiders. In this section we will show how these rewrite rules can be used on MBQC+LC diagrams.

\begin{lemma}\label{lem:lc-MBQC-form-non-input}
 Let $D$ be an MBQC+LC diagram and suppose $u\in G(D)$ is not an input vertex.
 Then the diagram resulting from applying Lemma~\ref{lem:ZX-lcomp} on $u$ (\ie locally complementing), can be turned back into an MBQC+LC diagram $D'$ with $G(D')=G(D)\star u$. If $D$ had gflow, then $D'$ will also have gflow.
\end{lemma}
\begin{proof}
 Suppose $D$ is an MBQC+LC diagram, $\Gamma=(G,I,O,\ld)$ the corresponding \LOG, and $\alpha:\comp{O}\to[0,2\pi)$ gives the associated measurement angles.
 By assumption, $u\notin I$, so -- with the exception of the output wire or the edge to the measurement effect -- all edges incident on $u$ connect to neighbouring vertices in the graph.
 The input wires on the other qubits can be safely ignored.
 To get back an MBQC+LC diagram after Lemma~\ref{lem:ZX-lcomp} is applied to $u$, we only need to rewrite the measurement effects, and hence we need to construct new $\lambda'$ and $\alpha'$ for these measurement effects. We do that as follows.

 First of all, there are no changes to the measurement effects on vertices $v\not\in N(u)\cup\{u\}$, and hence for those vertices we have $\lambda'(v)=\lambda(v)$ and $\alpha'(v)=\alpha(v)$.

 The vertex $u$ gets a red $\frac\pi2$ phase from the application of Lemma~\ref{lem:ZX-lcomp}. If $u\in O$, then it has no associated measurement plane or angle. In this case, this red $\frac\pi2$ simply stays on the output wire, as allowed in an MBQC+LC diagram. When $u\notin O$, there are three possibilities, depending on $\ld(u)$:
 \begin{itemize}
  \item If $\ld(u)=\XYm$, then the new measurement effect is
   \ctikzfig{lc-u-XY}
   i.e.\ $\ld'(u)=\XZm$ and $\alpha'(u)=\frac{\pi}{2}-\alpha(u)$.
  \item If $\ld(u)=\XZm$, then the new measurement effect is
   \ctikzfig{lc-u-XZ}
   i.e.\ $\ld'(u)=\XYm$ and $\alpha'(u)=\alpha(u)-\frac{\pi}{2}$.
  \item If $\ld(u)=\YZm$, then the new measurement effect is
   \ctikzfig{lc-u-YZ}
   i.e.\ $\ld'(u)=\YZm$ and $\alpha'(u)=\alpha(u)+\frac{\pi}{2}$.
 \end{itemize}
 The vertices $v$ that are neighbours of $u$ get a green $-\frac\pi2$ phase. Again, if such a $v$ is an output, this phase can be put as a local Clifford on the output. If it is not an output, then there are also three possibilities depending on $\ld(v)$:
 \begin{itemize}
  \item If $\ld(v)=\XYm$, then the new measurement effect is
   \ctikzfig{lc-N-XY}
   i.e.\ $\ld'(v)=\XYm$ and $\alpha'(v)=\alpha(v)-\frac{\pi}{2}$.
  \item If $\ld(v)=\XZm$, then the new measurement effect is
   \ctikzfig{lc-N-XZ}
   i.e.\ $\ld'(v)=\YZm$ and $\alpha'(v)=\alpha(v)$.
  \item If $\ld(v)=\YZm$, then the new measurement effect is
   \ctikzfig{lc-N-YZ}
   i.e.\ $\ld'(v)=\XZm$ and $\alpha'(v)=-\alpha(v)$.
 \end{itemize}

 With these changes, we see that the resulting diagram $D'$ is indeed in MBQC+LC form. The underlying graph $G(D')$ results from the local complementation about $u$ of the original graph $G(D)$. Furthermore, the measurement planes changed in the same way as in Lemma~\ref{lem:lc_gflow}, and hence if $D$ had gflow, then $D'$ will also have gflow.
\end{proof}

\begin{proposition}\label{prop:MBQC-lc-MBQC}
 Let $D$ be an MBQC+LC diagram and suppose $u$ is an arbitrary vertex.
Then we can find an MBQC+LC diagram $D'$ such that $D'$ is equivalent to $D$ and $G(D)\star u \sse G(D')$.
 Furthermore, if $D$ has gflow, then $D'$ does too.
\end{proposition}
\begin{proof}
 If $u$ is not an input vertex, the result is immediate from Lemma~\ref{lem:lc-MBQC-form-non-input}.

 If instead $u$ is an input vertex, we modify $D$ by replacing the input wire incident on $u$ by an additional graph vertex $u'$ measured in the \XY-plane at angle 0, and a Hadamard unitary on the input wire:
 \ctikzfig{input-replacement}
 Throughout this process, the measurement effect on $u$ (if any) does not change, so it is left out of the above equation.
 In the resulting diagram $D''$, $u$ is no longer an input.
 It is straightforward to see that $D''$ is equivalent to $D$.
 Furthermore, $D''$ is an MBQC+LC diagram with $G(D)\sse G(D'')$, and $D''$ has gflow if $D$ does.
 Thus, the desired result follows by applying Lemma~\ref{lem:lc-MBQC-form-non-input} to $D''$ and noting that local complementations commute with subgraph relationships.
\end{proof}

A pivot is just a sequence of three local complementations.
Thus, the previous lemma already implies that when Lemma~\ref{lem:ZX-pivot} is applied to an MBQC+LC diagram the resulting diagram can also be brought back into MBQC+LC form. Nevertheless, it will be useful to explicitly write out how the measurement planes of the vertices change.

\begin{lemma}\label{lem:pivot-MBQC-form-non-input}
 Let $D$ be an MBQC+LC diagram and suppose $u$ and $v$ are neighbouring vertices in the graph state and are not input vertices of the underlying \LOG.
 Then the diagram resulting from applying Lemma~\ref{lem:ZX-pivot} to $u$ and $v$ (\ie a pivot about $u\sim v$) can be brought back into MBQC+LC form.
 The resulting \zxdiagram $D'$ satisfies $G(D') = G(D)\wedge uv$. If $D$ had gflow, then $D'$ will also have gflow.
\end{lemma}
\begin{proof}
  Suppose $\Gamma=(G,I,O,\ld)$ is the \LOG\ underlying $D$ and suppose $\alpha:\comp{O}\to[0,2\pi)$ gives the measurement angles.
  We will denote the measurement planes after pivoting by $\ld':\comp{O}\to\{\XYm,\XZm,\YZm\}$ and the measurement angles after pivoting by $\alpha':\comp{O}\to[0,2\pi)$.
  Let $a\in\{u,v\}$, then:
  \begin{itemize}
    \item If $a$ is an output, we consider the Hadamard resulting from the pivot operation as a Clifford operator on the output.
    \item If $\ld(a)=\XYm$ then $\ld'(a) = \YZm$ and if $\ld(a)=\YZm$ then $\ld'(a) = \XYm$:
   \ctikzfig{pivot-u-XY}
   In both cases, the measurement angle stays the same: $\alpha'(a) = \alpha(a)$.
   \item If $\ld(a)=\XZm$, then
   \ctikzfig{pivot-u-XZ}
   \ie $\ld'(a) = \XZm$ and $\alpha'(a) = \frac\pi2 - \alpha(a)$.
  \end{itemize}

  The only other changes are new green $\pi$ phases on vertices $w\in N(u)\cap N(v)$.
  For measured (i.e.\ non-output) vertices, these preserve the measurement plane and are absorbed into the measurement angle in all three cases:
   \begin{align*}
  (\ld'(w), \alpha'(w)) =
  \begin{cases}
  (\XYm, \alpha(w) + \pi) & \text{if } \ld(w) = \XYm \mspace{-1.5mu} \quad \input{pivot-pi-phases-XY.tikz} \\
  (\YZm, -\alpha(w)) & \text{if } \ld(w) = \YZm \quad  \input{pivot-pi-phases-YZ.tikz} \\
  (\XZm, -\alpha(w)) & \text{if } \ld(w) = \XZm \quad  \input{pivot-pi-phases-XZ.tikz}
  \end{cases}
  \end{align*}
  If instead $w$ is an output vertex, we consider the green $\pi$ phase shift as a Clifford operator on the output wire.
  The measurement planes and the graph change exactly like in Corollary~\ref{cor:pivot_gflow} and hence $D'$ has gflow when $D$ had gflow.
\end{proof}

\subsection{Removing Clifford vertices}\label{sec:removecliffordqubits}

In this section, we show that if a qubit is measured in one of the Pauli bases, i.e.\ at an angle which is an integer multiple of $\frac\pi2$, it can be removed from a pattern while preserving the semantics as well as the property of having gflow.

\begin{definition}\label{dfn:internal-boundary-Clifford}
  Let $D$ be a \zxdiagram in MBQC+LC form, with underlying \LOG\ $(G,I,O,\ld)$. Let $\alpha:\comp{O}\to [0,2\pi)$ be the corresponding set of measurement angles. We say a measured vertex $u\in G$ is \emph{Clifford} when $\alpha(u) = k\frac\pi2$ for some $k$.
\end{definition}

Our goal will be to remove as many internal Clifford vertices as possible.
We make a key observation for our simplification scheme: a \YZ-plane measurement with a $0$ or $\pi$ phase can be removed from the pattern by modifying its neighbours in a straightforward manner.

\begin{lemma}\label{lem:ZX-remove-YZ-Pauli}
Let $D$ be a ZX-diagram in MBQC+LC form with vertices $V$.
Suppose $u\in V$ is a non-input vertex measured in the \YZ or \XZ plane with an angle of $a\pi$ where $a=0$ or $a=1$. Then there is an equivalent diagram $D'$ with vertices $V\setminus \{u\}$. If $D$ has gflow, then $D'$ does as well.
\end{lemma}
\begin{proof}
  Using the axioms of the ZX-calculus, it is straightforward to show that:
  \ctikzfig{remove-YZ-measurement}
  These $a\pi$ phase shifts on the right-hand side can be absorbed into the measurement effects of the neighbouring vertices (or, for output vertices, considered as a local Clifford operator). Absorbing an $a\pi$ phase shift into a measurement effect does not change the measurement plane, only the angle. The resulting diagram $D'$ is then also in MBQC+LC form. Since $G(D')$ is simply $G(D)$ with a \YZ or \XZ plane vertex removed, $D'$ has gflow if $D$ had gflow by Lemma~\ref{lem:deletepreservegflow}.
\end{proof}

We can combine this observation with local complementation and pivoting to remove vertices measured in other planes or at other angles.
\begin{lemma}\label{lem:lc-simp}
  Let $D$ be a ZX-diagram in MBQC+LC form with vertices $V$. Suppose $u\in V$ is a non-input vertex measured in the \YZ or \XY plane with an angle of $\pm\frac\pi2$. Then there is an equivalent diagram $D'$ with vertices $V\setminus \{u\}$. If $D$ has gflow, then $D'$ does as well.
\end{lemma}
\begin{proof}
  We apply a local complementation about $u$ using Lemma~\ref{lem:ZX-lcomp} and reduce the diagram to MBQC+LC form with Lemma~\ref{lem:lc-MBQC-form-non-input}. By these lemmas, if the original diagram had gflow, this new diagram will also have gflow.
  As can be seen from Lemma~\ref{lem:lc-MBQC-form-non-input}, if $u$ was in the \XY plane, then it will be transformed to the \XZ plane and will have a measurement angle of $\frac\pi2 \mp\frac\pi2$. As a result its measurement angle is of the form $a\pi$ for $a\in\{0,1\}$.
  If instead it was in the \YZ plane, then it stays in the \YZ plane, but its angle is transformed to $\frac\pi2 \pm\frac\pi2$ in which case it will also be of the form $a\pi$ for $a\in\{0,1\}$.
  In both cases we can remove the vertex $u$ using Lemma~\ref{lem:ZX-remove-YZ-Pauli} while preserving semantics and the property of having gflow.
\end{proof}

\begin{lemma}\label{lem:pivot-simp}
  Let $D$ be a ZX-diagram in MBQC+LC form with vertices $V$, and let $u,v \in V$ be two non-input measured vertices which are neighbours.
  Suppose that either $\ld(u)=\XYm$ with $\alpha(u) = a\pi$ for $a\in \{0,1\}$ or $\ld(u) = \XZm$ with $\alpha(u) = (-1)^a\frac\pi2$.
  Then there is an equivalent diagram $D'$ with vertices $V\setminus \{u\}$. Moreover, if $D$ has gflow, then $D'$ also has gflow.
\end{lemma}
\begin{proof}
  We apply a pivot to $uv$ using Lemma~\ref{lem:ZX-pivot} and reduce the diagram to MBQC+LC form with Lemma~\ref{lem:pivot-MBQC-form-non-input}. If the original diagram had gflow, this new diagram will also have gflow.
  As can be seen from Lemma~\ref{lem:pivot-MBQC-form-non-input}, if $\lambda(u) = \XYm$ then $\lambda'(u) = \YZm$ with $\alpha'(u)=\alpha(u) = a\pi$. If instead we had $\lambda(u) = \XZm$ (and thus $\alpha(u) = (-1)^a\frac\pi2$), then $\lambda'(u) = \XZm$, but $\alpha'(u) = \frac\pi2 - \alpha(u) = \frac\pi2 - (-1)^a \frac\pi2 = a\pi$. In both cases, using Lemma~\ref{lem:ZX-remove-YZ-Pauli}, $u$ can be removed while preserving semantics and the existence of gflow.
\end{proof}

\begin{remark}
  The `graph-like' \zxdiagrams studied in Ref.~\cite{cliff-simp} are MBQC+LC form diagrams where every vertex is measured in the \XY plane.
  Our Lemmas~\ref{lem:lc-simp} and \ref{lem:pivot-simp} are generalisations of the work found in Ref.~\cite[Lemmas~5.2 and 5.3]{cliff-simp} and Ref.~\cite[(P2) and (P3)]{tcountpreprint}.
\end{remark}

Combining the three previous lemmas we can remove most non-input Clifford vertices. The exceptions are some non-input Clifford vertices which are only connected to input and output vertices. While it might not always be possible to remove such vertices, when the diagram has a gflow, we can find an equivalent smaller diagram:

\begin{lemma}\label{lem:removeboundaryPauli}
  Let $D$ be a ZX-diagram in MBQC+LC form with vertices $V$ that has a gflow. Let $u$ be a non-input measured vertex that is only connected to input and output vertices. Suppose that either $\ld(u)=\XYm$ with $\alpha(u) = a\pi$ for $a\in \{0,1\}$ or $\ld(u) = \XZm$ with $\alpha(u) = (-1)^a\frac\pi2$. Then there exists an equivalent diagram $D'$ which has gflow and has vertices $V\backslash\{u\}$.
\end{lemma}
\begin{proof}
  We prove the result for $\ld(u)=\XYm$ and $\alpha(u) = a\pi$. The other case is similar.

  We claim that $u$ is connected to at least one output that is not also an input. In order to obtain a contradiction suppose otherwise. The diagram then looks as follows:
  \ctikzfig{ZX-Pauli-projector}
  Here `LC' indicates that there are local Clifford operators on the inputs.
  Since $D$ has gflow, the entire diagram must be (proportional to) an isometry, and hence it must still be an isometry if we remove the local Cliffords on the inputs. But we then have the map
  \ctikzfig{ZX-Pauli-projector2}
  as the first operation in the diagram. This map is a projector, and it is not invertible. This is a contradiction, as the entire diagram cannot then be an isometry.

  Therefore, $u$ must be connected to some output vertex $v$, which is not an input. We can thus perform a pivot about $uv$. This adds a Hadamard operator after $v$, and changes the label of $u$ to \YZ. We can then remove $u$ using Lemma~\ref{lem:ZX-remove-YZ-Pauli}. As all these operations preserve gflow, the resulting diagram still has gflow.
\end{proof}

The following result about removing Pauli measurements (i.e.\ Clifford vertices) from patterns while preserving semantics is already contained in Ref.~\cite[Section III.A]{hein2004multiparty} (if outside the context of MBQC), and is also mentioned in Ref.~\cite{houshmand2018minimal}.
Nevertheless, we are the first to explicitly state the effects of this process on the measurement pattern and the gflow.

\begin{theorem}\label{thm:simplifiedZXdiagram}
  Let $D$ be a ZX-diagram in MBQC+LC form that has gflow. Then we can find an equivalent ZX-diagram $D'$ in MBQC+LC form, which also has gflow and which contains no non-input Clifford vertices.
  The algorithm uses a number of graph operations that is polynomial in the number of vertices of $D$.
\end{theorem}
\begin{proof}
  Starting with $D$ we simplify the diagram step by step using the following algorithm:
  \begin{enumerate}
  \item Using Lemma~\ref{lem:lc-simp} repeatedly, remove any non-input \YZ or \XY measured vertex with a $\pm \frac\pi2$ phase.
  \item Using Lemma~\ref{lem:ZX-remove-YZ-Pauli} repeatedly, remove any non-input vertex with measured in the \YZ or \XZ plane with angle $a\pi$.
  \item Using Lemma~\ref{lem:pivot-simp} repeatedly, remove any \XY vertex with an $a\pi$ phase and any \XZ vertex with a $\pm \frac\pi2$ phase that are connected to any other internal vertex. If any have been removed, go back to step 1.
  \item If there are non-input measured Clifford vertices that are only connected to boundary vertices, use Lemma~\ref{lem:removeboundaryPauli} to remove them. Then go back to step 1. Otherwise we are done.
\end{enumerate}
By construction there are no internal Clifford vertices left at the end. Every step preserves gflow, so the resulting diagram still has gflow.
As every step removes a vertex, this process terminates in at most $n$ steps, where $n$ is the number of vertices in $D$. Each of the steps possibly requires doing a pivot or local complementation requiring $O(n^2)$ elementary graph operations. Hence, the algorithm takes at most $O(n^3)$ elementary graph operations.
\end{proof}

\begin{theorem}\label{thm:simplifiedMBQCpattern}
  Suppose $(G,I,O,\ld,\alpha)$ is a uniformly deterministic MBQC pattern representing a unitary operation.
  Assume the pattern involves $q$ inputs and $n$ qubits measured at non-Clifford angles, i.e.\ $q:=\abs{I}$ and $n := \abs{\{u\in\comp{O} \mid \alpha(u)\neq k\frac\pi2 \text{ for any } k \in \mathbb{Z}\}}$.
  Then we can find a uniformly deterministic MBQC pattern that implements the same unitary and uses at most $(n+8q)$ measurements.
  This process finishes in a number of steps that is polynomial in the number of vertices of $G$.
\end{theorem}
\begin{proof}
  Let $D$ be the ZX-diagram in MBQC form from Lemma~\ref{lem:zx-equals-linear-map} that implements the same unitary as the MBQC pattern $\pat:=(G,I,O,\ld,\alpha)$.
  Since $\pat$ is uniformly deterministic, it has gflow, and hence $D$ also has gflow by Definition~\ref{dfn:zx-gflow}.
  Let $D'$ be the ZX-diagram in MBQC+LC form produced by Theorem~\ref{thm:simplifiedZXdiagram}.
  Since $D'$ has no internal Clifford vertices, its MBQC-form part can have at most $n$ internal vertices.
  It may still have boundary Clifford vertices, and by unitarity $\abs{O}=\abs{I}=q$, so the MBQC-form part contains at most $(n+2q)$ vertices.

  Denote by $D''$ the MBQC-form diagram produced by applying Lemma~\ref{lem:SQU-to-MBQC-form} to $D'$.
  Then $D''$ has at most $((n+2q)+6q)$ vertices in its MBQC-form part.

  We can construct a new pattern $\pat'$ from $D''$ using Lemma~\ref{lem:zx-to-pattern}.
  As $D''$ has gflow, $\pat'$ also has gflow, and hence is uniformly deterministic.
  The new pattern $\pat'$ involves at most $(n+8q)$ qubits.

  For the complexity of these operations:
  \begin{itemize}
    \item Constructing $D$ from $\pat$ takes $O(\abs{G})$ operations
    \item Constructing $D'$ from $D$ takes $O(\abs{G}^3)$ operations
    \item Constructing $D''$ from $D$ takes $O(\abs{G})$ operations
    \item Constructing $\pat'$ from $D''$ takes $O(\abs{G})$ operations
  \end{itemize}
  So the entire process uses $O(\abs{G}^3)$ graph operations.
\end{proof}

\subsection{Phase-gadget form}\label{sec:phasegadgetform}
Using the local complementation and pivoting rules of Section~\ref{sec:pattern-local-complementation} we can transform the geometry of MBQC+LC form diagrams so that they no longer contain any vertices measured in the \XZ plane, nor will \YZ vertices be connected.

\begin{definition}\label{def:phasegadgetform}
 An MBQC+LC diagram is in \emph{phase-gadget form} if
 \begin{itemize}
  \item there does not exist any $v\in\comp{O}$ such that $\ld(v) = \XZm$, and
  \item there does not exist any pair of neighbours $v,w\in\comp{O}$ such that $\ld(v)=\ld(w)=\YZm$.
 \end{itemize}
\end{definition}
The name `phase gadget' refers to a particular configuration of spiders in the \zxcalculus, which, in our setting, corresponds to spiders measured in the \YZ plane. Phase gadgets have been used particularly in the study of circuit optimisation~\cite{phaseGadgetSynth,tcountpreprint,pi4parity}.

In Section~\ref{sec:further-opt} we introduce another form,
called \emph{reduced} (Definition~\ref{def:reduced-form}),
which requires the pattern to be in phase-gadget form.

\begin{example}
  The following MBQC+LC diagram is in phase-gadget form.
  \ctikzfig{example-phase-gadget-form}
\end{example}

\begin{proposition}\label{prop:ZXtophasegadgetform}
Let $D$ be a ZX-diagram in MBQC+LC form with gflow.
Then we can find an equivalent ZX-diagram $D'$ in MBQC+LC form that has gflow and is in phase-gadget form.
This process takes a number of steps that is polynomial in the number of vertices of $D$.
\end{proposition}
\begin{proof}
Set $D_0:=D$ and iteratively construct the diagram $D_{k+1}$ based on $D_k$.

\begin{itemize}
\item
  If the diagram $D_k$ contains a pair of vertices $u \sim v$ that are both measured in the \YZ-plane:
  First, note that any input vertex $w$ has $\ld(w) = \XYm$, as
  otherwise $w \in g(w)$ contradicting the co-domain of $g$ as given in Definition~\ref{defGFlow}.
  Therefore $u,v \notin I$.
  Let $D_{k+1}$ be the diagram that results from pivoting about the edge $u \sim v$ (Lemma~\ref{lem:ZX-pivot})
  and then transforming to MBQC+LC form (Lemma~\ref{lem:pivot-MBQC-form-non-input}.)
  This changes the measurement plane for $u$ and $v$ from \YZ to \XY
  and it does not affect the measurement planes for any other vertices:
  \ctikzfig{rm-adj-red}

\item
  If there is no such connected pair but there is some vertex $u$ that is measured in the \XZ-plane:
  Note that $u$ cannot be an input vertex by the same reasoning as in the first subcase.
  Let $D_{k+1}$ be the diagram that results from applying a local complementation
  on $u$ and transforming back to MBQC+LC form (Lemmas~\ref{lem:ZX-lcomp} and \ref{lem:lc-MBQC-form-non-input}).
  \ctikzfig{rm-adj-red2}

  As can be seen from Lemma~\ref{lem:lc-MBQC-form-non-input},
  this process changes the measurement plane of $u$ from \XZ to \YZ
  and it does not affect the labels of any vertices that are measured in the \XY-plane.
\item
  If there is no such connected pair, nor any vertex that is measured in the \XZ-plane
  then $D_k$ is already in the desired form, so halt.
\end{itemize}

The number of vertices not measured in the \XY-plane decreases with each step,
and no vertices are added, so this process terminates in at most $n$ steps, where $n$ is the number of vertices in $D$.
Each step requires checking every pair of vertices,
or performing local complementation,
each of which have complexity $O(n^2)$, so the total complexity is $O(n^3)$.
Since a pivot is just a sequence of local complementations,
$D_{k+1}$ has gflow if $D_k$ had gflow
(Proposition~\ref{prop:MBQC-lc-MBQC}).
Finally every step preserves equivalence, so $D_{k+1}$ is equivalent to $D_k$.
\end{proof}

Proposition~\ref{prop:ZXtophasegadgetform} finds a phase-gadget form for an MBQC+LC diagram, but note that the phase-gadget form is not guaranteed to be unique.

\subsection{Further pattern optimisation}\label{sec:further-opt}
In Section~\ref{sec:removecliffordqubits} we saw that we can remove all non-input
Clifford qubits from a pattern while preserving determinism and the computation
implemented by the pattern, up to local Clifford transformations.
We will show in this section that it is also possible to remove certain qubits
measured in non-Clifford angles.

These measurement pattern rewrite rules, seen then as transformations of ZX-diagrams, were used in Ref.~\cite{tcountpreprint} to reduce the T-count of circuits. We will see how in our context they can be used to remove a qubit from a pattern, again while preserving determinism.
First of all, any internal \YZ vertex with just one neighbour can be fused with this neighbour, resulting in the removal of the \YZ vertex:

\begin{lemma}\label{lem:removeidvertex}
    Let $D$ be an MBQC+LC diagram with an interior vertex $u$ measured in the \YZ plane, and suppose it has a single neighbour $v$, which is measured in the \XY plane. Then there is an equivalent MBQC+LC diagram $D'$ with $G(D') = G(D)\setminus \{u\}$. If $D$ has gflow, then $D'$ also has gflow.
\end{lemma}
\begin{proof}
    We apply the following rewrite:
    \ctikzfig{id-simp-1}
    The resulting diagram is again an MBQC+LC diagram. The change to the \LOG\ comes down to deleting a YZ vertex. By Lemma~\ref{lem:deletepreservegflow} this preserves gflow.
\end{proof}

Note that, by Proposition~\ref{prop:XY-neighbours}, if the diagram has gflow and an equal numbers of inputs and outputs, then it has no internal \XY vertices with just one neighbour.
Thus, if the diagram is in phase-gadget form (cf.\ Definition~\ref{def:phasegadgetform}), the above lemma allows us to remove all internal vertices which have a single neighbour.

Our second rewrite rule allows us to also `fuse' \YZ vertices that have the same set of neighbours:

\begin{lemma}\label{lem:removepairedgadgets}
    Let $D$ be an MBQC+LC diagram with two distinct interior vertices $u$ and $v$, both measured in the YZ plane and with $N(u) = N(v)$. Then there is an equivalent diagram $D'$ with $G(D') = G(D)\setminus\{u\}$. If $D$ has gflow, then $D'$ also has gflow.
\end{lemma}
\begin{proof}
    We apply the following rewrite:
    \ctikzfig{gadget-simp}
    A straightforward sequence of \zxcalculus transformations shows this rewrite preserves semantics:
    \begin{equation*}
    \scalebox{0.9}{\input{gf-proof.tikz}}
    \end{equation*}
    The new diagram is still an MBQC+LC diagram, and the only change in the underlying \LOG\ is the deletion of a \YZ vertex. Hence, by Lemma~\ref{lem:deletepreservegflow}, this rewrite preserves gflow.
\end{proof}

The analogous result is not true for a pair of \XY vertices. However, when the diagram has gflow, such pairs cannot exist to begin with:
\begin{lemma}\label{lem:nopairedXYvertices}
    Let $G$ be a \LOG{} with gflow and distinct interior vertices $u$ and $v$ both measured in the \XY plane. Then $N(u) \neq N(v)$.
\end{lemma}
\begin{proof}
	Assume for a contradiction that $N(u)=N(v)$ and that the diagram has gflow. Note that, for any subset of vertices $S$, we have $u\in \odd{}{S} \iff v\in \odd{}{S}$. In particular, as $u\in \odd{}{g(u)}$ by \ref{it:XY}, we have $v\in \odd{}{g(u)}$ and thus $u\prec v$ by \ref{it:odd}. Yet, swapping $u$ and $v$ in the above argument, we also find $v\prec u$, a contradiction.
	Thus, if the diagram has gflow, distinct vertices $u$ and $v$ must have distinct neighbourhoods $N(u) \neq N(v)$.\end{proof}

We can now combine these rewrite rules with our previous results to get a more powerful rewrite strategy:

\begin{definition} \label{def:reduced-form}
    Let $D$ be an MBQC+LC diagram. We say $D$ is \emph{reduced} when:
    \begin{itemize}
        \item It is in phase-gadget form (see Definition~\ref{def:phasegadgetform}).
        \item It has no internal Clifford vertices.
        \item Every internal vertex has more than one neighbour.
        \item If two distinct vertices are measured in the same plane, they have different sets of neighbours.
    \end{itemize}
\end{definition}

\begin{theorem}\label{thm:optimisation}
    Let $D$ be an MBQC+LC diagram with gflow and equal numbers of inputs and outputs.
    Then we can find an equivalent diagram $D'$ that is reduced and has gflow.
    This process finishes in a number of steps that is polynomial in the number of vertices of $D$.
\end{theorem}
\begin{proof}
    Starting with $D$, we simplify the diagram step by step with the following algorithm:
    \begin{enumerate}
        \item Apply Theorem~\ref{thm:simplifiedZXdiagram} to remove all interior Clifford vertices.
        \item Apply Proposition~\ref{prop:ZXtophasegadgetform} to bring the diagram into phase-gadget form. Then every vertex is of type \YZ or \XY, and the \YZ vertices are only connected to \XY vertices.
        \item Apply Lemma~\ref{lem:removeidvertex} to remove any YZ vertex that has a single neighbour.
        \item Apply apply Lemma~\ref{lem:removepairedgadgets} to merge any pair of \YZ vertices that have the same set of neighbours.
        \item If the application of these lemmas resulted in any new internal Clifford vertices, go back to step 1. Otherwise we are done.
    \end{enumerate}
    Each of the steps preserves gflow, and hence at every stage of the algorithm the diagram has gflow.
    By construction,
    when the algorithm has terminated,
    every vertex is now of type \YZ or \XY, and \YZ vertices are only connected to \XY vertices.
    Furthermore, every \YZ vertex must have more than one neighbour and have a different set of neighbours than any other \YZ vertex.
    This is also true for the \XY vertices by the existence of gflow and the requirement that the number of inputs match the number of outputs (using Lemma~\ref{lem:nopairedXYvertices} and Proposition~\ref{prop:XY-neighbours}).
    Hence, the resulting diagram has all the properties needed for it to be reduced.

    To show this process terminates consider the lexicographic order:
    \begin{itemize}
      \item Number of vertices in the diagram
      \item Number of vertices measured in the XZ or YZ planes
      \item Number of vertices measured in the XY plane
    \end{itemize}
    The result of each step of the algorithm on this order is:
    \begin{itemize}
      \item Applying Theorem~\ref{thm:simplifiedZXdiagram} reduces the number of vertices in the diagram,
      while possibly increasing the number of vertices in any given plane.
      \item Applying Proposition~\ref{prop:ZXtophasegadgetform} reduces the number of vertices in the XZ or YZ planes,
      while possibly increasing the number of vertices in the XY plane.
      \item Applying Lemmas~\ref{lem:removeidvertex} and \ref{lem:removepairedgadgets}
      reduces the number of vertices in the diagram,
      and the number of vertices measured in the YZ plane.
    \end{itemize}
    Therefore each step in the algorithm reduces our order, so the process terminates.
    Writing $n$ for the number of vertices in $D$
    we see that the algorithmic loop can be called at most $n$ times (since we remove vertices each iteration),
    and each of the steps in the loop take at most $O(n^3)$ operations,
    giving a total complexity of $O(n^4)$.
\end{proof}

\begin{remark}
    The algorithm described above uses the same idea as that described in Ref.~\cite{tcountpreprint}. But while they describe the procedure in terms of modifying a graph-like \zxdiagram, we describe it for MBQC+LC diagrams, a more general class of diagrams. Furthermore, we prove that the procedure preserves the existence of gflow. The existence of gflow is used in the next section to show how to recover a circuit from an MBQC+LC diagram.
\end{remark}

\section{Circuit extraction}\label{sec:circuitextract}

In this section we will see that we can extract a circuit from a measurement pattern whose corresponding \LOG\ has a gflow.

There are several results known about extracting circuits from measurement patterns. In the original paper defining the flow condition for measurement patterns~\cite{Danos2006Determinism-in-}, an ancilla-free circuit extraction algorithm for such patterns was presented. A number of procedures have also been found to extract circuits containing ancillae from XY-plane patterns~\cite{di2016parallelizing,daSilva2013compact,BKP10}.
In Ref.~\cite{duncan2010rewriting} an algorithm was sketched to extract an ancilla-free circuit from a XY-plane pattern with gflow. A version of this algorithm was further worked out in~\cite{miyazaki2015analysis}.
Independently, an algorithm for the same task was found in Ref.~\cite{cliff-simp} that uses the ZX-calculus to derive correctness.
Our algorithm modifies that of Ref.~\cite{cliff-simp} so that it can handle measurements in multiples planes (and not just the \XY plane).
While the algorithm itself is a relatively straightforward extension of that of Ref.~\cite{cliff-simp}, its proof of correctness relies on the new notion of focused gflow that we found in Section~\ref{sec:focusing-extended-gflow} and several other lemmas regarding extended gflow that are not obvious.

Instead of describing the algorithm for measurement patterns, we describe it for the more convenient form of MBQC+LC diagrams.
The general idea is that we modify the diagram vertex by vertex to bring it closer and closer to resembling a circuit. We start at the outputs of the diagram and work our way to the inputs. The gflow informs the choice of which vertex is next in line to be `extracted' (specifically, this will always be a vertex maximal in the gflow partial order).
By applying various transformations to the diagram, we change it so that the targeted vertex can easily be pulled out of the MBQC-form part of the diagram and into the circuit-like part. The remaining MBQC-form diagram is then one vertex smaller. Since all the transformations preserve gflow, we can then repeat the procedure on this smaller diagram until we are finished.

Before we explain the extraction algorithm in detail in Section~\ref{sec:generalextractalgorithm}, we state some relevant lemmas.

\begin{lemma}\label{lem:cnotgflow}
  The following equation holds:
  \begin{equation}
  \input{cnot-pivot.tikz}
  \end{equation}
  where $M$ is the biadjacency matrix of the output vertices to the neighbours of $D$, and $M^\prime$ is the matrix produced from $M$ by adding row~1 to row~2, modulo~2. If the full diagram on the LHS has gflow, then so does the RHS.
\end{lemma}
\begin{proof}
  The equality is proved in Ref.~\cite[Proposition~6.2]{cliff-simp}. There it is also shown that this preserves gflow when all measurements are in the \XY plane, but the same proof works when measurement in all three planes are present.
 \end{proof}

\begin{lemma}\label{lem:remove-output-edges-preserves-gflow}
	Suppose $(G,I,O,\lambda)$ is a \LOG{} with gflow.
	Let $G'$ be the graph containing the same vertices as $G$ and the same edges except those for which both endpoints are output vertices.
	Formally, if $G=(V,E)$, then $G'=(V,E')$, where
	\[
	 E' = \{v\sim w \in E\mid v\in\comp{O} \text{ or } w\in\comp{O}\}.
	\]
	Then $(G',I,O,\lambda)$ also has gflow.
\end{lemma}
\begin{proof}
	We claim that if $(g,\prec)$ is a gflow for $G$, then it is also a gflow for $G'$. Note that $\odd{G'}{g(v)}\cap \comp{O} = \odd{G}{g(v)}\cap \comp{O}$ as the only changes to neighbourhoods are among the output vertices. It is thus easily checked that all properties of Definition~\ref{defGFlow} remain satisfied.
\end{proof}

\begin{lemma}\label{lem:output-neighbours-are-XY}
	Let $(G,I,O,\ld)$ be a labelled open graph with a gflow and the same number of inputs as outputs: $\lvert I\rvert = \lvert O\rvert$.
  Let $v\in O\cap \comp{I}$ be an output which is not an input.
  Suppose $v$ has a unique neighbour $u\in\comp{O}$. Then $\ld(u)=\XYm$.
 \end{lemma}
\begin{proof}
	Suppose, working towards a contradiction, that $\ld(u) \neq \XYm$.
  Form the \LOG\ $(G',I,O,\ld')$ by removing from $G$ all vertices $w$
  such that $w \in \comp{O}$ and $\ld(w) \neq \XYm$, restricting $\ld$ accordingly.
  By Lemma~\ref{lem:gflow_drop_gadgets} the \LOG\ $(G',I,O,\ld')$ also has a gflow.
  Note that $G'$ does contain $v$, which is still an output vertex in $G'$, but does not contain $u$,
  and hence $v$ has no neighbours in $G'$.
  By Theorem~\ref{thm:mhalla2}, $G'$ has a focused gflow, and
  because $G'$ has the same number of inputs as outputs, its reversed graph also has a gflow $(g,\prec)$ by Corollary~\ref{cor:reverse_unitary_gflow}.
  In this reversed graph $v$ is an input and, since it is not an output, it is measured in the \XY plane.
  It therefore has a correction set $g(v)$ so that $v\in \odd{}{g(v)}$.
  But because $v$ has no neighbours, this is a contradiction.
  We conclude that indeed $\lambda(u)=\XYm$.
\end{proof}

\noindent For any set $A\sse V$, let $N_G(A) = \bigcup_{v\in A} N_G(v)$.
Recall the partition of vertices according to the partial order of the gflow into sets $V_k^\prec$, which is introduced in Definition~\ref{defVk}.

\begin{lemma}\label{lem:maxdelayednotempty}
 Let $(G,I,O,\ld)$ be a labelled open graph in phase-gadget form, which furthermore satisfies $\comp{O}\neq\emptyset$.
 Suppose $(G,I,O,\ld)$ has a gflow.
 Then the maximally delayed gflow, $(g,\prec)$, constructed in Proposition~\ref{prop:focused-gflow}
  exists and moreover $N_G(V_1^\prec)\cap O \neq \emptyset$, \ie the gflow has the property that,
  among the non-output vertices,
  there is a vertex which is maximal with respect to the gflow order and also connected to an output vertex.
\end{lemma}
\begin{proof}
    By Proposition~\ref{prop:focused-gflow}, there exists a maximally delayed gflow of $(G,I,O,\ld)$ such that
    no element of a correction set (other than possibly the vertex being corrected) is measured in the \YZ plane.

    Since the open graph does not consist solely of outputs, the set $V_1^\prec$ (as defined in Definition~\ref{defVk}) is non-empty, so the following arguments are non-trivial.
    For any $v\in V_1^\prec$ we must have $g(v)\subseteq O\cup \{v\}$.
    Now if there is a $v\in V_1^\prec$ with $\ld(v) =  \XYm$, then $v\in\odd{}{g(v)}$.
    There are no self-loops, hence this $v$ must be connected to at least one output, and we are done.
    As the graph is in phase-gadget form, there are no vertices labelled \XZ and hence from now on assume that $\ld(v)=\YZm$ for all $v\in V_1^\prec$.
    We distinguish two cases.
    \begin{enumerate}
     \item If $V_2^\prec = \emptyset$, then the only non-output vertices are in $V_1^\prec$.
    Now, any connected component of the graph $G$ must contain an input or an output.
    The vertices in $V_1^\prec$ are all labelled \YZ and thus appear in their own correction sets; this means they cannot be inputs because inputs do not appear in correction sets.
    The vertices in $V_1^\prec$ are not outputs either, so each of them must have at least one neighbour.
    Yet the \LOG{} is in phase-gadget form.
    This implies that two vertices both labelled \YZ cannot be adjacent, and all vertices in $V_1^\prec$ are labelled \YZ.
    Thus any vertex $v\in V_1^\prec$ must have a neighbour in $O$, and we are done.

     \item So now assume there is some vertex $w\in V_2^\prec$.
    Then, regardless of $\ld(w)$, we have $g(w) \sse V_1^\prec\cup O\cup\{w\}$ and $\odd{}{g(w)} \sse V_1^\prec\cup O\cup\{w\}$.
    We distinguish three subcases according to whether one of $g(w)$ or $\odd{}{g(w)}$ intersects $V_1^\prec$.
    \begin{itemize}
     \item Suppose $g(w)\cap V_1^\prec = \odd{}{g(w)}\cap V_1^\prec = \emptyset$, i.e.\ $g(w) \sse O\cup\{w\}$ and $\odd{}{g(w)} \sse O\cup\{w\}$.
     Let $\prec' = \prec \setminus \{(w,u): u\in V_1^\prec\}$.
     Then $(g,\prec')$ is a gflow: dropping the given inequalities from the partial order does not affect the gflow properties since $u\in V_1^\prec$ implies $w\notin g(u)$ and $w\notin \odd{}{g(u)}$.
     Furthermore, $(g,\prec')$ is more delayed than $(g,\prec)$ because $w$ (and potentially some of its predecessors) moves to an earlier layer, contradicting the assumption that $(g,\prec)$ is maximally delayed.
     Hence this case cannot happen.
     \item Suppose $g(w)\cap V_1^\prec \neq \emptyset$, then there exists a \YZ vertex in the correction set of $w$ since all elements of $V_1^\prec$ are measured in the \YZ plane.
     But our gflow satisfies the properties of Proposition~\ref{prop:focused-gflow}, and hence this cannot happen.
     \item Suppose $\odd{}{g(w)}\cap V_1^\prec \neq \emptyset$ and $g(w)\cap V_1^\prec = \emptyset$, then there is a $v\in V_1^\prec$ such that $v\in \odd{}{g(w)}$.
     There are two further subcases.
     \begin{itemize}
      \item If $\ld(w)=\XYm$, we have $w\not\in g(w)$ and hence $g(w)\sse O$ so that there must be some $o\in O$ that is connected to $v$ and we are done.
      \item Otherwise, if $\ld(w)=\YZm$, then $w\in g(w)$.
      Yet both $v$ and $w$ are measured in the \YZ plane, so they are not neighbours, and hence there still must be an $o\in O$ that is connected to $v$ to have $v\in \odd{}{g(w)}$.
     \end{itemize}
    \end{itemize}
    \end{enumerate}
    Thus, the gflow $(g,\prec)$ has the desired property in all possible cases.
\end{proof}

\subsection{General description of the algorithm}\label{sec:generalextractalgorithm}

We first walk through a high-level description of how to extract a circuit from a diagram in MBQC+LC form with gflow, explaining why every step works. After that, we present a more practical algorithm in Section~\ref{s:more-practical}. As we wish the output to be a unitary circuit, we will assume that the diagram has an equal number of inputs and outputs.

The process will be to make sequential changes to the \zxdiagram that make the diagram look progressively more like a circuit. During the process, there will be a `frontier': a set of green spiders such that everything to their right looks like a circuit, while everything to their left (and including the frontier vertices themselves) is an MBQC+LC form diagram equipped with a gflow.
We will refer to the MBQC-form diagram on the left as the \emph{unextracted} part of the diagram, and to the circuit on the right as the \emph{extracted} part of the diagram.
For example:
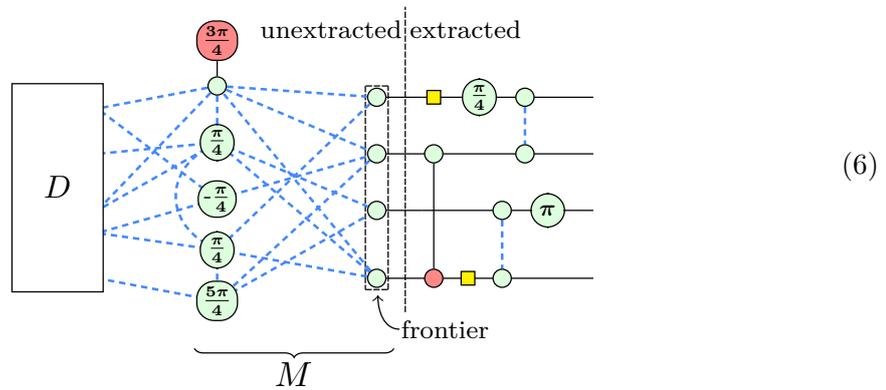
\begin{equation}\label{ex:frontier-example}
\scalebox{1.2}{\input{frontier-example.tikz}}
\end{equation}
In this diagram, we have merged the \XY measurement effects with their respective vertices, in order to present a tidier picture.
The matrix $M$ is the biadjacency matrix between the vertices on the frontier and all their neighbours to the left of the frontier.
For the purposes of the algorithm below, we consider the extracted circuit as no longer being part of the diagram, and hence the frontier vertices are the outputs of the \LOG\ of the unextracted diagram.

\textbf{Step 0}: First, we transform the pattern into phase-gadget form using Proposition~\ref{prop:ZXtophasegadgetform}, ensuring that all vertices are measured in the \XY or \YZ planes, and that vertices measured in the \YZ plane are only connected to vertices measured in the \XY plane.
This can be done in polynomial time, and preserves the interpretation of the diagram. Furthermore, the resulting diagram still has gflow.

\textbf{Step 1}: We unfuse any connection between the frontier vertices as a CZ gate into the extracted circuit, and we consider any local Clifford operator on the frontier vertices as part of the extracted circuit. For example:
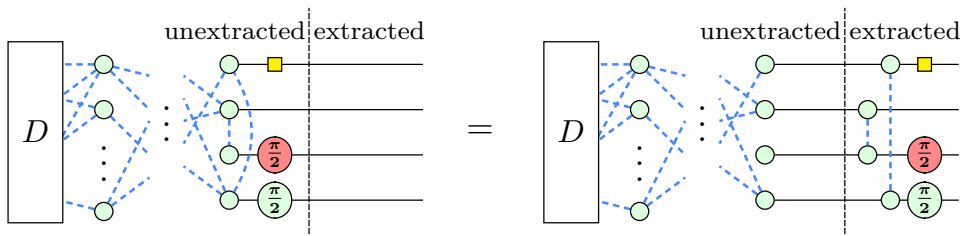
\[\scalebox{1.2}{\input{example-unfuse-gates.tikz}}\]
This process changes the unextracted diagram in two ways: by removing local Clifford operators and by removing connections among the frontier vertices.
The former does not affect the underlying \LOG{} and the latter preserves gflow by Lemma~\ref{lem:remove-output-edges-preserves-gflow}.
Thus, the unextracted diagram continues to be in MBQC form and it continues to have gflow.
If the only unextracted vertices are on the frontier, go to step~5, otherwise continue to step~2.

\textbf{Step 2}: The unextracted diagram is in phase-gadget form and has gflow.
Thus, by Lemma~\ref{lem:maxdelayednotempty}, it has a maximally delayed gflow $(g,\prec)$ such that $N_G(V_1^\prec)\cap O \neq \emptyset$, where $V_1^\prec$ is the `most delayed' layer before the frontier vertices, which are the outputs of the \LOG\ (see Definition~\ref{defVk}).
Such a gflow can be efficiently determined by first finding any maximally delayed gflow using the algorithm of Theorem~\ref{thmGFlowAlgo} and then following the procedure outlined in Proposition~\ref{prop:focused-gflow}.

Now, if any of the vertices in $V_1^\prec$ are labelled \XY, pick one of these vertices and go to step~3. Otherwise, all the maximal non-output vertices (with respect to $\prec$) must have label \YZ; go to step~4.

\textbf{Step 3}: We have a maximal non-output vertex $v$ labelled \XY, which we want to extract. Since it is maximal in $\prec$, we know that $g(v)\sse O$ by Definition~\ref{defVk}.
As the gflow is maximally delayed, we have $\odd{}{g(v)}\cap \comp O = \{v\}$.
We now follow the strategy used in Ref.~\cite{cliff-simp} for the `\XY-plane only' case, illustrating it with an example.
Consider the following diagram, in which the vertex $v$ and its correction set $g(v)$ are indicated:
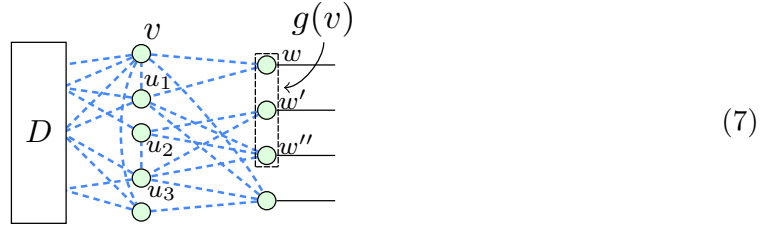
\begin{equation}\label{eq:example-extracted-vertex}
\scalebox{1.2}{\input{example-extracted-vertex.tikz}}
\end{equation}
For clarity, we are ignoring the measurement effects on the left-hand-side spiders, and we are not showing any frontier vertices that are inputs (although note that by definition of a gflow, the vertices of $g(v)$ cannot be inputs).
In the above example, the biadjacency matrix of the bipartite graph between the vertices of $g(v)$ on the one hand, and their neighbours in the unextracted part on the other hand, is
\begin{equation}\label{eq:biadjacency-example}
	\input{example-matrix.tikz}
\end{equation}
where the rows correspond to vertices of $g(v)$, and vertices are ordered top-to-bottom. We do not include the bottom-most frontier vertex in the biadjacency matrix, as it is not part of $g(v)$, and we do not include the bottom left spider, as it is not connected to any vertex in $g(v)$.

The property that $\odd{}{g(v)}\cap \comp O = \{v\}$ now corresponds precisely to the following property of the matrix: if we sum up all the rows of this biadjacency matrix modulo 2, the resulting row vector contains a single 1 corresponding to the vertex $v$ and zeroes everywhere else.
It is straightforward to see that this is indeed the case for the matrix of Eq.~\eqref{eq:biadjacency-example}.

Now pick any frontier vertex $w\in g(v)$.
Lemma~\ref{lem:cnotgflow} shows that the application of a CNOT to two outputs corresponds to a row operation on the biadjacency matrix, which adds the row corresponding to the target to the row corresponding to the control. Hence if, for each $w'\in g(v)\setminus\{w\}$, we apply a CNOT with control and target on the output wires of $w$ and $w'$, the effect on the biadjacency matrix is to add all the other rows of the vertices of $g(v)$ to that of $w$:
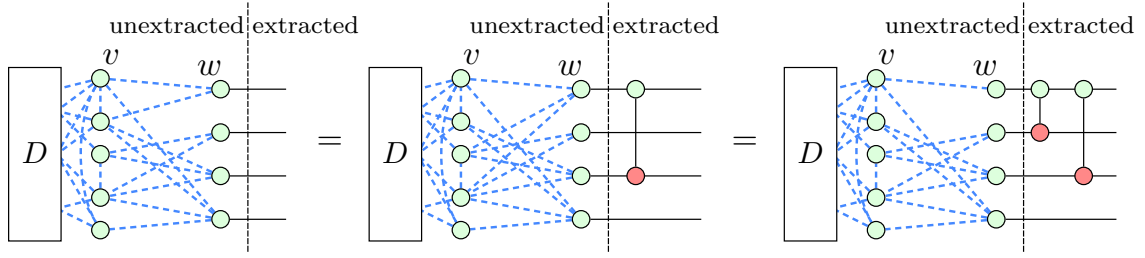
\[\scalebox{1.15}{\input{example-extracted-vertex-cnots.tikz}}\]
As a result, $w$ is now only connected to $v$, but $v$ may still be connected to other vertices in $O\setminus g(v)$. For each such vertex $u$, applying a CNOT with control $u$ and target $w$ removes the connection between $u$ and $v$:
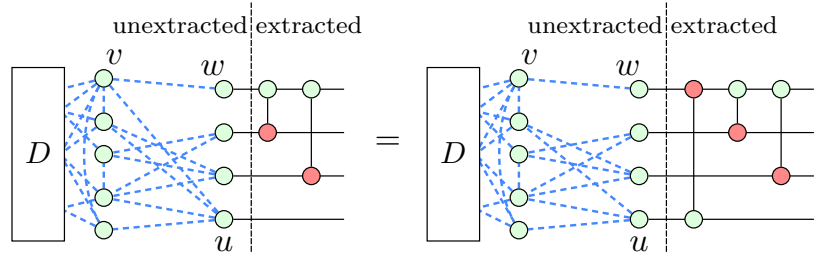
\[\scalebox{1.15}{\input{example-extracted-vertex-cnots2.tikz}}\]
Now we can move $v$ to the frontier by removing $w$ from the diagram, adding a Hadamard to the circuit (this comes from the Hadamard edge between $v$ and $w$), adding the measurement angle of $v$ to the circuit as a Z-phase gate, and adding $v$ to the set of outputs of the graph (i.e.\ the frontier):
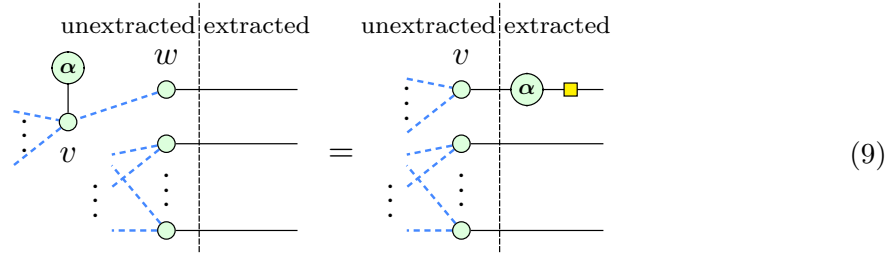
\begin{equation}\label{eq:extract-vertex}
\scalebox{1.15}{\input{extract-vertex.tikz}}
\end{equation}
On the underlying \LOG\ this corresponds to removing $w$ and adding $v$ to the list of outputs. We need to check that this preserves the existence of a gflow. The only change we need to make is that for all $v'\neq v$ with $w\in g(v')$ we set $g'(v') = g(v')\backslash\{w\}$. As $w$'s only neighbour is $v$, removing $w$ from $g(v')$ only toggles whether $v\in\odd{}{g'(v')}$. Since $v$ is a part of the outputs in the new \LOG, this preserves all the properties of being a gflow.

Now that the vertex $w$ has been removed, the number of vertices in the unextracted part of the diagram is reduced by 1. We now go back to step 1.

\textbf{Step 4}: All the maximal vertices are labelled \YZ. Since we chose our gflow according to Lemma~\ref{lem:maxdelayednotempty}, we know that at least one of these vertices is connected to an output, and hence a frontier vertex. Pick such a vertex $v$, and pick a $w\in O\cap N_G(v)$ (this set is non-empty). Pivot about $vw$ using Lemma~\ref{lem:ZX-pivot} and reduce the resulting diagram to MBQC form with Lemma~\ref{lem:pivot-MBQC-form-non-input}. Afterwards, $v$ has label \XY and $w$ has a new Hadamard gate on its output wire (which will be dealt with in the next iteration of step 1).

We have changed one vertex label in the unextracted part of the diagram from \YZ to \XY. Since no step introduces new \YZ vertices, step~4 can only happen as many times as there are \YZ vertices at the beginning. Go back to step 1.

\textbf{Step 5:} At this point, there are no unextracted vertices other than the frontier vertices, all of which have arity 2 and can be removed using rule $(\bm{i1})$ of Figure~\ref{fig:zx-rules}.
Yet the remaining frontier vertices might be connected to the inputs in some permuted manner and the inputs might carry some local Cliffords:
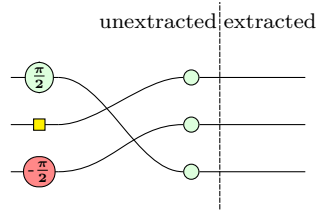
\ctikzfig{example-permutation}
This is easily taken care of by decomposing the permutation into a series of SWAP gates, at which point the entire diagram is in circuit form.

\textbf{Correctness and efficiency:} Since step 3 removes a vertex from the unextracted diagram, and step 4 changes a measurement plane from \YZ to \XY (and no step changes measurement planes in the other direction), this algorithm reduces the lexicographic order ($\#$ unextracted vertices, $\#$ unextracted vertices with $\ld = YZ$) each time we repeat an iteration of the process, so that the algorithm terminates. Each step described above takes a number of graph-operations polynomial in the number of unextracted vertices,
and therefore this entire algorithm takes a number of steps polynomial in the number of vertices.
All steps correspond to ZX-diagram rewrites, so the resulting diagram is a circuit that implements the same unitary as the pattern we started with.

\subsection{A more practical algorithm}\label{s:more-practical}

Now we know that the algorithm above is correct and will always terminate, we can take a couple of short-cuts that will make it more efficient.

In step 2, instead of using the gflow to find a maximal vertex, we do the following: Write down the biadjacency matrix of the bipartite graph consisting of frontier vertices on one side and all their neighbours on the other side. For example, the Diagram~\eqref{eq:example-extracted-vertex} would give the matrix:
\begin{equation}\label{eq:matrix2}
	\begin{pmatrix}
		1&1&0&0&0\\
		0&0&1&1&0\\
		0&1&1&1&0\\
		1&1&0&1&1
	\end{pmatrix}
\end{equation}
Now perform a full Gaussian elimination on this $\mathbb{Z}_2$ matrix. In the above case, this results in the matrix:
\begin{equation}\label{eq:matrix_after_elim}
	\begin{pmatrix}
		1&0&0&0&0\\
		0&1&0&0&0\\
		0&0&1&0&1\\
		0&0&0&1&1
	\end{pmatrix}
\end{equation}
Any row in this matrix containing a single 1 corresponds to an output vertex with a single neighbour. By Lemma~\ref{lem:output-neighbours-are-XY}, this neighbour is of type \XY. As an example, in the matrix in Eq.~\eqref{eq:matrix_after_elim}, the first row has a single 1 in the first column, and hence the top-left spider of Diagram~\eqref{eq:example-extracted-vertex} is the unique \XY neighbour to the first output. Similarly, the second row has a single 1, appearing in column 2, and hence the second spider from the top on the left in Diagram~\eqref{eq:example-extracted-vertex} is the unique neighbour to the second output.

If we found at least one row with a single 1 with this method, we implement the row operations corresponding to the Gaussian elimination procedure as a set of CNOT gates using Lemma~\ref{lem:cnotgflow}. Doing this with Diagram~\eqref{eq:example-extracted-vertex} gives:
\begin{equation}
	\scalebox{1.15}{\input{example-extracted-gauss.tikz}}
\end{equation}

We see that every row which had a single 1 now corresponds to a frontier spider with a single neighbour, and hence we can extract vertices using the technique of Eq.~\eqref{eq:extract-vertex}:
\begin{equation}\label{eq:example-extracted-3}
	\scalebox{1.15}{\input{example-extracted-3.tikz}}
\end{equation}

As we now extract multiple vertices at a time, there could be connections between the new frontier vertices (for instance between the top two frontier spiders in Eq.~\eqref{eq:example-extracted-3}). These are taken care of in the next iteration of step 1, turning those into CZ gates.

If the Gaussian elimination does not reveal a row with a single 1, then we are in the situation of step 4. We perform pivots involving a vertex with label \YZ and an adjacent frontier vertex until there is no vertex with a label \YZ which is connected to a frontier vertex. We then go back to step 1.

With these short-cuts, it becomes clear that we do not need an explicitly calculated gflow in order to extract a circuit. The fact that there is a gflow is only used to argue that the algorithm is indeed correct and will always terminate. Pseudocode for this algorithm can be found in Appendix~\ref{sec:pseudocode}.

With the results of this section we have then established the following theorem.
\begin{theorem}\label{thm:extraction-algorithm}
    Let $\pat$ be a measurement pattern with $n$ inputs and outputs containing a total of $k$ qubits, and whose corresponding \LOG\ has a gflow. Then there is an algorithm running in time $O(n^2k^2 + k^3)$ that converts $\pat$ into an equivalent $n$-qubit circuit that contains no ancillae.
 \end{theorem}
 \begin{proof}
The runtime for the extraction algorithm is dominated by Gaussian elimination of the biadjacency matrices which has complexity $O(n^2m)$, where $n$ is the number of rows, corresponding to the number of outputs, and $m$ is the number of neighbours these output vertices are connected to. In principle $m$ could be as large as the number of vertices in the graph and hence could be as large as $k$ (although in practice it will be much smaller than that).
In the worst case, performing a pivot operation also requires toggling the connectivity of almost the entire graph, which requires $k^2$ elementary graph operations.
Since we might have to apply a pivot and a Gaussian elimination process for every vertex in the graph, the complexity for the entire algorithm is bounded above by $O(k(n^2k + k^2)) = O(n^2k^2 + k^3)$.
\end{proof}

Note that if $k\geq O(n^2)$, which will be the case for most useful computation, the bound becomes $O(k^3)$. In practice however we would not expect to see this worst-case complexity as it would only be attained if everything is almost entirely fully connected all the time. This does not seem possible because the pivots and Gaussian elimination always toggle connectivity, and hence a highly connected graph in one step will become less connected in the following step.

\section{Conclusions and Future Work}\label{sec:conclusion}

We have given an algorithm which extracts a circuit from any measurement pattern whose underlying \LOG\ has extended gflow.
This is the first algorithm which works for patterns with measurements in multiple planes, and does not use ancillae.
Simultaneously, it is the most general known algorithm for extracting quantum circuits from \zxcalculus diagrams.

We have also developed a set of rewrite rules for measurement patterns containing measurements in all three planes.
For each of these rewrite rules, we have established the corresponding transformations of the extended gflow.
The rewrite rules can be used to reduce the number of qubits in a measurement pattern, in particular eliminating all qubits measured in a Pauli basis.
Additionally, we have generalised the notions of focused gflow and maximally delayed gflow to \LOG{}s with measurements in multiple planes, and we have described algorithms for finding such gflows.

The pattern optimisation algorithm of Theorem~\ref{thm:optimisation} and the circuit extraction algorithm of Section~\ref{sec:circuitextract} have been implemented in the \zxcalculus rewriting system \emph{PyZX}\footnote{PyZX is available at \url{https://github.com/Quantomatic/pyzx}. A Jupyter notebook demonstrating the T-count optimisation is available at \url{https://github.com/Quantomatic/pyzx/blob/5d409a246857b7600cc9bb0fbc13043d54fb9449/demos/T-count\%20Benchmark.ipynb}.}~\cite{pyzx}.
The reduction in non-Clifford gates using this method matches the state-of-the-art for ancillae-free circuits~\cite{tcountpreprint} at the time of development.

Our circuit extraction procedure resynthesises the CNOT gates in the circuit. Depending on the input circuit this can lead to drastic decreases in the 2-qubit gate count of the circuit~\cite{cliff-simp,pyzx}, but in many cases it can also lead to an \emph{increase} of the CNOT count.
Such increases are avoided in the procedure of Ref.~\cite{tcountpreprint}, where two-qubit gates are not resynthesised.

Yet re-synthesis of two-qubit gates may be necessary anyway in many applications: current and near-term quantum devices do not allow two-qubit gates to be applied to arbitrary pairs of qubits.
Thus, general circuits need to be adapted to the permitted connections; this is called routing.
Our extraction algorithm currently uses a standard process of Gaussian elimination to produce the two-qubit gates required to implement the circuit, which implicitly assumes that two-qubit gates can be applied between any pair of qubits in the circuit.
It may be useful to replace this procedure with one incorporating routing, such as the \zxcalculus-based routing approach by Kissinger and Meijer-van~de~Griend~\cite{kissinger2019cnot}.
This would allow circuits to be adapted to the desired architecture as they are being extracted.

It would also be interesting to consider whether these routing algorithms can be used more abstractly to transform general measurement patterns to more regular patterns with restricted connectivity.

\medskip

\noindent {\small \textbf{Acknowledgements.}
Many thanks to Fatimah Ahmadi for her contributions in the earlier stages of this project.
We would also like to thank Will Simmons for pointing out an error in Lemma~\ref{lemGFlowV1} in a previous version of this paper.
The majority of this work was developed at the Applied Category Theory summer school during the week of 22--26 July 2019; we thank the organisers of this summer school for bringing us together and making this work possible.
JvdW is supported in part by AFOSR grant FA2386-18-1-4028 and is supported by a Rubicon fellowship financed by the Dutch Research Council (NWO).
HJM-B is supported by the EPSRC.
}
\bibliographystyle{plainnat}
\bibliography{main}

\begin{thebibliography}{51}
\providecommand{\natexlab}[1]{#1}
\providecommand{\url}[1]{\texttt{#1}}
\expandafter\ifx\csname urlstyle\endcsname\relax
  \providecommand{\doi}[1]{doi: #1}\else
  \providecommand{\doi}{doi: \begingroup \urlstyle{rm}\Url}\fi

\bibitem[Aaronson and Gottesman(2004)]{aaronsongottesman2004}
Scott Aaronson and Daniel Gottesman.
\newblock Improved simulation of stabilizer circuits.
\newblock \emph{Physical Review A}, 70\penalty0 (5):\penalty0 052328, 2004.
\newblock \doi{10.1103/PhysRevA.70.052328}.

\bibitem[Amy et~al.(2013)Amy, Maslov, Mosca, and Roetteler]{amy2013meet}
Matthew Amy, Dmitri Maslov, Michele Mosca, and Martin Roetteler.
\newblock A meet-in-the-middle algorithm for fast synthesis of depth-optimal
  quantum circuits.
\newblock \emph{IEEE Transactions on Computer-Aided Design of Integrated
  Circuits and Systems}, 32\penalty0 (6):\penalty0 818--830, 2013.
\newblock \doi{10.1109/TCAD.2013.2244643}.

\bibitem[Amy et~al.(2014)Amy, Maslov, and Mosca]{amy2014polynomial}
Matthew Amy, Dmitri Maslov, and Michele Mosca.
\newblock {Polynomial-time T-depth optimization of Clifford+ T circuits via
  matroid partitioning}.
\newblock \emph{IEEE Transactions on Computer-Aided Design of Integrated
  Circuits and Systems}, 33\penalty0 (10):\penalty0 1476--1489, 2014.
\newblock \doi{10.1109/TCAD.2014.2341953}.

\bibitem[Backens(2014)]{backens1}
Miriam Backens.
\newblock The {ZX}-calculus is complete for stabilizer quantum mechanics.
\newblock \emph{New Journal of Physics}, 16\penalty0 (9):\penalty0 093021,
  2014.
\newblock ISSN 1367-2630.
\newblock \doi{10.1088/1367-2630/16/9/093021}.

\bibitem[Backens(2015)]{Backens:2015aa}
Miriam Backens.
\newblock {Making the stabilizer ZX-calculus complete for scalars}.
\newblock In Chris Heunen, Peter Selinger, and Jamie Vicary, editors,
  \emph{Proceedings of the 12th International Workshop on Quantum Physics and
  Logic (QPL 2015)}, volume 195 of \emph{Electronic Proceedings in Theoretical
  Computer Science}, pages 17--32, 2015.
\newblock \doi{10.4204/EPTCS.195.2}.

\bibitem[Bard(2006)]{Bard06}
Gregory Bard.
\newblock Achieving a log(n) speed up for boolean matrix operations and
  calculating the complexity of the dense linear algebra step of algebraic
  stream ciper attacks and of integer factorization methods.
\newblock \emph{IACR Cryptology ePrint Archive}, 2006:\penalty0 163, 01 2006.

\bibitem[Bian and Selinger(2015)]{Bian2Qubit}
Xiaoning Bian and Peter Selinger.
\newblock Relations for 2-qubit {Clifford+T} operator group, 2015.
\newblock URL \url{https://mathstat.dal.ca/~xbian/talks/slide_cliffordt2.pdf}.

\bibitem[Broadbent and Kashefi(2009)]{broadbent_2009_parallelizing}
Anne Broadbent and Elham Kashefi.
\newblock Parallelizing quantum circuits.
\newblock \emph{Theoretical Computer Science}, 410\penalty0 (26):\penalty0
  2489--2510, 2009.
\newblock \doi{10.1016/j.tcs.2008.12.046}.

\bibitem[Browne et~al.(2007)Browne, Kashefi, Mhalla, and Perdrix]{GFlow}
Daniel~E Browne, Elham Kashefi, Mehdi Mhalla, and Simon Perdrix.
\newblock Generalized flow and determinism in measurement-based quantum
  computation.
\newblock \emph{New Journal of Physics}, 9\penalty0 (8):\penalty0 250, 2007.
\newblock \doi{10.1088/1367-2630/9/8/250}.

\bibitem[Browne et~al.(2011)Browne, Kashefi, and Perdrix]{BKP10}
Daniel~E. Browne, Elham Kashefi, and Simon Perdrix.
\newblock Computational depth complexity of measurement-based quantum
  computation.
\newblock In \emph{Theory of Quantum Computation, Communication, and
  Cryptography (TQC'10)}, volume 6519, pages 35--46. LNCS, 2011.
\newblock \doi{10.1007/978-3-642-18073-6_4}.

\bibitem[Coecke and Duncan(2011)]{CD2}
B.~Coecke and R.~Duncan.
\newblock Interacting quantum observables: categorical algebra and
  diagrammatics.
\newblock \emph{New Journal of Physics}, 13:\penalty0 043016, 2011.
\newblock \doi{10.1088/1367-2630/13/4/043016}.
\newblock arXiv:quant-ph/09064725.

\bibitem[Coecke and Kissinger(2017)]{CKbook}
Bob Coecke and Aleks Kissinger.
\newblock \emph{Picturing Quantum Processes: A First Course in Quantum Theory
  and Diagrammatic Reasoning}.
\newblock Cambridge University Press, 2017.
\newblock ISBN 9781107104228.
\newblock \doi{10.1017/9781316219317}.

\bibitem[Cowtan et~al.(2020{\natexlab{a}})Cowtan, Dilkes, Duncan, Simmons, and
  Sivarajah]{phaseGadgetSynth}
Alexander Cowtan, Silas Dilkes, Ross Duncan, Will Simmons, and Seyon Sivarajah.
\newblock {Phase Gadget Synthesis for Shallow Circuits}.
\newblock In Bob Coecke and Matthew Leifer, editors, \emph{Proceedings 16th
  International Conference on Quantum Physics and Logic, Chapman University,
  Orange, CA, USA., 10-14 June 2019}, volume 318 of \emph{Electronic
  Proceedings in Theoretical Computer Science}, pages 213--228. Open Publishing
  Association, 2020{\natexlab{a}}.
\newblock \doi{10.4204/EPTCS.318.13}.

\bibitem[Cowtan et~al.(2020{\natexlab{b}})Cowtan, Simmons, and
  Duncan]{cowtan2020generic}
Alexander Cowtan, Will Simmons, and Ross Duncan.
\newblock {A Generic Compilation Strategy for the Unitary Coupled Cluster
  Ansatz}.
\newblock \emph{arXiv preprint arXiv:2007.10515}, 2020{\natexlab{b}}.

\bibitem[da~Silva and Galv{\~a}o(2013)]{daSilva2013compact}
Raphael~Dias da~Silva and Ernesto~F Galv{\~a}o.
\newblock Compact quantum circuits from one-way quantum computation.
\newblock \emph{Physical Review A}, 88\penalty0 (1):\penalty0 012319, 2013.
\newblock \doi{10.1103/physreva.88.012319}.

\bibitem[da~Silva et~al.(2013)da~Silva, Pius, and Kashefi]{daSilva2013global}
Raphael~Dias da~Silva, Einar Pius, and Elham Kashefi.
\newblock Global quantum circuit optimization.
\newblock \emph{arXiv:1301.0351}, 2013.

\bibitem[Danos and Kashefi(2006)]{Danos2006Determinism-in-}
V.~Danos and E.~Kashefi.
\newblock Determinism in the one-way model.
\newblock \emph{Phys. Rev. A}, 74\penalty0 (052310), 2006.
\newblock \doi{10.1103/PhysRevA.74.052310}.

\bibitem[Danos et~al.(2007)Danos, Kashefi, and Panangaden]{Patterns}
Vincent Danos, Elham Kashefi, and Prakash Panangaden.
\newblock The measurement calculus.
\newblock \emph{J. ACM}, 54\penalty0 (2), April 2007.
\newblock \doi{10.1145/1219092.1219096}.

\bibitem[Danos et~al.(2009)Danos, Kashefi, Panangaden, and
  Perdrix]{danos_kashefi_panangaden_perdrix_2009}
Vincent Danos, Elham Kashefi, Prakash Panangaden, and Simon Perdrix.
\newblock \emph{Extended Measurement Calculus}, pages 235--310.
\newblock Cambridge University Press, 2009.
\newblock \doi{10.1017/CBO9781139193313.008}.

\bibitem[de~Beaudrap(2010)]{beaudrap2010unitary}
Niel de~Beaudrap.
\newblock Unitary-circuit semantics for measurement-based computations.
\newblock \emph{International Journal of Quantum Information}, 8\penalty0
  (01n02):\penalty0 1--91, 2010.
\newblock \doi{10.1142/s0219749910006113}.

\bibitem[de~Beaudrap et~al.(2020)de~Beaudrap, Bian, and Wang]{pi4parity}
Niel de~Beaudrap, Xiaoning Bian, and Quanlong Wang.
\newblock {Techniques to Reduce $\pi/4$-Parity-Phase Circuits, Motivated by the
  ZX Calculus}.
\newblock In Bob Coecke and Matthew Leifer, editors, \emph{Proceedings 16th
  International Conference on Quantum Physics and Logic, Chapman University,
  Orange, CA, USA., 10-14 June 2019}, volume 318 of \emph{Electronic
  Proceedings in Theoretical Computer Science}, pages 131--149. Open Publishing
  Association, 2020.
\newblock \doi{10.4204/EPTCS.318.9}.

\bibitem[Di~Matteo and Mosca(2016)]{di2016parallelizing}
Olivia Di~Matteo and Michele Mosca.
\newblock Parallelizing quantum circuit synthesis.
\newblock \emph{Quantum Science and Technology}, 1\penalty0 (1):\penalty0
  015003, 2016.
\newblock \doi{10.1088/2058-9565/1/1/015003}.

\bibitem[Duncan and Perdrix(2009)]{DP1}
Ross Duncan and Simon Perdrix.
\newblock Graph states and the necessity of {E}uler decomposition.
\newblock \emph{Mathematical Theory and Computational Practice}, pages
  167--177, 2009.
\newblock \doi{10.1007/978-3-642-03073-4_18}.

\bibitem[Duncan and Perdrix(2010)]{duncan2010rewriting}
Ross Duncan and Simon Perdrix.
\newblock Rewriting measurement-based quantum computations with generalised
  flow.
\newblock In \emph{International Colloquium on Automata, Languages, and
  Programming}, pages 285--296. Springer, 2010.
\newblock \doi{10.1007/978-3-642-14162-1_24}.

\bibitem[Duncan and Perdrix(2014)]{DP3}
Ross Duncan and Simon Perdrix.
\newblock {Pivoting makes the ZX-calculus complete for real stabilizers}.
\newblock In Bob Coecke and Matty Hoban, editors, \emph{{\rm Proceedings of the
  10th International Workshop on} Quantum Physics and Logic (QPL)}, volume 171
  of \emph{Electronic Proceedings in Theoretical Computer Science}, pages
  50--62. Open Publishing Association, 2014.
\newblock \doi{10.4204/EPTCS.171.5}.

\bibitem[Duncan et~al.(2020)Duncan, Kissinger, Perdrix, and van~de
  Wetering]{cliff-simp}
Ross Duncan, Aleks Kissinger, Simon Perdrix, and John van~de Wetering.
\newblock {Graph-theoretic Simplification of Quantum Circuits with the
  ZX-calculus}.
\newblock \emph{{Quantum}}, 4:\penalty0 279, 6 2020.
\newblock ISSN 2521-327X.
\newblock \doi{10.22331/q-2020-06-04-279}.

\bibitem[Eslamy et~al.(2018)Eslamy, Houshmand, Zamani, and
  Sedighi]{eslamy2018optimization}
Maryam Eslamy, Mahboobeh Houshmand, Morteza~Saheb Zamani, and Mehdi Sedighi.
\newblock Optimization of one-way quantum computation measurement patterns.
\newblock \emph{International Journal of Theoretical Physics}, 57\penalty0
  (11):\penalty0 3296--3317, 2018.
\newblock \doi{10.1007/s10773-018-3844-x}.

\bibitem[Hadzihasanovic et~al.(2018)Hadzihasanovic, Ng, and
  Wang]{HarnyAmarCompleteness}
Amar Hadzihasanovic, Kang~Feng Ng, and Quanlong Wang.
\newblock {Two Complete Axiomatisations of Pure-state Qubit Quantum Computing}.
\newblock In \emph{Proceedings of the 33rd Annual ACM/IEEE Symposium on Logic
  in Computer Science}, LICS '18, pages 502--511, New York, NY, USA, 2018. ACM.
\newblock ISBN 978-1-4503-5583-4.
\newblock \doi{10.1145/3209108.3209128}.

\bibitem[Hamrit and Perdrix(2015)]{hamrit2015reversibility}
Nidhal Hamrit and Simon Perdrix.
\newblock Reversibility in extended measurement-based quantum computation.
\newblock In Jean Krivine and Jean-Bernard Stefani, editors, \emph{Reversible
  Computation}, pages 129--138. Springer International Publishing, 2015.
\newblock ISBN 978-3-319-20860-2.
\newblock \doi{10.1007/978-3-319-20860-2_8}.

\bibitem[Hein et~al.(2004)Hein, Eisert, and Briegel]{hein2004multiparty}
M.~Hein, J.~Eisert, and H.~J. Briegel.
\newblock Multiparty entanglement in graph states.
\newblock \emph{Physical Review A}, 69\penalty0 (6):\penalty0 062311, 2004.
\newblock \doi{10.1103/PhysRevA.69.062311}.

\bibitem[Hein et~al.(2006)Hein, D{\"u}r, Eisert, Raussendorf, Nest, and
  Briegel]{hein2006entanglement}
Marc Hein, Wolfgang D{\"u}r, Jens Eisert, Robert Raussendorf, M~Nest, and H-J
  Briegel.
\newblock {Entanglement in graph states and its applications}.
\newblock In G.~Casati, D.~L. Shepelyansky, P.~Zoller, and G.~Benenti, editors,
  \emph{Quantum Computers, Algorithms and Chaos}, volume 162 of
  \emph{Proceedings of the International School of Physics "Enrico Fermi"},
  pages 115--218. IOS Press Books, 2006.
\newblock \doi{10.3254/978-1-61499-018-5-115}.

\bibitem[Houshmand et~al.(2017)Houshmand, Sedighi, Zamani, and
  Marjoei]{houshmand2017quantum}
Mahboobeh Houshmand, Mehdi Sedighi, Morteza~Saheb Zamani, and Kourosh Marjoei.
\newblock Quantum circuit synthesis targeting to improve one-way quantum
  computation pattern cost metrics.
\newblock \emph{ACM Journal on Emerging Technologies in Computing Systems
  (JETC)}, 13\penalty0 (4):\penalty0 55, 2017.
\newblock \doi{10.1145/3064834}.

\bibitem[Houshmand et~al.(2018)Houshmand, Houshmand, and
  Fitzsimons]{houshmand2018minimal}
Monireh Houshmand, Mahboobeh Houshmand, and Joseph~F Fitzsimons.
\newblock Minimal qubit resources for the realization of measurement-based
  quantum computation.
\newblock \emph{Physical Review A}, 98\penalty0 (1):\penalty0 012318, 2018.
\newblock \doi{10.1103/physreva.98.012318}.

\bibitem[Jeandel et~al.(2018{\natexlab{a}})Jeandel, Perdrix, and
  Vilmart]{JPV-universal}
Emmanuel Jeandel, Simon Perdrix, and Renaud Vilmart.
\newblock {Diagrammatic Reasoning Beyond Clifford+T Quantum Mechanics}.
\newblock In \emph{Proceedings of the 33rd Annual ACM/IEEE Symposium on Logic
  in Computer Science}, LICS '18, pages 569--578, New York, NY, USA,
  2018{\natexlab{a}}. ACM.
\newblock \doi{10.1145/3209108.3209139}.

\bibitem[Jeandel et~al.(2018{\natexlab{b}})Jeandel, Perdrix, and
  Vilmart]{SimonCompleteness}
Emmanuel Jeandel, Simon Perdrix, and Renaud Vilmart.
\newblock {A Complete Axiomatisation of the ZX-Calculus for Clifford+T Quantum
  Mechanics}.
\newblock In \emph{Proceedings of the 33rd Annual ACM/IEEE Symposium on Logic
  in Computer Science}, LICS '18, pages 559--568, New York, NY, USA,
  2018{\natexlab{b}}. ACM.
\newblock \doi{10.1145/3209108.3209131}.

\bibitem[Kissinger and Meijer-van~de Griend(2020)]{kissinger2019cnot}
Aleks Kissinger and Arianne Meijer-van~de Griend.
\newblock {CNOT circuit extraction for topologically-constrained quantum
  memories}.
\newblock \emph{Quantum Information and Computation}, 20:\penalty0 581--596,
  2020.
\newblock \doi{10.26421/QIC20.7-8}.

\bibitem[Kissinger and van~de Wetering(2019)]{kissinger2017MBQC}
Aleks Kissinger and John van~de Wetering.
\newblock Universal {MBQC} with generalised parity-phase interactions and
  {P}auli measurements.
\newblock \emph{{Quantum}}, 3:\penalty0 134, 2019.
\newblock \doi{10.22331/q-2019-04-26-134}.

\bibitem[Kissinger and van~de Wetering(2020{\natexlab{a}})]{pyzx}
Aleks Kissinger and John van~de Wetering.
\newblock {PyZX: Large Scale Automated Diagrammatic Reasoning}.
\newblock In Bob Coecke and Matthew Leifer, editors, \emph{Proceedings 16th
  International Conference on Quantum Physics and Logic, Chapman University,
  Orange, CA, USA., 10-14 June 2019}, volume 318 of \emph{Electronic
  Proceedings in Theoretical Computer Science}, pages 229--241. Open Publishing
  Association, 2020{\natexlab{a}}.
\newblock \doi{10.4204/EPTCS.318.14}.

\bibitem[Kissinger and van~de Wetering(2020{\natexlab{b}})]{tcountpreprint}
Aleks Kissinger and John van~de Wetering.
\newblock {Reducing the number of non-Clifford gates in quantum circuits}.
\newblock \emph{Physical Review A, Vol.102-2}, 102:\penalty0 022406, 8
  2020{\natexlab{b}}.
\newblock \doi{10.1103/PhysRevA.102.022406}.

\bibitem[Kliuchnikov and Maslov(2013)]{CliffOpt}
Vadym Kliuchnikov and Dmitri Maslov.
\newblock {Optimization of Clifford circuits}.
\newblock \emph{Phys. Rev. A}, 88:\penalty0 052307, 2013.
\newblock \doi{10.1103/PhysRevA.88.052307}.

\bibitem[Kotzig(1968)]{kotzig}
Anton Kotzig.
\newblock Eulerian lines in finite 4-valent graphs and their transformations.
\newblock In \emph{Colloqium on Graph Theory Tihany 1966}, pages 219--230.
  Academic Press, 1968.

\bibitem[Mhalla and Perdrix(2008)]{MP08-icalp}
Mehdi Mhalla and Simon Perdrix.
\newblock Finding optimal flows efficiently.
\newblock In \emph{the 35th International Colloquium on Automata, Languages and
  Programming (ICALP), LNCS}, volume 5125, pages 857--868, 2008.
\newblock \doi{10.1007/978-3-540-70575-8_70}.

\bibitem[Mhalla and Perdrix(2012)]{mhalla2012graph}
Mehdi Mhalla and Simon Perdrix.
\newblock Graph states, pivot minor, and universality of {(X, Z)}-measurements.
\newblock \emph{International Journal of Unconventional Computing}, 9\penalty0
  (1--2):\penalty0 153--171, 2012.

\bibitem[Mhalla et~al.(2011)Mhalla, Murao, Perdrix, Someya, and
  Turner]{mhalla2011graph}
Mehdi Mhalla, Mio Murao, Simon Perdrix, Masato Someya, and Peter~S Turner.
\newblock Which graph states are useful for quantum information processing?
\newblock In \emph{Conference on Quantum Computation, Communication, and
  Cryptography}, pages 174--187. Springer, 2011.
\newblock \doi{10.1007/978-3-642-54429-3_12}.

\bibitem[Miyazaki et~al.(2015)Miyazaki, Hajdu{\v{s}}ek, and
  Murao]{miyazaki2015analysis}
Jisho Miyazaki, Michal Hajdu{\v{s}}ek, and Mio Murao.
\newblock Analysis of the trade-off between spatial and temporal resources for
  measurement-based quantum computation.
\newblock \emph{Physical Review A}, 91\penalty0 (5):\penalty0 052302, 2015.
\newblock \doi{10.1103/physreva.91.052302}.

\bibitem[Ng and Wang(2017)]{ng2017universal}
Kang~Feng Ng and Quanlong Wang.
\newblock A universal completion of the {ZX}-calculus.
\newblock \emph{arXiv:1706.09877}, 2017.

\bibitem[Nielsen and Chuang(2010)]{NielsenChuang}
M.~A. Nielsen and Isaac~L. Chuang.
\newblock \emph{Quantum computation and quantum information}.
\newblock Cambridge university press, 2010.
\newblock \doi{10.1119/1.1463744}.

\bibitem[Raussendorf et~al.(2003)Raussendorf, Browne, and Briegel]{MBQC2}
R.~Raussendorf, D.E. Browne, and H.J. Briegel.
\newblock Measurement-based quantum computation on cluster states.
\newblock \emph{Physical Review A}, 68\penalty0 (2):\penalty0 22312, 2003.
\newblock \doi{10.1103/physreva.68.022312}.

\bibitem[Raussendorf and Briegel(2001)]{MBQC1}
Robert Raussendorf and Hans~J. Briegel.
\newblock A one-way quantum computer.
\newblock \emph{Phys. Rev. Lett.}, 86:\penalty0 5188--5191, 2001.
\newblock \doi{10.1103/PhysRevLett.86.5188}.

\bibitem[Van~den Nest et~al.(2004)Van~den Nest, Dehaene, and De~Moor]{NestMBQC}
M.~Van~den Nest, J.~Dehaene, and B.~De~Moor.
\newblock Graphical description of the action of local {C}lifford
  transformations on graph states.
\newblock \emph{Physical Review A}, 69\penalty0 (2):\penalty0 9422, 2004.
\newblock \doi{10.1103/physreva.69.022316}.

\bibitem[Vilmart(2019)]{euler-zx}
Renaud Vilmart.
\newblock A {Near}-{Minimal} {Axiomatisation} of {ZX}-{Calculus} for {Pure}
  {Qubit} {Quantum} {Mechanics}.
\newblock In \emph{2019 34th {Annual} {ACM}/{IEEE} {Symposium} on {Logic} in
  {Computer} {Science} ({LICS})}, pages 1--10, 2019.
\newblock \doi{10.1109/LICS.2019.8785765}.

\end{thebibliography}

\appendix

\section{Gflow and determinism}\label{sec:gflow-determinism}

Here, we introduce the notation and concepts needed to state and prove the fact that gflow is a sufficient condition for uniform, strong and stepwise determinism of a measurement pattern. Then we give the detailed statement along with the proof.

Throughout this appendix, $\simeq$ shall denote equality up to a global phase; that is, for linear maps $L$ and $K$, we have $L\simeq K$ iff there is an angle $\al\in[0,2\pi)$ such that $L=e^{i\al}K$.

Recall from Definition~\ref{def:determinism} that the linear map implemented by a particular set of measurement outcomes in a pattern is called a branch of the pattern. We begin by making this more precise.
\begin{definition}\label{def:branch}
Let $\pat$ be a measurement pattern with $n$ measurement commands. A {\em branch} of the pattern is a sequence $p$ of length $n$ whose entries are binary digits $0$ and $1$ together with a linear map $\intf{\pat}_p$, where each measurement is replaced with an outcome corresponding to either $0$ or $1$ as indicated by the corresponding entry in $p$.
\end{definition}
If $s$ is a binary sequence with $k$ entries where $k\leq n$, we may also define an `intermediate pattern' $\intf{\pat}_s$ by replacing the first $k$ measurements in $\pat$ by an outcome corresponding to the entries in $s$.

We may now state various forms of determinism (Definition~\ref{def:determinism}) using the notation of the above definition. A pattern $\pat$ with $n$ measurement commands is {\em deterministic} if for any branches $p$ and $q$ there is a complex number $c$ such that $\intf{\pat}_p=c\intf{\pat}_q$. It is {\em strongly deterministic} if for any branches $p$ and $q$ we have $\intf{\pat}_p\simeq\intf{\pat}_q$. It is {\em stepwise deterministic} if the pattern $\intf{\pat}_s$ is deterministic for any binary sequence $s$ with at most $n$ entries. Recall that a pattern is {\em uniformly deterministic} if it is deterministic for any choice of measurement angles in the measurement commands appearing in $\pat$.

It will be useful to encode the measurement planes using following functions for any \LOG\ $(G,I,O,\ld)$.
\begin{align*}
\beta_X,\beta_Z : \comp O &\rightarrow\{0,1\} \\
\beta_X(i) &= \begin{cases} 1 & \text{if } \ld(i)\in\{\XYm,\XZm\} \\
                            0 & \text{if } \ld(i)=\YZm \end{cases} \\
\beta_Z(i) &= \begin{cases} 1 & \text{if } \ld(i)\in\{\XZm,\YZm\} \\
                            0 & \text{if } \ld(i)=\XYm \end{cases}
\end{align*}

The following observation will be the starting point for our proof of the main theorem.
\begin{lemma}\label{lma:corrections}
Let $s_i$ be the binary digit indicating the measurement outcome of qubit $i$. Then\footnote{Note that the pattern below is not runnable, as the corrections depend on a future measurement outcome.}
\begin{equation*}
\intf{M^{\ld(i), \al(i)}_i [X_i]^{\beta_Z(i)s_i} [Z_i]^{\beta_X(i)s_i}}_p \simeq \bra{+_{\ld(i),\al(i)}}_i
\end{equation*}
for both branches $p$.
\end{lemma}
\begin{proof}
If $p=(0)$, so that $s_i=0$, this is clear, as the corrections vanish while there is nothing to correct. Hence suppose $p=(1)$. We have three cases:
\begin{enumerate}
\item $\ld(i) = \XYm$, so that $\beta_Z(i)=0$ and $\beta_X(i)=1$,
\item $\ld(i) = \XZm$, so that $\beta_Z(i)=\beta_X(i)=1$,
\item $\ld(i) = \YZm$, so that $\beta_Z(i)=1$ and $\beta_X(i)=0$.
\end{enumerate}
Recalling that the Pauli operators act on the basis qubits as
\begin{align*}
X\ket{0}=\ket{1}\quad &X\ket{1}=\ket{0} \\
Z\ket{0}=-\ket{0}\quad &Z\ket{1}=\ket{1}
\end{align*}
we compute in each case, starting from the (dual of the) left-hand side:
\begin{align*}
1. && Z_i(\ket 0 - e^{i\al(i)}\ket 1) &= -(\ket 0 + e^{i\al(i)}\ket 1) \\
&&& \simeq \ket{+_{\XYm,\al(i)}}_i, \\
2. && Z_iX_i\left(\sin\left(\frac {\al(i)} 2\right)\ket 0 - \cos\left(\frac {\al(i)} 2\right)\ket 1\right) &= \sin\left(\frac {\al(i)} 2\right)\ket 1 + \cos\left(\frac {\al(i)} 2\right)\ket 0 \\
&&& = \ket{+_{\XZm,\al(i)}}_i, \\
3. && X_i\left(\sin\left(\frac {\al(i)} 2\right)\ket 0 - i\cos\left(\frac {\al(i)} 2\right)\ket 1\right) &= \sin\left(\frac {\al(i)} 2\right)\ket 1 - i\cos\left(\frac {\al(i)} 2\right)\ket 0 \\
&&&\simeq \cos\left(\frac {\al(i)} 2\right)\ket 0 + i\sin\left(\frac {\al(i)} 2\right)\ket 1 \\
&&&= \ket{+_{\YZm,\al(i)}}_i. \qedhere
 \end{align*}
\end{proof}

\begin{definition}\label{def:graph_stabiliser}
Let $(G,I,O)$ be an open graph. View $(V,I,O)$ as the register of a measurement pattern as in Definition~\ref{def:meas_pattern}. Define the following command for each $i\in\comp I$:
$$K_i\coloneqq X_i\prod_{j\in N_G(i)}Z_j.$$
We call $K_i$ the \emph{graph stabiliser} at qubit $i$ for reasons explained below. Observe that the commands $Z_j$ in the product act on different qubits and hence commute with each other. Thus the product does not depend on the order, so that $K_i$ is well-defined.
\end{definition}
The reason for calling the $K_i$ graph stabilisers is the following lemma.
\begin{lemma}[{\cite[pp.~6-7]{GFlow}}]\label{lma:graph_stabiliser}
For any open graph $(G,I,O)$ and for any $i\in\comp I$ with graph stabiliser $K_i$ we have
$$K_iE_GN_{\comp I} = E_GN_{\comp I}.$$
\end{lemma}

\paragraph{Notation:} Given a command $C$ in a measurement pattern with register $V$, for any subset of qubits $S\sse V$ we write $C_{S}\coloneqq\prod_{i\in S}C_i$ if the right-hand side is well-defined (in particular, if it does not depend on the order in which the different operations are applied). If all the $C_i$'s are moreover Pauli $X$ or $Z$ operators, and $s_i$ is a binary digit corresponding to the measurement outcome of qubit $i$, we may treat $C_{S}$ as a joint correction by writing $C_{S}^{s_i}\coloneqq\prod_{i\in S}[C_i]^{s_i}$. Note that the graph stabiliser may now be written as $K_i = X_iZ_{N_G(i)}$.

We are now ready to state and prove the central theorem about gflow.
\begin{theorem}\label{t-flow-app}
Let $\Gamma = (G,I,O, \ld)$ be a \LOG. If $\Gamma$ has a gflow $\gflow = (g,\prec)$, then for any set of measurement angles $\alpha:\comp O\rightarrow[0,2\pi)$ the following pattern is runnable as well as strongly and stepwise deterministic:
\begin{equation}\label{pattern}
\pat_{\Gamma,\gflow,\al}\coloneqq\prod^\prec_{i\in \comp{O}} \left( X_{g(i)\setminus\{i\}}^{s_i}Z_{\odd{}{g(i)}\setminus\{i\}}^{s_i} M^{\lambda(i), \al(i)}_i\right) E_GN_{\comp{I}},
\end{equation}
where $\prod^\prec$ denotes concatenation of commands following the order $\prec$. Moreover, $\pat_{\Gamma,\gflow,\al}$ realizes the linear map associated with $\Gamma$ at angles $\al$ (Definition~\ref{def:ogs-to-linear-map}) in the sense that for any branch $p$ we have
\begin{equation}\label{eqn:unit_embedding}
\intf{\pat_{\Gamma,\gflow,\al}}_p\simeq\prod_{i\in \comp{O}}\bra{{+_{\lambda(i), \al(i)}}}_i E_G\phZ {\comp{I}}{}.
\end{equation}
\end{theorem}
\begin{remark}\label{rem:strong-unitary}
Note that the pattern in equation~\eqref{pattern} is uniformly deterministic, as it is (strongly) deterministic independently of the choice of $\al$. It can be shown \cite[Proposition 7.3.4]{danos_kashefi_panangaden_perdrix_2009} that any strongly deterministic pattern implements a unitary embedding. Thus the map in equation \eqref{eqn:unit_embedding} is guaranteed to be one. In particular, if $|I|=|O|$, then the map is unitary.
\end{remark}
\begin{proof}
First observe that $\pat_{\Gamma,\gflow,\al}$ is indeed runnable; the conditions \ref{it:g} and \ref{it:odd} ensure that no command acts on a qubit that has been measured\footnote{We remark that the condition \eqref{eq:wrong-g2} would fail to guarantee this.}, and all the other conditions in Definition~\ref{def:runnable_pattern} are easily verified from the shape of the pattern.

For strong determinism, it is sufficient to show the claimed equality \eqref{eqn:unit_embedding}, as this shows that any two branches are equal (up to a global phase). The proof of the equality is based on Lemma~\ref{lma:corrections} and the following observation. Using the fact that the Pauli matrices are involutive, we show that for all $i\in\comp O$,
\begin{align}\label{eqn:identity_Z}
\prod_{u\in g(i)}Z_{\eve{}{g(i)}\cap N_G(u)} &= I, \\ \nonumber
\prod_{u\in g(i)}Z_{\odd{}{g(i)}\cap N_G(u)} &= Z_{\odd{}{g(i)}},
\end{align}
where $I$ is the identity matrix. The first equality follows by taking $u\in g(i)$ and $j\in N_G(u)$ such that $j$ is in the even neighbourhood of $g(i)$. Then the number of neighbours of $j$ in $g(i)$ is even, so that $Z_j$ appears in the product an even number of times. For the second equality, take $u\in g(i)$ and $j\in N_G(u)$ such that $j$ is in the odd neighbourhood of $g(i)$. Then the number of neighbours of $j$ in $g(i)$ is odd, so that $Z_j$ appears in the product an odd number of times.

Thus we compute from equation~\eqref{pattern}:
\begin{align*}
& \intf{\pat_{\Gamma,\gflow,\al}}_p \\
&= \prod^\prec_{i\in \comp{O}} \intf{ M^{\lambda(i), \al(i)}_i X_{g(i)\setminus\{i\}}^{s_i} Z_{\odd{}{g(i)}\setminus\{i\}}^{s_i} }_p E_G N_{\comp{I}} && \textrm{(s1)} \\
&= \prod^\prec_{i\in \comp{O}} \intf{ M^{\lambda(i), \al(i)}_i \cx{i}{2\beta_Z(i)s_i} \cz{i}{2\beta_X(i)s_i} X_{g(i)\setminus\{i\}}^{s_i} Z_{\odd{}{g(i)}\setminus\{i\}}^{s_i} }_p E_G N_{\comp{I}} && \textrm{(s2)} \\
&\simeq \prod^\prec_{i\in \comp{O}} \left( \bra{+_{\lambda(i), \al(i)}}_i \cx{i}{\beta_Z(i)s_i} \cz{i}{\beta_X(i)s_i} X_{g(i)\setminus\{i\}}^{s_i} Z_{\odd{}{g(i)}\setminus\{i\}}^{s_i} \right) E_G N_{\comp{I}} && \textrm{(s3)} \\
&\simeq \prod^\prec_{i\in \comp{O}} \left( \bra{+_{\lambda(i), \al(i)}}_i X_{g(i)}^{s_i} Z_{\odd{}{g(i)}}^{s_i} \right) E_G N_{\comp{I}} && \textrm{(s4)} \\
&= \prod^\prec_{i\in \comp{O}} \left( \bra{+_{\lambda(i), \al(i)}}_i \prod_{u\in g(i)}\cx{u}{s_i} \prod_{v\in g(i)}Z_{\odd{}{g(i)}\cap N_G(v)}^{s_i} \prod_{w\in g(i)}Z_{\eve{}{g(i)}\cap N_G(w)}^{s_i} \right) E_G N_{\comp{I}} && \textrm{(s5)} \\
&\simeq \prod^\prec_{i\in \comp{O}} \left( \bra{+_{\lambda(i), \al(i)}}_i \prod_{u\in g(i)} \left( \cx{u}{s_i} Z_{\odd{}{g(i)}\cap N_G(u)}^{s_i} Z_{\eve{}{g(i)}\cap N_G(u)}^{s_i} \right) \right) E_G N_{\comp{I}} && \textrm{(s6)} \\
&= \prod^\prec_{i\in \comp{O}} \left( \bra{+_{\lambda(i), \al(i)}}_i \prod_{u\in g(i)} \left( \cx{u}{s_i} Z_{N_G(u)}^{s_i} \right) \right) E_G N_{\comp{I}} && \\
&= \prod^\prec_{i\in \comp{O}} \left( \bra{+_{\lambda(i), \al(i)}}_i \prod_{u\in g(i)}K_u^{s_i} \right) E_G N_{\comp{I}} && \\
&= \prod^\prec_{i\in \comp{O}} \bra{+_{\lambda(i), \al(i)}}_i E_G N_{\comp{I}}. && \textrm{(s7)}
\end{align*}
Here we used:
\begin{enumerate}[(s1)]
\item The fact that the correction terms do not act on the vertex $i$.
\item Involutivity of the Pauli matrices.
\item Lemma~\ref{lma:corrections} and that $X_i$ and $Z_i$ anti-commute.
\item Conditions~\ref{it:XY}-\ref{it:YZ}. For example, if $\ld(i)=YZ$, then $\beta_Z(i)=1$ and by \ref{it:YZ} $i\in g(i)$, so that $\cx{i}{s_i}X_{g(i)\setminus\{i\}}^{s_i}=X_{g(i)}^{s_i}$, while $\beta_X(i)=0$ and again by \ref{it:YZ} $i\notin\odd{}{g(i)}$, so that $Z_{\odd{}{g(i)}\setminus\{i\}}=Z_{\odd{}{g(i)}}$. The cases of other measurement planes are similar.
\item Equations~\eqref{eqn:identity_Z} and the definition of $X_{g(i)}^{s_i}$.
\item The facts that $X_i$ and $Z_i$ anti-commute, and that commands acting on different qubits commute.
\item Lemma~\ref{lma:graph_stabiliser} and~\ref{it:g} in order to commute the graph stabilisers past the liner maps.
\end{enumerate}
Thus the pattern is strongly deterministic.

It remains to show that $\pat_{\Gamma,\gflow,\al}$ is stepwise deterministic. Thus consider a subgraph of $\Gamma$ (with $I$, $O$ and $\ld$ restricted appropriately) obtained by removing some minimal vertices in the order $\prec$. The interpretation is that those vertices have already been measured. By conditions \ref{it:g} and \ref{it:odd}, for any vertex $i$ in the subgraph, the correction set $g(i)$ and its odd neighbourhood do not contain any of the removed vertices, and are thus themselves contained in the subgraph. Therefore we may restrict $g$ to the subgraph, and the odd neighbourhood of each $g(i)$ in the subgraph will be equal to its odd neighbourhood in the full graph. Thus the conditions \ref{it:XY}, \ref{it:XZ} and \ref{it:YZ} are satisfied by the subgraph. If we also restrict the partial order to the subgraph, we find that the conditions \ref{it:g} and \ref{it:odd} are satisfied. Hence the gflow of the full graph restricts to a gflow on the subgraph. Therefore, since the measurements in $\pat_{\Gamma,\gflow,\al}$ follow the order $\prec$, we may think of each measurement as removing a minimal vertex in $\Gamma$. As we have just seen, the graph with some minimal vertices removed has a gflow given by restriction of the original gflow, so that repeating the preceding parts of the proof will yield another (strongly) deterministic pattern. This pattern is equal to the proper sub-pattern of $\pat_{\Gamma,\gflow,\al}$ which does not involve the already-performed measurement and its corresponding corrections. Therefore, the original pattern is stepwise deterministic.
\end{proof}

\section{Proofs for Section~\ref{sec:rewriting}}\label{sec:proofs}

The following lemma records how the odd neighbourhood of an arbitrary set of vertices changes under local complementation.

\begin{lemma}(\cite[Lemma B.4]{cliff-simp})\label{lem:oddneighbours}
Given a graph $G=(V,E)$, a subset $A\subseteq V$ and a vertex $u \in V$, we have
$$\odd {G\star u} A=\begin{cases}\odd G A \symd (N_G(u)\cap A) & \text{if $u\notin \odd G A$}\\ \odd G {A} \symd (N_G(u)\setminus A)&\text{if $u\in \odd G A$}.\end{cases}$$
\end{lemma}

Now we can prove the main lemma.

\begin{replem}{lem:lc_gflow}
 \statelcgflow
\end{replem}
\begin{proof}
We divide the proof into three parts, first proving that $g'(u)$ satisfies the required properties, then proving that $g'(v)$ for $v\neq u$ satisfies \ref{it:g} and \ref{it:odd}, and finally showing that $g''(v)$ satisfies the conditions \ref{it:XY}--\ref{it:YZ}. Each part is further subdivided based on the values of $\lambda(u)$ or $\lambda(v)$.

\noindent \textbf{Part 1}: Note that condition~\ref{it:g} is trivially satisfied by $g'(u)$. We now show that $g'(u)$ satisfies the remaining conditions, subdividing into two cases based on $\lambda(u)$.

 \textbf{Case 1a}: If $\ld(u)\in\{\XYm,\XZm\}$, then $u\in\odd{G}{g(u)}$ and $\ld'(u)\in\{\XYm,\XZm\}$ with $\ld'(u)\neq\ld(u)$.
 We have $\odd{G\star u}{\{u\}} = N_{G\star u}(u) = N_G(u)$.
 Thus
 \begin{align}
  \odd{G\star u}{g'(u)} &= \odd{G\star u}{g(u)\symd\{u\}} \nonumber \\
  &= \odd{G\star u}{g(u)} \symd \odd{G\star u}{\{u\}} \nonumber \\
  &= \odd{G}{g(u)} \symd (N_G(u)\setminus g(u)) \symd N_G(u) \nonumber \\
  &= \odd{G}{g(u)} \symd (N_G(u)\cap g(u)), \label{eq:ld_XYXZ}
 \end{align}
 where the third equality uses Lemma~\ref{lem:oddneighbours}.
 Hence we have
 \begin{itemize}
  \item $u\in\odd{G\star u}{g'(u)}$ since $u\in\odd{G}{g(u)}$ and $u\notin N_G(u)$,
  \item $u\in g'(u)$ if and only if $u\notin g(u)$ (this is desired because $\ld'(u)\neq \ld(u)$); together with the previous item this yields Condition~\ref{it:XY} or \ref{it:XZ}, as appropriate,
  \item if $v\in\odd{G\star u}{g'(u)}$, then either $v\in\odd{G}{g(u)}$ or $v\in g(u)$, so either $u=v$ or $u\prec v$; this is Condition~\ref{it:odd}.
 \end{itemize}
 So $g'(u)$ indeed satisfies all the required conditions.

 \textbf{Case 1b}: If $\ld(u)=\YZm$, then $u\notin\odd{G}{g(u)}$, hence
 \begin{equation}\label{eq:ld_YZ}
  \odd{G\star u}{g'(u)} = \odd{G\star u}{g(u)} = \odd{G}{g(u)} \symd (N_G(u)\cap g(u)),
 \end{equation}
 again using Lemma~\ref{lem:oddneighbours}.
 This implies
 \begin{itemize}
  \item $u\notin\odd{G\star u}{g'(u)}$ because $u\notin\odd{G}{g(u)}$ and $u\notin N_G(u)$,
  \item $u\in g'(u)$ because $u\in g(u)$; together with the previous item this yields Condition~\ref{it:YZ},
  \item if $v\in\odd{G\star u}{g'(u)}$, then either $v\in\odd{G}{g(u)}$ or $v\in g(u)$, so either $u=v$ or $u\prec v$; this is Condition~\ref{it:odd}.
 \end{itemize}
 So $g'(u)$ indeed satisfies all the required conditions.

 \noindent\textbf{Part 2}:
We show that the correction set $g(v)$ for any $v\neq u$ satisfies Conditions~\ref{it:g} and \ref{it:odd} of the gflow definition.
The proof is split into subcases depending on whether $u\in\odd{G}{g(v)}$.

 \textbf{Case 2a}:
 If $u\notin\odd{G}{g(v)}$, first note that
 \[
  \odd{G\star u}{g'(v)} = \odd{G\star u}{g(v)} = \odd{G}{g(v)} \symd (N_G(u)\cap g(v))
 \]
 Hence if $w\in\odd{G\star u}{g(v)}$, then either $w\in\odd{G}{g(v)}$ or $w\in g(v)$, so $v=w$ or $v\prec w$; this is Condition~\ref{it:odd}.
 Condition~\ref{it:g} is trivially still satisfied.

 \textbf{Case 2b}:
 If $u\in\odd{G}{g(v)}$, first note that $v\prec u$.
 If $w\in g'(v)$, then $w\in g(v)$ or $w\in g'(u)$ or $w=u$, and in each case either $w=v$ or $v\prec w$, this is Condition~\ref{it:g}.
 Furthermore, we have
 \begin{align*}
  \odd{G\star u}{g'(v)} &= \odd{G\star u}{g(v) \symd g'(u) \symd \{u\}} \\
  &= \odd{G\star u}{g(v)} \symd \odd{G\star u}{g'(u)} \symd N_G(u) \\
  &= \odd{G}{g(v)} \symd (N_G(u)\setminus g(v)) \symd \odd{G}{g(u)} \symd (N_G(u)\cap g(u)) \symd N_G(u) \\
  &= \odd{G}{g(v)} \symd (N_G(u)\cap g(v)) \symd \odd{G}{g(u)} \symd (N_G(u)\cap g(u))
 \end{align*}
 where the third step uses the property that $\odd{G\star u}{g'(u)} = \odd{G}{g(u)} \symd (N_G(u)\cap g(u))$ for any $\ld(u)$, which follows from combining \eqref{eq:ld_XYXZ} and \eqref{eq:ld_YZ}.
 Hence if $w\in\odd{G\star u}{g'(v)}$, then at least one of the following holds:
 \begin{itemize}
  \item $w\in\odd{G}{g(v)}$, so $v=w$ or $v\prec w$
  \item $w\in g(v)$, so $v=w$ or $v\prec w$
  \item $w\in\odd{G}{g(u)}$, so $u=w$ or $u\prec w$; in both cases $v\prec w$ since $v\prec u$
  \item $w\in g(u)$, so $u=w$ or $u\prec w$; in both cases $v\prec w$ since $v\prec u$.
 \end{itemize}
 In each case, Condition~\ref{it:odd} is satisfied.

 \noindent\textbf{Part 3}:
 Finally we show that the correction set $g(v)$ for any $v\neq u$ satisfies Conditions~\ref{it:XY}, \ref{it:XZ} or \ref{it:YZ} of gflow.
 The proof is split into subcases depending on whether $v\in N_G(u)$.

 \textbf{Case 3a}: Suppose $v\in N_G(u)$ and distinguish cases according to $\ld(v)$.
 \begin{itemize}
  \item Suppose $\ld(v)=\XYm$, then $v\notin g(v)$ and $v\in\odd{G}{g(v)}$. Furthermore, $\ld'(v)=\XYm$.
  We have
   \begin{itemize}
    \item $v\in\odd{G\star u}{g(v)}$ since $v\in\odd{G}{g(v)}$ and $v\notin g(v)$,
    \item $v\notin g'(v)$,
   \end{itemize}
   which together give Condition~\ref{it:XY}.
  \item Suppose $\ld(v)=\XZm$, then $v\in g(v)$ and $v\in\odd{G}{g(v)}$.
   Furthermore, $\ld'(v) = \YZm$.
   We have
    \begin{itemize}
     \item $v\notin\odd{G\star u}{g(v)}$ since $v\in\odd{G}{g(v)}$ and $v\in N_G(u)\cap g(v)$,
     \item $v\in g'(v)$,
    \end{itemize}
   which together give Condition~\ref{it:YZ}.
  \item Suppose $\ld(v)=\YZm$, then $v\in g(v)$ and $v\notin\odd{G}{g(v)}$.
   Furthermore, $\ld'(v) = \XZm$.
   We have
    \begin{itemize}
     \item $v\in\odd{G\star u}{g(v)}$ since $v\notin\odd{G}{g(v)}$ and $v\in N_G(u)\cap g(v)$,
     \item $v\in g'(v)$,
    \end{itemize}
   which together give Condition~\ref{it:XZ}.
 \end{itemize}

 \textbf{Case 3b}: Suppose $v\notin N_G(u)$ and distinguish cases according to $\ld(v)$.
 \begin{itemize}
  \item Suppose $\ld(v)=\XYm$, this case is analogous to the corresponding one in Subcase~3a, so Condition~\ref{it:XY} is satisfied.
  \item Suppose $\ld(v)=\XZm$, then $v\in g(v)$ and $v\in\odd{G}{g(v)}$.
   Furthermore, $\ld'(v) = \XZm$.
   We have
    \begin{itemize}
     \item $v\in\odd{G\star u}{g(v)}$ since $v\in\odd{G}{g(v)}$ and $v\notin N_G(u)\cap g(v)$,
     \item $v\in g'(v)$,
    \end{itemize}
   which together give Condition~\ref{it:XZ}.
  \item Suppose $\ld(v)=\YZm$, then $v\in g(v)$ and $v\notin\odd{G}{g(v)}$.
   Furthermore, $\ld'(v) = \YZm$.
   We have
    \begin{itemize}
     \item $v\notin\odd{G\star u}{g(v)}$ since $v\notin\odd{G}{g(v)}$ and $v\notin N_G(u)\cap g(v)$,
     \item $v\in g'(v)$,
    \end{itemize}
   which together give Condition~\ref{it:YZ}. \qedhere
 \end{itemize}
\end{proof}

\begin{repcor}{cor:pivot_gflow}
 \statecorpivotgflow
\end{repcor}

\begin{proof}
Recall that $G\land uv=G\star u\star v\star u$. Hence by applying Lemma~\ref{lem:lc_gflow} three times we get the correct underlying graph, which has a gflow. It remains to show that the labels change in the manner stated.
Let us denote the labelling function associated with $\Gp$ by $\ld'$ (as in Lemma~\ref{lem:lc_gflow}). Similarly, we denote the labelling and correction set functions for $\Gpp$ by $\ld''$ and $g''$. We then compute:
 \[
  \ld''(v) = \begin{cases} \XZm &\text{if } \ld'(v)=\XYm \text{ iff } \ld(v)=\XYm \\
                           \XYm &\text{if } \ld'(v)=\XZm \text{ iff } \ld(v)=\YZm \\
                           \YZm &\text{if } \ld'(v)=\YZm \text{ iff } \ld(v)=\XZm, \end{cases}
 \]
so that
 \[
  \hat\ld(v) = \begin{cases} \XYm &\text{if } \ld''(v)=\XYm \text{ iff } \ld(v)=\YZm \\
                             \YZm &\text{if } \ld''(v)=\XZm \text{ iff } \ld(v)=\XYm \\
                             \XZm &\text{if } \ld''(v)=\YZm \text{ iff } \ld(v)=\XZm, \end{cases}
 \]
and by symmetry the same holds for $u$.
For $w\in\comp O\setminus\{u,v\}$ we have
 \[
  \ld''(w) = \begin{cases} \YZm &\text{if } w\in N_{\Gp}(v) \text{ and } \ld'(w)=\XZm \\
                           \XZm &\text{if } w\in N_{\Gp}(v) \text{ and } \ld'(w)=\YZm \\
                           \ld'(w) &\text{otherwise}. \end{cases}
 \]
Note that $w\in N_{\Gp}(v)$ iff $w\in N_G(v)\symd N_G(u)$, whence, denoting complementation in the vertex set by an overline, we obtain
  \begin{align*} \ld''(w) &= \begin{cases} \YZm &\text{if}\quad \begin{aligned} &w\in \comp{N_G(v)}\cap N_G(u), \ld(w)=\YZm, \text{ or} \\
                                                                            &w\in N_G(v)\cap\comp{N_G(u)}, \ld(w)=\XZm, \text{ or} \\
                                                                            &w\in N_G(v)\cap N_G(u), \ld(w)=\XZm \end{aligned} \\ \\
                                             \XZm &\text{if}\quad \begin{aligned} &w\in \comp{N_G(v)}\cap N_G(u), \ld(w)=\XZm, \text{ or} \\
                                                                            &w\in N_G(v)\cap\comp{N_G(u)}, \ld(w)=\YZm, \text{ or} \\
                                                                            &w\in N_G(v)\cap N_G(u), \ld(w)=\YZm \end{aligned} \\ \\
                                             \ld(w) &\text{otherwise}, \end{cases} \\
                            &= \begin{cases} \YZm &\text{if } w\in N_G(v) \text{ and } \ld(w)=\XZm \\
                                             \XZm &\text{if } w\in N_G(v) \text{ and } \ld(w)=\YZm \\
                                             \ld(w) &\text{otherwise}. \end{cases} \end{align*}
Similarly, we have
 \[
  \hat\ld(w) = \begin{cases} \YZm &\text{if } w\in N_{\Gpp}(u) \text{ and } \ld''(w)=\XZm \\
                             \XZm &\text{if } w\in N_{\Gpp}(u) \text{ and } \ld''(w)=\YZm \\
                             \ld''(w) &\text{otherwise}. \end{cases}
 \]
Noting that $w\in N_{\Gpp}(u)$ iff $w\in N_{\Gp}(u)\symd N_{\Gp}(v)$ iff $w\in N_G(v)$ we get
 \[
  \hat\ld(w) = \begin{cases} \YZm &\text{if } w\in N_{G}(v) \text{ and } \ld(w)=\YZm \\
                             \XZm &\text{if } w\in N_{G}(v) \text{ and } \ld(w)=\XZm \\
                             \ld(w) &\text{otherwise}, \end{cases}
 \]
which reduces to $\hat\ld(w)=\ld(w)$. \qedhere

\end{proof}

\section{Finding maximally delayed gflow in multiple measurement planes}\label{sec:FindingGflow}

The original algorithm for finding a maximally delayed gflow from Ref.~\cite{MP08-icalp}
is restricted to patterns where all measurements are in the \XY-plane.
Here we extend this procedure to measurements in the \YZ and \XZ-planes,
in the manner of Ref.~\cite{GFlow}.
We will always assume that any connected component of a labelled open graph contains an input or an output.
This is because any connected component containing neither inputs nor outputs contributes only an irrelevant scalar factor to a measurement-based computation.
In the presence of gflow, this assumption implies the following stronger property.

\begin{lemma}\label{lem:connected-output}
 Let $(G,I,O,\ld)$ be a labelled open graph with the property that any connected component of $G$ intersects $I\cup O$.
 Suppose furthermore that $(G,I,O,\ld)$ has gflow.
 Then any isolated vertex is an output.
\end{lemma}
\begin{proof}
 Let $(g,\prec)$ be a gflow on $(G,I,O,\ld)$, and let $v\in V$ be arbitrary.
 We will show that either
 \begin{itemize}
  \item $v\in O$, or
  \item $v\in\comp{O}$ and there exists $w\in V$ with $w \sim v$.
 \end{itemize}
 In the first case, there is nothing to prove.
 In the second case, consider the different measurement planes.
 \begin{itemize}
  \item Suppose $\ld(v)=\XYm$, then $v\notin g(v)$ and $v\in\odd{}{g(v)}$ by \ref{it:XY}.
   This implies there exists $w\in g(v)\cap N_G(v)$.
  \item Suppose $\ld(v)=\XZm$, then $v\in g(v)$ and $v\in\odd{}{g(v)}$ by \ref{it:XZ}.
   This again implies that there exists $w\in g(v)\cap N_G(v)$.
  \item Finally, suppose $\ld(v)=\YZm$, then $v\in g(v)$ and $v\notin\odd{}{g(v)}$ by \ref{it:YZ}.
   Note that $v\notin I$ since inputs cannot appear in correction sets.
   Also, $v$ is not an output, but its connected component must intersect $I\cup O$, so $v$ must have a neighbour $w$. \qedhere
 \end{itemize}
 \end{proof}

\begin{lemma}[Generalisation of {\cite[Lemma~1]{MP08-icalp}} to multiple measurement planes]
\label{lemV0precOutputsMaximal}
If $(g,\prec)$ is a maximally delayed gflow of $(G,I,O,\ld)$, then $V_0^\prec = O$.
\end{lemma}

\begin{proof}
To show that $V_0^\prec \sse O$:
Assume towards a contradiction that there exists $u\in V_0^\prec \setminus O$.
Then $u \in \max_\prec V$ by Definition~\ref{defVk}
and $g(u)$ is either $\set{u}$ or $\emptyset$
by \ref{it:g} of Definition \ref{defGFlow}.
\begin{itemize}
  \item If $g(u) = \emptyset$ then $u \notin \odd{}{g(u)}$ and
$u \notin g(u)$, contradicting all of \ref{it:XY}--\ref{it:YZ}, so the \LOG{} could not have gflow.
  \item Thus $g(u)=\{u\}$.
Now, since $u$ is not an output, Lemma~\ref{lem:connected-output} implies $u$ has a neighbour; call this neighbour $v$.
But $u\in g(u)$ and $v\sim u$ imply $v \in \odd{}{g(u)}$ and therefore $u \prec v$ by \ref{it:odd}.
\end{itemize}

This contradicts the property that $u\in \max_\prec V$.
Therefore such a $u$ cannot exist and instead $V_0^\prec \sse O$.

To show that $O \sse V_0^\prec$ we assume that $O \setminus V_0^\prec\neq\emptyset$.
Let $\prec'=\prec\setminus (O \times V)$,
which considers all outputs to be maximal.
Then $(g,\prec')$ is a gflow of $(G,I,O,\ld)$:
\begin{itemize}
  \item Conditions \ref{it:g} and \ref{it:odd} of Definition~\ref{defGFlow} are satisfied by $\prec'$ because
  elements of $O$ are not in the domain of $g$
\item Conditions \ref{it:XY}--\ref{it:YZ} are unaffected because the image of $g$ has not changed.
\end{itemize}

Thus, $(g,\prec')$ is a gflow of $(G,I,O,\ld)$.
Moreover, for any $k$, $\bigcup_{i=1}^k V_i^\prec \sse \bigcup_{i=1}^k V_i^{\prec'}$,
and $\abs{V_0^\prec} < \abs{V_0^{\prec'}}$, thus $(g,\prec')$ is more delayed than $(g,\prec)$,
which is a contradiction.
\end{proof}

\begin{lemma}[Generalisation of {\cite[Lemma~2]{MP08-icalp}} to multiple measurement planes]
\label{lemGFlowV0V1}
If $(g,\prec)$ is a maximally delayed gflow of $(G,I,O,\ld)$,
then $(\tilde{g},\tilde{\prec})$ is a maximally delayed gflow of $(G,I,O \cup V_1^\prec,\ld)$,
where $\tilde{g}$ is the restriction of $g$ to the domain $V\setminus (V_0^\prec\cup V_1^\prec)$ and $\tilde{\prec}=\prec\setminus (V_1^\prec\times V_0^\prec)$.
\end{lemma}

\begin{proof}
First show that $(\tilde{g},\tilde{\prec})$ is a gflow:
\begin{itemize}
  \item \ref{it:g} for $\tilde{\prec}$: If $v \in V \setminus (V_0^\prec\cup V_1^\prec)$ and $w\in \tilde{g}(v)$, then $w\in g(v)$ and thus $v=w$ or $v \prec w$
 (\ref{it:g} for $\prec$.)
 Therefore $v=w$ or $v \,\tilde{\prec}\, w$.
 \item \ref{it:odd} for $\tilde{\prec}$: If $v \in V \setminus (V_0^\prec\cup V_1^\prec)$ and  $w\in\odd{}{\tilde{g}(v)}$, then $w\in\odd{}{g(v)}$ and thus $v=w$ or $v\prec w$ (\ref{it:odd} for $\prec$.)
 Again $v=w$ or $v\,\tilde{\prec}\, w$.
 \item  \ref{it:XY}, \ref{it:XZ}, and \ref{it:YZ} for $\tilde{\prec}$: Since $\tilde{g}=g$ on $V\setminus (V_0^\prec\cup V_1^\prec)$, these conditions still hold for all $v\in V\setminus (V_0^\prec\cup V_1^\prec)$.
\end{itemize}

Now to show that $(\tilde{g},\tilde{\prec})$ is maximally delayed:
Assume for a contradiction that there exists a gflow $(g',\prec')$ which is more delayed than $(\tilde{g},\tilde{\prec})$,
i.e. that $\abs{\bigcup_{i=0}^k V_i^{\prec'}} \geq \abs{\bigcup_{i=0}^k V_i^{\tilde\prec}}$,
where the inequality is strict at some $k=j$.
Extend $(g', \prec')$ on $(G,I,O \cup V_1^\prec,\ld)$ to $(g'',\prec'')$ on $(G,I,O,\ld)$ where
\begin{align}
g''(v) &= \begin{cases} g'(v) &\text{if } v\in V\setminus (V_0^\prec\cup V_1^\prec) \\ g(v) &\text{if } v\in V_1^\prec \end{cases} \\
\prec'' &= \prec' \cup \{(u,v)\mid u\in V_1^\prec \wedge u\prec v\}
\end{align}
We can relate $(g',\prec')$ to $(g'', \prec'')$ and $(\tilde g, \tilde \prec)$ to $(g, \prec)$ via
\begin{align}
\bigcup_{i=0}^k V_i^{\prec'} &= \bigcup_{i=0}^{k+1} V_i^{\prec''} &
\bigcup_{i=0}^k V_i^{\tilde\prec} &= \bigcup_{i=0}^{k+1} V_i^{\prec} &
V_0^{\prec} & = V_0^{\prec''}
\end{align}

Therefore $\abs{\bigcup_{i=0}^{k} V_i^{\prec''}} \geq \abs{\bigcup_{i=0}^{k} V_i^{\prec}}$
and the inequality is strict at $k=j+1$.
This is a contradiction of maximality of $(g, \prec)$, so $(\tilde{g},\tilde{\prec})$ must be maximally delayed.
\end{proof}

\begin{lemma}[Generalisation of {\cite[Lemma~3]{MP08-icalp}} to multiple measurement planes]
\label{lemGFlowV1}
 If $(g,\prec)$ is a maximally delayed gflow of a labelled open graph $(G,I,O,\ld)$, then $V_1^\prec = V_1^{\prec,\XYm} \cup V_1^{\prec,\XZm} \cup V_1^{\prec,\YZm}$, where
 \begin{align*}
  V_1^{\prec,\XYm} &= \{u\in \comp{O} \mid \ld(u)=\XYm \wedge \exists K\sse O \cap \comp{I} \text{ s.t.\ } \odd{}{K}\cap \comp{O} = \{u\}\} \\
  V_1^{\prec,\XZm} &= \{u\in \comp{O} \mid \ld(u)=\XZm \wedge \exists K\sse O \cap \comp{I} \text{ s.t.\ } \odd{}{K\cup\{u\}}\cap \comp{O} = \{u\}\} \\
  V_1^{\prec,\YZm} &= \{u\in \comp{O} \mid \ld(u)=\YZm \wedge \exists K\sse O \cap \comp{I} \text{ s.t.\ } \odd{}{K\cup\{u\}}\cap \comp{O} = \emptyset\}.
 \end{align*}
\end{lemma}
\begin{proof}
We first show that if $(g,\prec)$ is a maximally delayed gflow, then for any $u\in V_1^\prec$, we have $g(u)\sse O\cup\{u\}$.
This is because if $v\in g(u)$ and $v\neq u$, then $u\prec v$ by property \ref{it:g}.
Furthermore, by the definition of $V_1^\prec$, if $u\prec v$ then $v\in V_0^\prec$, and therefore $v \in O$ by Lemma \ref{lemV0precOutputsMaximal}.
Combined with properties~\ref{it:odd} and \ref{it:XY}--\ref{it:YZ}, this shows that
$\odd{}{g(u)}\cap \comp{O} = \{u\}$ if $\ld(u)\in\{\XYm,\XZm\}$,
and $\odd{}{g(u)}\cap \comp{O} = \emptyset$ if $\ld(u)=\YZm$.
Therefore $V_1^\prec \sse V_1^{\prec,\XYm} \cup V_1^{\prec,\XZm} \cup V_1^{\prec,\YZm}$.

 To prove that any vertex satisfying one of the three conditions
 \begin{enumerate}
  \item $\ld(u)=\XYm$ and $\exists K\sse O \cap \comp{I} \text{ s.t.\ } \odd{}{K}\cap \comp{O} = \{u\}$
  \item $\ld(u)=\XZm$ and $\exists K\sse O \cap \comp{I} \text{ s.t.\ } \odd{}{K\cup\{u\}}\cap \comp{O} = \{u\}$
  \item $\ld(u)=\YZm$ and $\exists K\sse O \cap \comp{I} \text{ s.t.\ } \odd{}{K\cup\{u\}}\cap \comp{O} = \emptyset$
 \end{enumerate}
 must be in $V_1^\prec$, proceed by contradiction.
 As in the proof of Ref.~\cite[Lemma~3]{MP08-icalp}, we show that a gflow cannot be maximally delayed
 if there exists some vertex $w$ which satisfies one of conditions~1--3 but is not in $V_1^\prec$.

 Suppose $(g,\prec)$ is a maximally delayed gflow and there exists $w\in V\setminus V_0^\prec$
 such that $w$ satisfies condition~$j$, where $j\in\{1,2,3\}$, and let $K$ be the set identified in that condition.
 Define\footnote{The statement of Ref.~\cite[Lemma~3]{MP08-icalp} is missing the
 condition that the elements of $K$ are not inputs,
 however this error does not appear in that paper's pseudo-code for finding gflow.}
\begin{align}
  g'(u) := \begin{cases}
  K &\text{if } u=w \text{ and } \ld(w)=\XYm \\
  K\cup\{w\} &\text{if } u=w \text{ and } \ld(w)\in\{\XZm,\YZm\} \\
  g(u) &\text{otherwise.}
  \end{cases}
\end{align}
Let $\prec'$ be the partial order defined by\footnote{The proof of Ref.~\cite[Lemma~3]{MP08-icalp} is missing the third case of the definition of $\prec'$,
but it is actually required in that proof as well.}
\begin{align}
  u\prec' v \text{ if } \begin{cases}
  u\neq w \text{ and } u\prec v \\
  u=w \text{ and } v\in K \\
  u=w \text{ and } v\in \odd{}{g'(w)} \setminus \set{w}
  \end{cases}
\end{align}
Then property~\ref{it:g} of the gflow definition holds by the properties of $(g,\prec)$ and the second case of the definition of $\prec'$.
Property~\ref{it:odd} holds by the properties of $(g,\prec)$ and the third case of the definition of $\prec'$.
Properties~\ref{it:XY}--\ref{it:YZ} are satisfied for all $v\in V\setminus\{w\}$ since $g$ has not changed for those vertices.
For $w$ itself, only one of properties~\ref{it:XY}--\ref{it:YZ} is relevant, and it is satisfied by condition~$j$.

Now, assume towards a contradiction that $w\notin V_1^\prec$, i.e.\ assume that there exists $k>1$ such that $w\in V_k^\prec$.
Our construction of $\prec'$ in conjunction with condition $j$ implies that $w \in V_1^{\prec'}$.
For any $u\in V$ define the \emph{depth} of $u$, written $d(u)$, as the index $k$ such that $u \in V_k^\prec$,
and define $d'(u)$ as the index $k$ such that $u \in V_k^{\prec'}$.
Our assumption implies that $d'(w) < d(w)$,
and our construction of $\prec'$ implies that $d'(u) \leq d(u)$ for all vertices $u$.
This shows that:
\begin{align}
V_0^{\prec} & = V_0^{\prec'} &
\abs{V_1^{\prec}} & < \abs{V_1^{\prec'}} &
\abs{\bigcup_{i=0}^k V_i^{\prec}} &\leq \abs{\bigcup_{i=0}^{k} V_i^{\prec'}}
\end{align}

Therefore $(g',\prec')$ is more delayed than $(g,\prec)$, contradicting the assumption that $(g,\prec)$ is maximally delayed.
Hence any vertex $w$ satisfying one of conditions~1--3 must be in $V_1^\prec$,
i.e.\ $V_1^{\prec,\XYm} \cup V_1^{\prec,\XZm} \cup V_1^{\prec,\YZm} \sse V_1^\prec$.
\end{proof}

We combine Lemma~\ref{lemGFlowV1} with Lemma~\ref{lemV0precOutputsMaximal}
to construct a partial converse to Lemma~\ref{lemGFlowV0V1};
if we know that $(\tilde g, \tilde \prec)$ is maximally delayed
on $(G, I, O \cup V_1^\prec, \ld)$,
(with $V_1^\prec$ defined as in Lemma~\ref{lemGFlowV1})
we want to show that $(g, \prec)$ is maximally delayed on $(G, I, O, \ld)$.

\begin{lemma} \label{lemGFlowV0V1Back}
Suppose $(g, \prec)$ is a gflow of $(G, I, O, \lambda)$
such that $V_0^\prec = O$ and $V_1^\prec = V_1^{\prec,\XYm} \cup V_1^{\prec,\XZm} \cup V_1^{\prec,\YZm}$.
Now if $(\tilde g, \tilde \prec)$ defined as in Lemma~\ref{lemGFlowV0V1} is a maximally delayed gflow for $(G,I,O \cup V_1^\prec,\ld)$,
then $(g, \prec)$ is also maximally delayed.
\end{lemma}

\begin{proof} \label{prfLemGFlowV0V1Back}
Recall from Lemma~\ref{lemGFlowV0V1}
that $\tilde{g}$ is the restriction of $g$ to the domain
$V\setminus (V_0^\prec\cup V_1^\prec)$ and $\tilde{\prec}=\prec\setminus (V_1^\prec\times V_0^\prec)$.
Working towards a contradiction, suppose that $(g', \prec')$ is a gflow of $(G, I, O, \lambda)$ which is more delayed than
$(g, \prec)$.
By Lemmas~\ref{lemV0precOutputsMaximal} and \ref{lemGFlowV1}
we know that $V_0^\prec = V_0^{\prec'}$
and that $V_1^\prec = V_1^{\prec'}$.
Let $j$ be such that:
\begin{align}
\abs{V^\precprimed_j} &< \abs{V^{\prec}_j} &
\forall k \quad \abs{\bigcup_{i=0}^k V_i^{\precprimed}} &\leq \abs{\bigcup_{i=0}^{k} V_i^{\prec}}
\end{align}
Such a $j$ must exist, and be greater than 1, because $(g', \prec')$ is more delayed that $(g, \prec)$.
Define $(\tilde g', \tilde \prec')$ analogously to $(\tilde g, \tilde \prec)$:
\begin{align}
\tilde g' &= \text{ the restriction of $g'$ to } V\setminus (V_0^\precprimed\cup V_1^\precprimed) \\
\tilde{\prec}'&=\precprimed\setminus (V_1^\precprimed\times V_0^\precprimed)
\end{align}
By Lemma~\ref{lemGFlowV0V1} $(\tilde g', \tildeprecprimed)$ is maximally delayed
for $(G, I, O \cup V_1^\prec, \lambda)$,
and by our assumption so is $(\tilde g, \tilde \prec)$.
Our contradiction will be shown by noting that $(\tilde g', \tildeprecprimed)$
is more delayed than $(\tilde g, \tilde \prec)$:
\begin{align}
V_j^\prec &= V_{j-1}^{\tilde \prec} & V_j^\precprimed &= V_{j-1}^{\tildeprecprimed}
& \abs{V^\precprimed_j} &< \abs{V^{\prec}_j} &\implies \abs{V^\tildeprecprimed_{j-1}} &< \abs{V^{\tilde \prec}_{j-1}}
\end{align}

and
\begin{align}
\forall k \quad  \abs{\bigcup_{i=0}^k V_i^{\precprimed}} &\leq \abs{\bigcup_{i=0}^{k} V_i^{\prec}}
& \implies  \quad \forall k \quad  \abs{\bigcup_{i=0}^{k-1} V_i^{\tildeprecprimed}} &\leq \abs{\bigcup_{i=0}^{k-1} V_i^{\tilde \prec}}
\end{align}

This contradicts $(\tilde g, \tilde \prec)$ being maximally delayed,
and therefore $(g, \prec)$ must be maximally delayed.
\end{proof}

We now adapt the main gflow theorem and proof from {Ref.~\cite[Theorem~2]{MP08-icalp}}
to our multi-planar setting.
The idea is to iteratively construct the layers $V_k^\prec$,
finding candidate correction sets via matrix manipulation in $\mathbb F_2$.

\begin{theorem}[Generalisation of {\cite[Theorem~2]{MP08-icalp}} to multiple measurement planes]
There exists an algorithm that decides whether a given
\LOG\ has an (extended) gflow, and that outputs such a gflow if it exists.
Moreover, this output gflow is maximally delayed,
and the algorithm takes a number of steps that is polynomial in the number of vertices of the \LOG.
\end{theorem}
\begin{proof}
Let $(G,I,O,\ld)$ be a \LOG.
The algorithm \verb@GFLOW@$(M, I, O, \ld)$ (Algorithm \ref{algoGFlow}),
where $M$ is the adjacency matrix of $G$,
finds a maximally delayed gflow and returns true if one exists and returns false otherwise.

\begin{itemize}
  \item \textbf{Correctness of gflow:}
At the $k^{th}$ recursive call,
the set $C$ constructed by the algorithm in the \verb@for@ loop at line \ref{algoOperatorSolveForx} corresponds to the layer $V^\prec_k$
of the partition induced by the returned strict partial order (Definition \ref{defVk}).
The construction of $K'$ corresponds to
finding a correcting set in $\bigcup_{j \le k} V^\prec_j$
that has the desired properties for its odd neighbourhood
(as specified in conditions~\ref{it:odd}--\ref{it:YZ} of the gflow definition,
see Lemma~\ref{lemGFlowV1}.)
Furthermore the restriction $K' \sse O'$ ensures condition~\ref{it:g}, thus
if the algorithm returns a solution, then the solution satisfies the conditions of a gflow.

\item \textbf{Maximality of gflow:}
By induction on the number of non-output vertices, we prove that if
the given labelled open graph has a gflow, then the algorithm outputs a
maximally delayed gflow.
Assume that the given labelled open graph has a gflow.
\begin{itemize}
  \item The base case is when there is no non-output vertex,
and here there is no correction needed:
The empty flow $(g,\emptyset)$ (where $g$ is a function with an empty
domain) is a maximally delayed gflow.
\item For the inductive step suppose that there exist some non-output vertices;
according to Lemma~\ref{lemGFlowV1} the elements of $V^\prec_1$ are those that satisfy one of
the tests inside the loop at line \ref{algoOperatorSolveForx}.
Thus, after the loop (line \ref{algoGFlowAfterLoop}),
$C=V^\prec_1$.
Since the existence of a gflow is assumed,
$V^\prec_1$ cannot be empty (all non-output vertices have to be corrected),
thus the algorithm is called recursively
on the labelled open graph $(G, I, O \cup V^\prec_1, \lambda)$ which has fewer non-output vertices.
Lemma \ref{lemGFlowV0V1} ensures the existence of a gflow in
$(G,I,O \cup V^\prec_1, \ld)$,
and by induction we know that the recursive call returns
a maximally delayed gflow for $(G,I,O \cup V^\prec_1, \ld)$.
By Lemma \ref{lemGFlowV0V1Back}
we know that this construction for $V_0^\prec$ and $V_1^\prec$
yields a maximally delayed gflow for $(G,I,O, \ld)$.
\end{itemize}

\item \textbf{Complexity:}
First form the matrix $A$ from the adjacency matrix $M$
by removing those rows that correspond to output vertices,
and keeping all columns that correspond to outputs that are not also inputs
(i.e.\ keep candidate vertices for $K$).
Notice that finding a solution for $K$ inside the \verb@for@ loop at line \ref{algoOperatorSolveForx}
consists of solving a system $Ax=b_j$ in $\mathbb F_2$,
where $b_j$ is either a column vector with a 1 corresponding to the vertex $u$
(in the \XY and \XZ cases), or $b_j$ is the column vector of all 0s (in the \YZ case.)
The resulting column vector $x$ has a 1 in position $x_v$ if $v \in K$ and a 0 otherwise.

Set
$n=\abs{V}$,
$\ell=\abs{O}$,
$\ell' = \abs{O \setminus I}$, and
$j$ to range over the $n - \ell$ vertices in $\comp O$.
In order to solve these $n-\ell$ equations
the block matrix
\begin{align}
A'=\begin{bmatrix}
\quad A &\vert& b_1 &\vert& \ldots &\vert& b_{n-\ell} \quad
\end{bmatrix}
\end{align}
of dimensions $(n-\ell)\times (\ell' + n - \ell)$
is transformed into an upper
triangular form,
for instance using Gaussian elimination,
within $O(n^3)$ operations.
Then for each $b_j$ a back substitution within $O(n^2)$
operations is used to find $x_i$, if it exists, such that
$Ax_j=b_j$ (see Ref.~\cite{Bard06}). The back substitutions costs $O(n^3)$ operations at each
call of the function. Since there are at most $n$ recursive calls, the
overall complexity is $O(n^4)$.

\end{itemize}  \qedhere
\end{proof}

\newpage
\section{Pseudocode for algorithms}\label{sec:pseudocode}

\makeatletter
\algrenewcommand\ALG@beginalgorithmic{\small}
\makeatother

\renewcommand\algorithmiccomment[1]{\hfill{\color{gray!80!green}$\triangleright$\,\textit{#1}}}

\begin{algorithm}
\caption{Extended Generalized Flow \label{algoGFlow}}
\begin{algorithmic}[1]
\Procedure{GFLOW}{$M$, $I$, $O$, $\ld$}
\State initialise functions $g$ and $d$
\ForAll{$v\in O$}
  \State $d(v) \gets 0$ \Comment{Outputs have depth $0$}
\EndFor
\State \textbf{return} \Call{GFLOWAUX}{$M$,  $I$, $O$, $\ld$, $1$, $d$, $g$} \Comment{Start the recursive process of finding $V_j^\prec$}
\EndProcedure

\Statex

\Procedure{GFLOWAUX}{$M$, $I$, $O$, $\ld$, $k$, $d$, $g$}
  \State $O' \gets O \setminus I$
  \State $C \gets \emptyset$
  \ForAll{$u \in \comp{O}$} \label{algoOperatorSolveForx}
    \If{$\ld (u) =$ \XY}
      \State $K' \gets$ Solution for $K \sse O'$ where $\odd{}{K} \cap \comp{O} = \set{u}$
    \ElsIf{$\ld (u) =$ \XZ}
      \State $K' \gets \; \set{u} \; \cup$ (Solution for $K \sse O'$ where $\odd{}{K \cup{\set{u}}} \cap \comp{O} = \set{u}$)
    \Else{ $\ld (u) =$ \YZ}
      \State $K' \gets \; \set{u} \; \cup$ (Solution for $K \sse O'$ where $\odd{}{K \cup{\set{u}}} \cap \comp{O} = \emptyset$)
    \EndIf
    \If{$K'$ exists}
      \State $C \gets C \cup \set{u}$ \label{algoOperatorUpdateC}
      \State $g(u) \gets K'$ \Comment{Assign a correction set for $u$}
      \State $d(u) \gets k$ \Comment{Assign a depth value for $u$}
    \EndIf
  \EndFor
  \If{$C = \emptyset$} \label{algoGFlowAfterLoop}
    \If{$O = V$}
      \State \textbf{return} (true, $g$, $d$) \Comment{Halt, returning maximally delayed $g$ and depth values $d$}
    \Else
      \State \textbf{return} (false, $\emptyset$, $\emptyset$) \Comment{Halt if no gflow exists}
    \EndIf
  \Else
    \State \textbf{return} \Call{GFLOWAUX}{$M$, $I$, $O \cup C$, $\ld$, $k+1$, $d$, $g$}
  \EndIf
\EndProcedure
\end{algorithmic}
\end{algorithm}

\begin{algorithm}[!htb]
\caption{Circuit Extraction}\label{alg:extraction}
\begin{algorithmic}[1]
\Procedure{Extract}{$D$}\Comment{input is MBQC diagram $D$}
  \State Init empty circuit $C$
  \State $G,I,O\gets $ Graph$(D)$\Comment{get the underlying graph of $D$}
  \State $F \gets O$ \Comment{initialise the frontier}

  \State $D,C \gets$ ProcessFrontier($D,F,C$)

  \While{$\exists v\in D\setminus F$}\Comment{there are still vertices to be processed}
    \State $D,F,C \gets $ UpdateFrontier$(D,F,C)$
  \EndWhile
  \For{$v\in F$} \Comment{the only vertices still in $D$ are in $F$}
    \If{$v$ connected to input via Hadamard}
      \State $C\gets$ Hadamard(Qubit($v$))
    \EndIf
  \EndFor
  \State Perm $\gets$ Permutation of Qubits of $F$ to qubits of inputs \Comment{step 5 of Section~\ref{sec:generalextractalgorithm}}
  \For{swap$(q_1,q_2)$ in PermutationAsSwaps(Perm)}
    \State $C\gets$ swap$(q_1,q_2)$
  \EndFor
  \State \textbf{return} $C$
\EndProcedure

\Procedure{ProcessFrontier}{$D,F,C$} \Comment{Corresponds to step 1 of Section~\ref{sec:generalextractalgorithm}}
    \For{$v\in F$}
        \If{$v$ has local Cliffords}
          \State $C\gets $ Cliffords(Qubit($v$))
          \State Remove Cliffords from $v$ to output
        \EndIf
      \EndFor
      \For{edge between $v$ and $w$ in $F$}
        \State $C\gets $ CZ(Qubit($v$), Qubit($w$))
        \State Remove edge between $v$ and $w$
  \EndFor
  \State \textbf{return} $D,C$
\EndProcedure

\Procedure{UpdateFrontier}{$D,F,C$} \Comment{Corresponds to steps 3,4 of Section~\ref{sec:generalextractalgorithm}}
  \State $N\gets $ Neighbours$(F)$
  \State $M \gets$ Biadjacency$(F,N)$
  \State $M^\prime \gets $ GaussReduce$(M)$
  \State Init $vs$ \Comment{initialise empty set $vs$}
  \For{row $r$ in $M^\prime$}
    \If{sum$(r) == 1$}\Comment{there is a single 1 on row $r$}
      \State Set $v$ to vertex corresponding to non-zero column of $r$
      \State Add $v$ to $vs$ \Comment{$v$ will be part of the new frontier}
    \EndIf
  \EndFor
  \If{$vs$ is empty}
    \While{There is \YZ vertex connected to $F$}
        \State $v \gets$ \YZ vertex connected to $F$
        \State $w \gets$ a neighbour of $w$ in $F$
        \State $D\gets $ Pivot($D$,$v$,$w$)
    \EndWhile
    \State $D,C \gets$ ProcessFrontier$(D,F,C)$
    \State \textbf{return} $D,F,C$
  \EndIf
  \State $M \gets $  Biadjacency($F,ws$) \Comment{smaller biadjacency matrix}
  \For{$(r_1,r_2) \in $ GaussRowOperations$(M)$}
    \State $C\gets $ CNOT(Qubit$(r_1)$, Qubit$(r_2)$)
    \State Update $D$ based on row operation
  \EndFor
  \For{$v\in vs$} \Comment{all $v$ now have a unique neighbour in $F$}
    \State $w \gets $ Unique neighbour of $v$ in $F$
    \State $C\gets $ Hadamard(Qubit$(w)$)
    \State $C\gets $ Phase-gate$(Phase(v),Qubit($w$))$
    \State Remove $w$ from $F$
    \State Add $v$ to $F$
  \EndFor
  \State $D,C \gets$ ProcessFrontier$(D,F,C)$
  \State \textbf{return} $D,F,C$
\EndProcedure
\end{algorithmic}

\end{algorithm}

\end{document}